\documentclass[11pt]{article}

\usepackage{preamble}
\newcommand{\ucsc}{University of California, Santa Cruz}

\title{Classical Algorithms for Bipartite Quantum Max-Cut on Dense Expanders}
\author{%
    \setcounter{footnote}{1}%
    Zackary Jorquera\thanks{\ucsc}
    \and
    Alexandra Kolla\footnotemark[2]
    \and
    Stuart Wayland\footnotemark[2]
}
\date{\today}

\begin{document}

\pagenumbering{roman}

\maketitle

\begin{abstract}
Quantum Max-Cut is a well-studied local Hamiltonian problem for which the classical complexity of the bipartite 
case remains open. We give a polynomial-time classical randomized algorithm for the spin-$\tfrac12$ 
antiferromagnetic Heisenberg model, equivalently Quantum Max-Cut, on dense balanced bipartite expanders.
This class includes $G(n,n,p)$ with high probability for every fixed $p>0$. The algorithm
estimates the ground energy and any ground state edge correlation to additive error $\varepsilon$ in time polynomial
in $n$ and $1/\varepsilon$. We construct a Markov chain on perfect matchings of the complete bipartite graph, each
matching representing a product of singlets, that converges to the ground state, and use the expansion of the graph to show that polynomially many steps
and samples of this chain suffice.
\end{abstract}

\begingroup
\hypersetup{linktocpage}
\tableofcontents
\endgroup

\clearpage
\pagenumbering{arabic}

\section{Introduction}\label{section_intro}

The Local Hamiltonian problem is the quantum analogue of classical constraint satisfaction, with local interaction 
terms replacing local constraints. While the problem is \textsf{QMA}-complete in general, a central question is to 
identify natural classes of interactions for which the ground energy can be computed efficiently. 
A particularly well-studied example is the spin-$\tfrac12$ antiferromagnetic Heisenberg model \cite{heisenberg_zur_1928}, also known as 
Quantum Max-Cut~\cite{GharibianParekh2019}. Each edge favors a singlet on its two endpoints, providing a quantum analogue 
of the preference for differing endpoints in classical Max-Cut. On general weighted graphs, computing
the ground energy is \textsf{QMA}-complete~\cite{PiddockMontanaro2017}, although the known hardness constructions use polynomially large
interaction weights and no corresponding hardness result is known for the unweighted problem~\cite{Gharibian2023}.

The contrast with Max-Cut is especially striking on bipartite graphs. Classical Max-Cut is trivial in this setting. Recent work using Lee-Yang zero-freeness places
the ground-energy problem for Quantum Max-Cut on every bipartite graph in
\textsf{BQP}~\cite{BravyiGossetLiuWong2026,RayuduTakahashi2026}. Whether the problem also lies in
\textsf{BPP}, that is, whether it admits a polynomial-time classical randomized algorithm, remains
open~\cite{Gharibian2023,TakahashiSlezakCrosson2024}.

There is good reason to expect bipartiteness to help in designing classical algorithms. The antiferromagnetic
Heisenberg Hamiltonian on a bipartite graph is known to be in \textsf{StoqMA}~\cite{Marshall1955,BDOT2008,king_improved_2023}, which enables sign-free quantum Monte Carlo methods widely used to estimate ground
energies and correlations~\cite{Sandvik1999,Sandvik2005}. Sign-freeness, however, does not by itself imply polynomial running time. A 
rigorous algorithmic guarantee must also control how quickly the ground state is reached and how many samples are required to estimate its 
observables. Establishing such guarantees for bipartite Heisenberg models has remained difficult, with polynomial-time analyses available 
only for restricted graph families~\cite{TakahashiSlezakCrosson2024}.

In this work, we give a classical randomized algorithm for dense bipartite expanders. For constant edge density $p$,
the algorithm estimates the ground energy and any ground state edge correlation to additive error $\varepsilon$ in time polynomial 
in $n$ and $1/\varepsilon$ (\cref{thm_sampler}). It samples a reference walk on perfect matchings, each representing a product of singlet pairs, and reweights its paths. We prove that the resulting estimator has polynomial sample complexity over the
required projection depth. The spectral input is a gap of order $n$ within the singlet sector and constant overlap with the Lieb-Mattis starting state
(\cref{thm_spectral}), which we establish by comparing low-energy singlets with a solvable quadratic model of paired excitations.

\subsection{Setting}\label{section_intro_setting}

Throughout, let $G=(L\sqcup R,E)$ be a connected balanced bipartite graph with
$|L|=|R|=n$ and edge weights $0<A_{ij}\le 1$, where $A$ is its weighted biadjacency
matrix ($A_{ij}=0$ for $ij\notin E$). We associate one qubit with each
vertex and write $\bS_v=(S_v^x,S_v^y,S_v^z)$ for the spin-$\tfrac12$ operators acting on
vertex $v$ (see \cref{eq_pauli_spin_ops}). The \emph{Heisenberg Hamiltonian} is
\begin{equation}\label{eq_heisenberg_intro}
  H_G\coloneq\sum_{ij\in E}A_{ij}\,\bS_i\!\cdot\bS_j,
  \qquad
  \bS_i\!\cdot\bS_j\coloneq S_i^xS_j^x+S_i^yS_j^y+S_i^zS_j^z
\end{equation}
Since $\frac14I-\bS_i\!\cdot\bS_j$ is the singlet projector on the edge $ij$, minimizing
$H_G$ is the weighted Quantum Max-Cut problem up to an additive constant and a change of
sign.

Our goal is to estimate the ground energy and local correlations of $H_G$. By the
Lieb-Mattis theorem~\cite{LiebMattis1962}, the ground state is unique up to phase and has total spin zero. Total spin zero states form the \emph{singlet sector} $\cS\coloneq\ker\bS_{\mathrm{tot}}^{2}$, where
$\bS_{\mathrm{tot}}\coloneq\sum_{v\in L\sqcup R}\bS_v$. We write $\Psi_0\in\cS$ for the unit ground
state, $\lambda_0<\lambda_1\le\cdots$ for the eigenvalues of $H_G|_{\cS}$, and
$\gamma\coloneq\lambda_1-\lambda_0$ for the \emph{singlet-sector gap}. This gap is defined within $\cS$; it need not be the gap above the ground state in the full Hilbert space. For an edge $ij\in E$ with
$i\in L$ and $j\in R$, the ground state correlation is
\[
  \ev{\bS_i\!\cdot\bS_j}\coloneq\ip{\Psi_0}{\bS_i\!\cdot\bS_j\,\Psi_0}\in[-\tfrac34,\tfrac14]
\]
and the ground energy is $\lambda_0=\sum_{ij\in E}A_{ij}\ev{\bS_i\!\cdot\bS_j}$.
The \emph{Lieb-Mattis state} $\LM$ (\cref{def_LM}) is the known ground state of the
Heisenberg model on the complete bipartite graph, and will serve as the starting point of our algorithm.

The graph family is described by
the edge density $\hatp\coloneq\frac1{n^2}\sum_{ij\in E}A_{ij}$, the largest degree
deviation $\Delta\coloneq\max\{\max_i|\rho_i|,\ \max_j|\rho'_j|\}$, where $\rho_i\coloneq\sum_{j}A_{ij}-\hatp n$ and $\rho'_j\coloneq\sum_{i}A_{ij}-\hatp n$ are the weighted degree deviations, and the largest singular value $\sigma_1\coloneq\norm{\xi_c}_{\opn}$
of the doubly centered matrix
$\xi_c\coloneq\xi-\frac1n\rho\one^{\top}-\frac1n\one\rho'^{\top}$, where
$\xi\coloneq A-\hatp\,\one\one^{\top}$ and $A$ is the weighted biadjacency matrix above.

\begin{definition}[Dense bipartite expanders]\label{def_class}
For $p\in(0,1)$ and $C>0$, let $\mathcal G(p,C)$ be the class of connected balanced
bipartite graphs $G$ on $n+n$ vertices satisfying
\[
\hatp\ge p,
\qquad
\sigma_1\le C\sqrt{n\log n},
\qquad
\Delta\le C\sqrt{n\log n}
\]
\end{definition}

The definition requires constant edge density together with square-root-scale spectral
and degree fluctuations, up to logarithmic factors. In particular,
$G(n,n,p)$ belongs to $\mathcal G(p',C)$ with high probability for every fixed
$p>p'>0$ and a suitable constant $C$.

\subsection{Main results}\label{section_intro_results}

Our main result is a polynomial-time classical algorithm for the ground state edge
correlations of dense bipartite expanders.

\begin{maintheorem}[Polynomial-Time Ground State Estimation]\label{thm_sampler}
Let $G\in\mathcal G(p,C)$ be a dense bipartite expander. There exists a classical randomized algorithm that, given an edge
$ij\in E$ and $\varepsilon,\delta\in(0,1)$, approximates $\ev{\bS_i\!\cdot\bS_j}$ within
$\varepsilon$ with probability at least $1-\delta$ in time
\[
\mathrm{poly}(n)\cdot(n/\varepsilon)^{O(1/p)}\cdot\bigl(1+\log(1/\delta)\bigr)
\]
\end{maintheorem}

The proof of \cref{thm_sampler} has three components, namely the algorithm and the two inputs it
requires. The algorithm, described in \cref{section_algorithm}, runs in polynomial time given two inputs, as established in \cref{thm_reduction}. The first is
a lower bound on the singlet-sector gap $\gamma$, which sets the projection depth needed to
separate the ground state from the rest of the singlet spectrum. The second is a lower
bound on the overlap $|\ip{\LM}{\Psi_0}|$, which sets how much ground state weight is
present at the start of the projection. Our second main result establishes both inputs on
$\mathcal G(p,C)$.

\begin{maintheorem}[Singlet-Sector Gap and Starting-State Overlap]\label{thm_spectral}
Let $G\in\mathcal G(p,C)$. Then
the singlet-sector gap of the Hamiltonian $H_G$ and the overlap of the Lieb-Mattis state with the ground state of $H_G$ satisfy
\[
\gamma\ \ge\ \frac{p\,n}{4},
\qquad
|\ip{\LM}{\Psi_0}|\ \ge\ \tfrac12e^{-(4/p+1)/2}
\]
\end{maintheorem}

Both bounds are established in \cref{section_boson}, using the decomposition of the
Hamiltonian described in \cref{sec:rep-theory}.

\subsection{Proof overview}\label{section_intro_overview}

The algorithm approximates the ground state through the projection method. For this to be efficient, we need
a starting state that is easy to sample and has sufficient overlap with the ground
state, together with a spectral gap of our Hamiltonian within the subspace known to contain the ground state.

We implement this projection using perfect matchings of the complete bipartite
graph. Our starting state is the uniform superposition over these matchings,
so sampling it amounts to drawing a uniformly random matching, and corresponds directly to the ground state of the complete bipartite case. 
Each edge interaction acts by a transposition on either the left or right sublattice, connecting two matchings that differ by a four-cycle and turning repeated application of the Hamiltonian into a sum over paths of matchings. Adjusting the probability of staying at the
current matching gives a simple random walk from which we can sample these
paths directly. We then weight each sampled path to undo the adjustment. The
correction depends only on how many bonds in the current matching are edges of
our graph. Bounding the fluctuations of these weights ensures that only
polynomially many samples are needed.

To bound the gap and the overlap, we compare our graph with the complete bipartite
graph. The complete model organizes states into levels of increasing energy, with our
starting state at level zero. These levels can roughly be understood as excitations in the left and right sublattices similar to those considered in \cite{apte_conjectured_2025,bakshi_sharp_2026,basu_quantum_2026}. At low levels, we capture the contribution to the ground state with an exactly solvable reference model,
using expansion to show it is a good approximation. At higher levels, the energy cost is
too large for these differences to bring states close to the ground energy.
Together, these bounds give a gap of order $n$. The approximation also provides a good
trial state that is used to show overlap of our starting state with the ground state.

\subsection{Related work}\label{section_intro_related}

A general line of work asks when the ground energy of a local Hamiltonian can be
approximated efficiently by classical states or convex relaxations.  High-degree and
dense Hamiltonians admit product-state approximations to their energy
density~\cite{BrandaoHarrow2016}, and for Quantum Max-Cut this perspective leads to
efficient constant-factor approximation algorithms~\cite{GharibianParekh2019,hwang_unique_2023}.
Semidefinite programming hierarchies and rounding algorithms give further approximations of
instances~\cite{ParekhThompson2021,ParekhThompson2022,TakahashiRayuduZhouKingThompsonParekh2023}.
These results give multiplicative approximations to the ground energy of Quantum Max-Cut. However, in the
regime considered here, the product state with all $L$ spins up and all $R$ spins down already gives a multiplicative approximation to the order-\(n^2\) ground energy. We call this state the N\'eel state; the remaining error is of order $n$. Our goal is different.
\cref{thm_sampler} computes the full ground energy to arbitrary additive
accuracy and estimates individual ground state correlations.

Even gapped sign-free Hamiltonians can exhibit slowly mixing Monte Carlo
dynamics~\cite{HastingsFreedman2013,GilyenHastingsVazirani2021}.
Rigorous polynomial-time guarantees are known for several restricted classes,
including one-dimensional stoquastic systems~\cite{CrossonHarrow2021} and stoquastic
XY models~\cite{BravyiGosset2017}.  Most directly related to the present work,
Takahashi, Slezak, and Crosson~\cite{TakahashiSlezakCrosson2024} prove rapid mixing of a
loop-representation projector Monte Carlo algorithm for Heisenberg antiferromagnets on
star-like bipartite graphs.  More generally, a gapped Hamiltonian gives
rise to a rapidly mixing chain when the required ratios of ground state amplitudes can
be computed efficiently~\cite{BravyiCarleoGossetLiu2023}.  Our analysis requires no
mixing-time bound, since the reference walk is sampled exactly, and the central quantity is
instead the relative second moment of the accumulated weights.

The \emph{quantum} algorithms for bipartite Heisenberg models introduce a magnetic field
and prove spectral gap relative to this field~\cite{BravyiGossetLiuWong2026,RayuduTakahashi2026}.  As the resulting gap is controlled
by the field strength, it vanishes when the field is removed. This is sufficient for an adiabatic algorithm, but not for classical simulation. Here we work directly
with the zero-field Hamiltonian $H_G$, where \cref{thm_spectral} gives a linear
gap within the singlet sector.

The valence-bond basis is a standard representation for bipartite quantum
antiferromagnets, and projector Monte Carlo methods in this basis have been used
extensively to compute ground state energies and correlation
functions~\cite{LiangDoucotAnderson1988,Sutherland1988,Sandvik1999,Sandvik2005,BeachSandvik2006}.  These methods exploit the positivity and simple reconnection
rules of the Heisenberg interaction in the valence-bond representation.  Our
algorithm uses the same structure, with the diagonal correction and importance weights
described above replacing the usual valence-bond Monte Carlo dynamics.

Our spectral argument is related to the spin-wave approximation, in which fluctuations
around an ordered state are represented by bosons and the quadratic part is
diagonalized by a Bogoliubov
transformation~\cite{HolsteinPrimakoff1940,Bogoliubov1947,Anderson1952}.
Rigorous versions of spin-wave theory have been established for Heisenberg
ferromagnets, notably through quantitative control of the Holstein-Primakoff
representation~\cite{CorreggiGiulianiSeiringer2015}.  The setting here is different, as we require
finite-size control of the bottom of the \emph{singlet} spectrum of an
inhomogeneous antiferromagnet.  The complete bipartite model is exactly solvable by
total-spin decomposition~\cite{LiebMattis1962}, while additional symmetry yields
exact energies for some other graph families~\cite{WattsChowdhuryEpperlyHeltonKlep2024}.  Our graphs need not possess such
symmetry, and we use the Lieb-Mattis decomposition only as a reference problem.

\section{Preliminaries}\label{section_prelim}

\subsection{Representation theory}\label{section_rep_thy_prelims}

In this manuscript we will need a few results from the representation theory of the symmetric group and the \(\fraksu(2)\) Lie algebra. We assume familiarity with the basic representation theory of these objects and provide this section to recall the results and conventions needed throughout the manuscript. For a more in-depth introduction to representation theory, we recommend~\cite{steinberg_representation_2012,sagan_symmetric_2001,hall_lie_2015}.

Let \(\Sn\) be the symmetric group acting on \([n] \coloneq \{1,2,\dotsc,n\}\). We consider the representation \(\pi_k\colon \Sn \to \GL(\C^{\binom{[n]}{k}})\). We identify \(\cY_k \coloneq \C^{\binom{[n]}{k}}\) with the space of set functions \(f\colon \binom{[n]}{k} \to \C\), and define \((\pi_k(\sigma)f)(S) = f(\sigma^{-1}(S))\), for \(S \in \binom{[n]}{k}\). Viewed as an \(\Sn\)-module, \(\cY_k\) is the Young permutation module associated with the partition \((n-k,k)\), assuming \(k \leq n/2\). The complement map \(S \mapsto [n]\setminus S\) gives an \(\Sn\)-equivariant isomorphism \(\cY_k \cong \cY_{n-k}\). For \(k \leq n/2\), Young's rule~\cite{sagan_symmetric_2001,grinberg_introduction_2025} gives the decomposition
\begin{align}
\cY_k \cong \bigoplus_{\ell = 0}^k \operatorname{Sp}_{(n-\ell,\ell)}\label{eq_youngs_rule}
\end{align}
where \(\operatorname{Sp}_{(n-\ell,\ell)}\) denotes the Specht module, or irreducible \(\Sn\)-representation (irrep), corresponding to the partition \((n-\ell,\ell)\).

Next, we define the Lie algebra \(\mathfrak{su}(2)\) via the generators\footnote{We use the physics convention in which the generators are represented by Hermitian matrices rather than skew-Hermitian matrices.} \(J^x, J^y, J^z\). They satisfy \([J^a,J^b] = i \epsilon^{abc} J^c\). Often it is useful to pass to the complexified Lie algebra, which is isomorphic to \(\mathfrak{sl}(2,\mathbb{C})\). Here it is useful to define the complexified algebra elements \(J^\pm \coloneq J^x \pm i J^y\), known as the raising and lowering operators, which satisfy \([J^+, J^-] = 2J^z\) and \([J^z, J^{\pm}] = \pm J^{\pm}\). We typically use the notation \(\bJ = [J^x, J^y, J^z]^\top\) to denote the three-vector of these operators. We can then define the Casimir operator as \(\bJ^2 = \bJ \cdot \bJ = \left(J^x\right)^2 + \left(J^y\right)^2 + \left(J^z\right)^2\). The significance of this operator is twofold. First, it is central in the universal enveloping algebra, meaning that by Schur's lemma its action is by a scalar multiple of the identity under any irrep. Second, for finite-dimensional irreps, this scalar uniquely determines the irrep (up to normalization of the generators). We identify the irreps of \(\mathfrak{su}(2)\) by a nonnegative half-integer ``spin'' \(s\). In particular, for \(\rho_s\colon \mathfrak{su}(2) \to \mathfrak{gl}(V_s)\) being the irrep identified by spin \(s\), we have \(\dim(V_s) = 2s + 1\) and \(\rho_s(\bJ^2) = \rho_s(\bJ)^2 = s(s+1) I\). Throughout this manuscript, we will use the notation \(\bS = \rho_{1/2}(\bJ)\) for the spin-\(1/2\) spin operators. For reference, they are the following.
\begin{equation}\label{eq_pauli_spin_ops}
\begin{split}
S^x = \frac{1}{2}\begin{bmatrix}
0 & 1 \\ 1 & 0
\end{bmatrix}, \ \ \ \ S^y = \frac{1}{2}&\begin{bmatrix}
0 & -i \\ i & 0
\end{bmatrix}, \ \ \ \ S^z = \frac{1}{2}\begin{bmatrix}
1 & 0 \\ 0 & -1
\end{bmatrix}\\
S^+ = \begin{bmatrix}
0 & 1 \\ 0 & 0
\end{bmatrix}&, \ \ \ \ S^- = \begin{bmatrix}
0 & 0 \\ 1 & 0
\end{bmatrix}
\end{split}
\end{equation}

We now consider the setting of multiple systems. We write \(S_i^a \coloneq I^{\otimes(i-1)} \otimes S^a \otimes I^{\otimes(2n-i)}\) to denote the spin operator \(S^a\) applied to just the \(i\)th system and identity everywhere else. Furthermore, we similarly define \(\bS_i = [S^x_i, S^y_i, S^z_i]^\top\). For two systems, \(i\) and \(j\), we might define the diagonal representation, which maps \(J^a \mapsto S^a_i + S^a_j\). Recall that this is the Lie algebra tensor product representation. Often, the operator \(\bS_{\mathrm{tot}} = \bS_i + \bS_j\) is called the total spin for systems \(i\) and \(j\), and \(\bS_{\mathrm{tot}}^2\) the total spin Casimir. Because \(\bS_{\mathrm{tot}}\) acts on a tensor product representation, its action is generally reducible. In general, for two spins \(s_1,s_2\) and their corresponding Lie algebra representations, the tensor product representation \(\rho_{s_1} \otimes \rho_{s_2}\colon J^a \mapsto \rho_{s_1}(J^a) \otimes I + I \otimes \rho_{s_2}(J^a)\) decomposes into irreps of \(\mathfrak{su}(2)\) according to the Clebsch-Gordan decomposition~\cite{hall_lie_2015,sakurai_modern_2020}:
\begin{align}
V_{s_1} \otimes V_{s_2} \cong \bigoplus_{\ell = |s_1-s_2|}^{s_1 + s_2} V_\ell
\end{align}
where \(\ell\) increases in integer steps.

Finally, we consider multiple spin-\(1/2\) systems, or qubits, labeled by \(V = L \sqcup R = [2n]\). The space of these qubits is exactly \((\C^2)^{\otimes 2n}\). Here, there are three types of representations of \(\fraksu(2)\) that we will consider. First, the representation onto a single qubit, with spin operator \(\bS_i\). Second, as our model is bipartite, we will consider the representation onto the whole left (resp. right) sublattice. Here the spin operator is denoted \(\bS_L = \sum_{i \in L} \bS_i\) (resp. \(\bS_R = \sum_{i \in R} \bS_i\)). Finally, the representation onto all qubits, which has total spin operator given by \(\bS_{\mathrm{tot}} \coloneq \sum_{i \in V} \bS_i = \bS_L + \bS_R\).

By recursive Clebsch-Gordan decomposition, its action decomposes into irreps of spins from \(0\) to \(n\) in integer steps, with multiplicities. Instead of calculating these multiplicities directly from Clebsch-Gordan, we will appeal to another decomposition which follows from Schur-Weyl duality~\cite{goodman_symmetry_2009}. In short, Schur-Weyl duality tells us that our space of \(2n\) qubits decomposes into spin sectors and Specht modules:
\begin{align}
(\mathbb{C}^2)^{\otimes 2n} \cong \bigoplus_{s = 0}^n V_s \otimes \operatorname{Sp}_{(n+s,n-s)}\label{eq_schur_weyl_dual_decomp}
\end{align}
Ignoring the nature of the Specht modules, this gives the multiplicities of the spin irreps via the hook-length formula~\cite{sagan_symmetric_2001}.

We denote the singlet sector by \(\cS = V_0 \otimes \operatorname{Sp}_{(n,n)} \subseteq (\mathbb{C}^2)^{\otimes 2n}\). From the \(\fraksu(2)\) perspective, it is the isotypical subspace corresponding to \(\frac{1}{n+1}\binom{2n}{n}\) copies of the spin-\(0\) irrep. Equivalently, this is exactly specified by the kernel of the total-spin Casimir, \(\cS = \ker(\bS_{\mathrm{tot}}^2)\). From the \(\mathrm{Sym}(2n)\) perspective, because \(V_0\) is one-dimensional, this space is precisely the multiplicity-one irrep \(\operatorname{Sp}_{(n,n)}\). The significance of this space is that the ground state of our Hamiltonian lies in this space. This is a celebrated result of Lieb and Mattis~\cite{LiebMattis1962}.

We briefly recall some additional facts about irreducible representations of \(\fraksu(2)\), or rather of \(\fraksl(2,\C)\). For a more thorough introduction we recommend~\cite{hall_lie_2015,sakurai_modern_2020}. Let \(\rho_s\colon \fraksu(2) \to \frakgl(V_s)\) be a spin-\(s\) irrep. The matrix \(\rho_s(J^z)\) has spectrum \(\{-s,-s+1,\dotsc,s\}\), each with multiplicity one. There is an eigenbasis \(\{\ket{s,q}\}_q\) such that \(\rho_s(J^z)\ket{s,q} = q \ket{s,q}\) and
\begin{align}
    \rho_s(J^+)\ket{s,q}
    =
    \sqrt{(s-q)(s+q+1)}\ket{s,q+1},\quad \rho_s(J^-)\ket{s,q}
    =
    \sqrt{(s+q)(s-q+1)}\ket{s,q-1} \label{eq_weight_vec_raise_lower}
\end{align}
for \(-s \leq q \leq s\); in particular \(\rho_s(J^+)\ket{s,s} = 0\) and \(\rho_s(J^-)\ket{s,-s} = 0\). We call such vectors weight vectors, \(\ket{s,s}\) a highest-weight vector, and \(\ket{s,-s}\) a lowest-weight vector. One can identify the space \(V_s\) from a highest-weight vector, since all other weight vectors can be reached by repeated application of the lowering operator. More precisely,
\[
    \rho_s(J^-)^{s-q}\ket{s,s}
    =
    \sqrt{\frac{(2s)!(s-q)!}{(s+q)!}}\ket{s,q}
\]

\subsection{Valence-bond states (singlet matchings)}\label{section_dimer}

A \emph{matching} is a bijection $\pi\colon L\to R$.
After identifying both sides with $\{1,\dots,n\}$, we regard matchings as permutations
$\pi\in\Sn$. To each matching associate the \emph{valence-bond state}
\[
  \ket{\pi}
  \coloneq2^{-n/2}\bigotimes_{i\in L}
  \bigl(
  \ket{0_i1_{\pi(i)}}
  -
  \ket{1_i0_{\pi(i)}}
  \bigr)
\]
the tensor product of singlets across the bonds $\{i,\pi(i)\}$, all oriented from
$L$ to $R$.

Every $\ket\pi$ is a unit singlet, and the family
$\{\ket\pi:\pi\in\Sn\}$ spans $\cS$. This is standard in the physics of valence-bond states~\cite{Sutherland1988,BeachSandvik2006}. In representation-theoretic terms, $\ket\pi$ is the image of a polytabloid of shape $(n,n)$ whose columns are the bonds $\{i,\pi(i)\}$, which is why these states lie in $\operatorname{Sp}_{(n,n)}\cong\cS$~\cite{sagan_symmetric_2001}. Since
$\dim\cS=\frac1{n+1}\binom{2n}{n}\ll n!$,
these states form a highly redundant spanning family for $\cS$, so their expansion coefficients need not be unique. Their overlaps
are nevertheless explicit. The loop formula of Sutherland~\cite{Sutherland1988}
(see~\cite[eq.~(45)]{BeachSandvik2006}) gives
\begin{equation}\label{eq_overlap}
  \ip{\pi}{\tau}=2^{\cyc(\pi,\tau)-n}\in[2^{1-n},1]
\end{equation}
where $\cyc(\pi,\tau)$ is the number of cycles of $\pi^{-1}\tau$, equivalently the number
of loops in the superposition of the two matchings.

While the valence-bond basis is overcomplete in general, when \(n=2\) then \(\dim\cS = 2!\). This means that when locally restricted, the action of a single edge interaction is uniquely representable in the valence-bond states.
For $i\in L$, $j\in R$ and a matching \(\pi\), the corresponding valence-bond state bonds \(i \mapsto \pi(i)\) and \(\pi^{-1}(j) \mapsto j\). Let \(\ket{\pi\circ(i\ \pi^{-1}(j))}\) be the valence-bond state that bonds \(i \mapsto j\) and \(\pi^{-1}(j) \mapsto \pi(i)\) and all other bonds unchanged (it is $\pi$ itself when $\pi(i)=j$). Then we have the following
\begin{equation}\label{eq_edgeaction}
  \bS_i\!\cdot\bS_j\ket\pi=
  \begin{cases}
    -\tfrac34\ket\pi,
     & \pi(i)=j,    \\[2pt]
    \tfrac14\ket\pi-\tfrac12\ket{\pi\circ(i\ \pi^{-1}(j))},
     & \pi(i)\ne j.
  \end{cases}
\end{equation}
Thus the singlet projector $\tfrac14-\bS_i\!\cdot\bS_j$ fixes $\ket\pi$ when $ij$ is
already a bond and otherwise maps it to $\tfrac12\ket{\pi\circ(i\ \pi^{-1}(j))}$ with positive
coefficient.

\begin{definition}[Lieb-Mattis state]\label{def_LM}
  The \emph{Lieb-Mattis state} is the normalized uniform superposition of valence-bond states,
  \begin{equation}\label{eq_LM}
    \LM\coloneq\frac{2^{n/2}}{n!\,\sqrt{n+1}}\sum_{\pi\in\Sn}\ket\pi
  \end{equation}
\end{definition}

The Lieb-Mattis state is the unique ground state of the Heisenberg model on the complete bipartite graph~\cite{LiebMattis1962}.

\subsection{The walk matrix on matchings}\label{section_walk}

We first represent the shifted Hamiltonian by a nonnegative matrix $W$ on matchings. Its row sums vary, so \cref{section_exact} modifies its diagonal to obtain the transition matrix that the algorithm samples. Fix a shift $c>|E|$ and define the \emph{shifted Hamiltonian}
$\Tc\coloneq cI-4H_G$.
The factor $4$ clears the quarters in~\cref{eq_edgeaction}, while the sign reversal makes
the reconnection coefficients nonnegative.

For matchings $\pi,\tau\in\Sn$, define the \emph{reconnection multiplicity}
\begin{equation}\label{eq_mult}
  \mult(\pi,\tau)
  \coloneq\#\{ij\in E:\pi\circ(i\ \pi^{-1}(j))=\tau\}
\end{equation}
If $\tau\ne\pi$ then $\mult(\pi,\tau)\in\{0,1,2\}$, since the reconnection exchanging the
partners of $a,b\in L$ is induced exactly by those of the cross-edges $a\pi(b)$ and
$b\pi(a)$ that are present in $G$. On the diagonal,
\[
  \mult(\pi,\pi)
  =\#\{i\in L:(i,\pi(i))\in E\}
  \eqcolon\eG(\pi)\in[0,n]
\]
the number of bonds of $\pi$ that are themselves edges of $G$. Since every edge induces
exactly one reconnection,
$\sum_\tau\mult(\pi,\tau)=|E|$ for every $\pi$.

\begin{definition}[Walk matrix]\label{def_walk}
  The \emph{walk matrix} $W\in\mathbb R^{\Sn\times\Sn}$ associated with $(G,c)$ is
  \[
    W_{\pi\pi}\coloneq c-|E|+4\eG(\pi),\qquad
    W_{\pi\tau}\coloneq2\,\mult(\pi,\tau)\quad(\tau\ne\pi)
  \]
  with rows indexed by source matchings and columns by target matchings.
\end{definition}

The name is justified by the following intertwining relation.

\begin{lemma}[Action and row sums]\label{lemma_walk}
  For every matching $\pi$,
  \[
    \Tc\ket\pi=\sum_\tau W_{\pi\tau}\ket\tau,
    \qquad
    \sum_\tau W_{\pi\tau}=c+|E|+2\eG(\pi)
  \]
\end{lemma}

\begin{proof}
  Equation~\cref{eq_edgeaction} gives
  \[
    4\bS_i\!\cdot\bS_j\ket\pi=
    \begin{cases}
      -3\ket\pi,                     & \pi(i)=j,    \\
      \ket\pi-2\ket{\pi\circ(i\ \pi^{-1}(j))}, & \pi(i)\ne j.
    \end{cases}
  \]
  Summing over $ij\in E$ therefore yields
  \[
    \Tc\ket\pi
    =
    \bigl(c+3\eG(\pi)-(|E|-\eG(\pi))\bigr)\ket\pi
    +
    2\!\!\sum_{\substack{ij\in E\\\pi(i)\ne j}}
    \ket{\pi\circ(i\ \pi^{-1}(j))}
  \]
  Collecting equal targets gives the stated entries of $W$, whose sum is
  $c-|E|+4\eG(\pi)+2(|E|-\eG(\pi))=c+|E|+2\eG(\pi)$.
\end{proof}

A coefficient row vector $z^{\top}$ evolves to $z^{\top}W$ under one application
of $\Tc$, and to $z^{\top}W^t$ under $t$ applications. Because the valence-bond states are linearly dependent, an operator has more than one coefficient matrix on this family; $W$ is the one read off
edge by edge from~\cref{eq_edgeaction}. It is nonnegative, and its only
matching-dependent row-sum contribution is $2\eG(\pi)\in[0,2n]$, a relative $O(1/n)$ at
the shift $c=\Theta(n^{2})$ used later. The exact sampler of \cref{section_exact}
removes precisely this defect.

We also require the coefficient matrix of the observable
$O_{ij}\coloneq\bS_i\!\cdot\bS_j$, for $i\in L$ and $j\in R$.
Define $Q=Q^{(ij)}\in\mathbb R^{\Sn\times\Sn}$ by
\[
  Q_{\pi\pi}=
  \begin{cases}
    -\tfrac34, & \pi(i)=j,    \\
    \tfrac14,  & \pi(i)\ne j,
  \end{cases}
  \qquad
  Q_{\pi,\,\pi\circ(i\ \pi^{-1}(j))}=-\tfrac12
  \quad\text{if }\pi(i)\ne j
\]
with all remaining entries zero. Then $O_{ij}\ket\pi=\sum_\tau Q_{\pi\tau}\ket\tau$
by~\cref{eq_edgeaction}. Each row has at most two nonzero entries and absolute row sum
at most $3/4$.

Finally, we record the spectral contraction that determines the required projection
depth. On $\cS$ the eigenvalues of $\Tc$ are
\[
  \mu_k\coloneq c-4\lambda_k,
  \qquad
  \mu_0>\mu_1\ge\mu_2\ge\cdots>0
\]
all positive because $\lambda_k\le|E|/4<c/4$, and $\Tc$ shares its eigenvectors with
$H_G$. The ratio of the two leading eigenvalues is
\begin{equation}\label{eq_ratio}
  \frac{\mu_1}{\mu_0}
  =1-\frac{4\gamma}{\mu_0}
  \le
  1-\frac{4\gamma}{c+3|E|}
\end{equation}
where the last inequality uses $\lambda_0\ge-\tfrac34|E|$. Therefore, for every singlet
$\varphi\perp\Psi_0$,
\[
  \norm{\Tc^t\varphi}
  \le\mu_1^t\norm\varphi
\]
for every $t\ge0$. Relative to the ground state component, the orthogonal component
therefore contracts by a factor at most $(\mu_1/\mu_0)^t$, so the singlet-sector gap controls
the projection depth.

\section{Algorithm}\label{section_algorithm}

This section gives the algorithm of \cref{thm_sampler}; combined with the spectral results established in \cref{sec:rep-theory,section_boson}, it proves \cref{thm_sampler}.
We estimate the ground state correlation $\ev{\bS_i\!\cdot\bS_j}$ by applying powers of the shifted Hamiltonian to a known starting state on both sides of the observable. We implement these powers by sampling walks on valence-bond matchings, following the forward-walking method~\cite{TrivediCeperley1990,NightingaleUmrigar1999}. By~\cref{eq_ratio}, $t$ applications of $\Tc$ suppress the
excited singlet components of a starting state, relative to $\Psi_0$, by the factor
$(\mu_1/\mu_0)^t$. The estimator starts from the Lieb-Mattis state $\LM$ of
\cref{def_LM}. By~\cref{eq_LM}, its
coefficient vector in the valence-bond spanning family is constant, so sampling from it amounts to
drawing a uniformly random matching in $O(n)$ time, and its normalization cancels from
every ratio below. Its overlap with the ground state, $a_{\LM}\coloneq|\ip{\LM}{\Psi_0}|$, is
bounded below on the graph class by \cref{thm_spectral}.

If we project the Lieb-Mattis state from both sides, the ratio
\[
  \frac{\ip{\LM}{\Tc^tO_{ij}\Tc^t\LM}}{\ip{\LM}{\Tc^{2t}\LM}}
\]
converges to $\ev{\bS_i\!\cdot\bS_j}$ as $t\to\infty$, provided $a_{\LM}>0$. To avoid the
overlaps~\cref{eq_overlap}, we replace the bra $\LM$ by the \emph{N\'eel state}
$\Neel\coloneq\ket{0\cdots0}_L\otimes\ket{1\cdots1}_R$. By the Marshall sign convention, every
valence-bond state has overlap $\ip{\Neel}{\pi}=2^{-n/2}$ with it~\cite[eq.~(36)]{BeachSandvik2006},
so the overlap of $\sum_\pi z_\pi\ket\pi$ with $\Neel$ is just $2^{-n/2}\sum_\pi z_\pi$. The scaling
factor cancels out in the ratio, and we use the \emph{forward-walking ratio}
\begin{equation}\label{eq_fw}
  \ev{\bS_i\!\cdot\bS_j}_t
  \coloneq\frac{\ip{\Neel}{\Tc^tO_{ij}\Tc^t\LM}}{\ip{\Neel}{\Tc^{2t}\LM}}
  =\frac{\ip{\LM}{\Tc^tO_{ij}\Tc^t\LM}}{\ip{\LM}{\Tc^{2t}\LM}}
\end{equation}

\subsection{The reference walk}\label{section_exact}

By \cref{lemma_walk}, the row of $W$ indexed by a matching $\pi$ has sum
$c+|E|+2\eG(\pi)$, where $\eG(\pi)$ is the number of bonds of $\pi$ that are edges of
$G$. We make the row sums constant by adding holding weight: define
\begin{equation}\label{eq_hatW}
  D_{\pi\pi}\coloneq2n-2\eG(\pi),
  \qquad
  \widehat W\coloneq W+D,
  \qquad
  Z\coloneq c+|E|+2n
\end{equation}
so that every row of $\widehat W$ sums to $Z$.

\begin{lemma}[Reference walk]\label{lemma_hatW}
  The matrix $P\coloneq\widehat W/Z$ is stochastic, and
  \[
    W\le\widehat W
    \le
    \Bigl(1+\frac{2n}{c-|E|}\Bigr)W
  \]
  entrywise.
\end{lemma}

\begin{proof}
  Adding $D_{\pi\pi}=2n-2\eG(\pi)$ changes the row sum to $Z$. Since $D_{\pi\pi}\ge0$,
  we have $W\le\widehat W$. The two matrices agree off the diagonal, while on it
  \[
    D_{\pi\pi}
    \le2n
    \le
    \frac{2n}{c-|E|}W_{\pi\pi}
  \]
  which proves the upper bound.
\end{proof}

A reference trajectory $x=(x_0,\dots,x_t)$ is sampled by choosing $x_0$ uniformly from
$\Sn$ and then taking $t$ steps of $P$.

One step can be generated without constructing $P$. From the current matching $\pi$,
with probability $2|E|/Z$ choose an edge $i'j'\in E$ uniformly at random and move to
$\pi\circ(i'\ \pi^{-1}(j'))$, which is $\pi$ itself when $\pi(i')=j'$; otherwise hold.
By \cref{def_walk} and~\cref{eq_hatW}, this moves to each $\tau\ne\pi$ with
probability $2\mult(\pi,\tau)/Z=P_{\pi\tau}$. After $O(n^2)$ preprocessing, each step and
the update of $\eG(\pi)$ take $O(1)$ time.

\begin{remark}[Weighted graphs]\label{rem_weighted}
\cref{section_walk} and this section are written for unit weights. For general weights
$0<A_{ij}\le1$, each edge count is replaced by its weighted version: $|E|$ by
$\sum_{ij\in E}A_{ij}$, $\eG(\pi)$ by $\sum_{i\in L}A_{i\pi(i)}$, and $\mult(\pi,\tau)$ by the
total weight of the edges inducing the reconnection $\pi\to\tau$. \cref{lemma_walk,lemma_hatW}
then hold verbatim, with $D_{\pi\pi}\ge0$ because $A_{ij}\le1$, and the sampler chooses the
edge $i'j'$ with probability proportional to $A_{i'j'}$, still in $O(1)$ time after $O(n^2)$
preprocessing. The only other change is at the midpoint of \cref{alg_sample}, where the
reconnection entry of $M$ is at least $2A_{ij}$ instead of $2$; this multiplies the constant
in \cref{lemma_secondmoment} by at most $1/A_{ij}$.
\end{remark}

The reference walk uses $\widehat W$, whereas the shifted Hamiltonian acts through $W$.
For a sampled trajectory define the \emph{importance weight}
\begin{equation}\label{eq_weight}
  w(x)
  \coloneq
  \prod_{s=0}^{t-1}
  \frac{W_{x_sx_{s+1}}}
       {\widehat W_{x_sx_{s+1}}}
\end{equation}
Since $W$ and $\widehat W$ agree off the diagonal, only holding steps contribute to $w(x)$,
each through a factor $W_{x_sx_s}/\widehat W_{x_sx_s}$ that depends only on $\eG(x_s)$.
Multiplying the probability of a trajectory by $w(x)$ replaces each factor $\widehat W$ by
$W$, so weighted trajectories of the reference walk reproduce powers of $W$. This is the
weighted random-walk estimator of matrix powers of Forsythe and
Leibler~\cite{ForsytheLeibler1950}; the transition probabilities come from normalizing the constant-row-sum matrix $\widehat W$~\cite{NightingaleUmrigar1999}.

\subsection{Forward walking}\label{section_forward}

The observable $O_{ij}$ acts locally on a valence-bond state through the matrix $Q$ of
\cref{section_walk}. By~\cref{eq_edgeaction}, the row of $Q$ indexed by $\pi$ is supported
on $\tau=\pi$ and, when $\pi(i)\ne j$, on the single reconnected matching
$\tau=\pi\circ(i\ \pi^{-1}(j))$, with entries of absolute value at most $\tfrac34$.
\cref{alg_sample} inserts $O_{ij}$ at the midpoint of a trajectory by resolving
this two-term sum stochastically.

\begin{figure}[t]
\begin{algorithm}[One forward-walking sample]\label{alg_sample}
\begin{mdframed}[style=algobox]
  \textbf{Input:} projection depth $t$ and edge $ij\in E$.
  \begin{enumerate}[label=\arabic*.,leftmargin=1.8em,itemsep=6pt,topsep=5pt]
    \item Draw $x_0$ uniformly from $\Sn$ and run the reference walk $P$ for $t$ steps,
          obtaining $x=(x_0,\dots,x_t)$.
    \item At the midpoint $\pi\coloneq x_t$, choose the branch $y_0$ and the branch weights
          $\varsigma,\iota$ as follows.
          \begin{itemize}
            \item If $\pi(i)=j$, take $y_0\coloneq\pi$, $\varsigma\coloneq-\tfrac34$ and
                  $\iota\coloneq1$.
            \item If $\pi(i)\ne j$, then with probability $\tfrac13$ take $y_0\coloneq\pi$,
                  $\varsigma\coloneq\tfrac34$ and $\iota\coloneq3$; otherwise take
                  $y_0\coloneq\pi\circ(i\ \pi^{-1}(j))$, $\varsigma\coloneq-\tfrac34$ and
                  $\iota\coloneq0$.
          \end{itemize}
    \item Starting from $y_0$, run the reference walk $P$ for another $t$ steps,
          obtaining $y=(y_0,\dots,y_t)$.
    \item Compute the importance weights $w(x)$ and $w(y)$, and \textbf{return}
          \begin{equation}\label{eq_XY}
            X\coloneq\varsigma\,w(x)w(y),
            \qquad
            Y\coloneq\iota\,w(x)w(y)
          \end{equation}
  \end{enumerate}
\end{mdframed}
\end{algorithm}
\end{figure}

The midpoint branch is chosen with probability proportional to the absolute value of the
corresponding entry of $Q$, and $\varsigma$ restores its sign and magnitude, so averaging
$X$ over the branch reproduces $O_{ij}$. The weight $\iota$ vanishes on a reconnection and
otherwise compensates for the probability of keeping $\pi$, so averaging $Y$ reproduces the
identity.

\begin{lemma}[Forward-walking identity]\label{lemma_fwid}
  Let $t\ge0$ and $ij\in E$, and let $Q$ be the coefficient matrix of $O_{ij}$ from
  \cref{section_walk}. The output~\cref{eq_XY} of \cref{alg_sample} satisfies
  \begin{align}
    \E X
    &=
    \frac{\one^\top W^tQW^t\one}{n!\,Z^{2t}} \label{eq_EX}\\
    \E Y
    &=
    \frac{\one^\top W^{2t}\one}{n!\,Z^{2t}} \label{eq_EY}
  \end{align}
  In particular, $\E X/\E Y=\ev{\bS_i\!\cdot\bS_j}_t$.
\end{lemma}

\begin{proof}
  A reference trajectory $x$ has probability
  $\prod_{s<t}\widehat W_{x_sx_{s+1}}/(n!\,Z^t)$, so, as noted after~\cref{eq_weight},
  $\E[w(x)f(x_t)]=\one^\top W^tf/(n!\,Z^t)$ for every $f\colon\Sn\to\mathbb R$. The branch
  rule of \cref{alg_sample} gives
  $\E[\varsigma\,f(y_0)\mid x_t=\pi]=\sum_\tau Q_{\pi\tau}f(\tau)$ and
  $\E[\iota\,f(y_0)\mid x_t=\pi]=f(\pi)$. We apply these with
  $f(\tau)\coloneq\E[w(y)\mid y_0=\tau]=(W^t\one)_\tau/Z^t$, the expected weight of a second leg
  started at $\tau$. Conditioning on the first leg then gives~\cref{eq_EX}
  and~\cref{eq_EY}.

  For the ratio, \cref{lemma_walk}, its analogue
  $O_{ij}\ket\pi=\sum_\tau Q_{\pi\tau}\ket\tau$ for the observable, and
  $\ip{\Neel}{\pi}=2^{-n/2}$ show that the numerators of~\cref{eq_EX}
  and~\cref{eq_EY} are $2^{n/2}$ times
  \[
    \ip{\Neel}{\Tc^tO_{ij}\Tc^t\textstyle\sum_\pi\ket\pi}
    \quad\text{and}\quad
    \ip{\Neel}{\Tc^{2t}\textstyle\sum_\pi\ket\pi}
  \]
  respectively. Since $\sum_\pi\ket\pi$ is a multiple of $\LM$
  by~\cref{eq_LM}, their ratio is~\cref{eq_fw}.
\end{proof}

\subsection{Projection depth}\label{section_depth}

The deterministic error comes entirely from the finite projection depth.

\begin{lemma}[Projection error]\label{lemma_fwconv}
  Let $\varepsilon>0$, and let the depth $t$ satisfy
  \begin{equation}\label{eq_projectioncondition}
    \Bigl(\frac{\mu_1}{\mu_0}\Bigr)^t
    \le
    \frac{a_{\LM}^2\varepsilon}{24\sqrt{n+1}}
  \end{equation}
  Then
  $|\ev{\bS_i\!\cdot\bS_j}_t-\ev{\bS_i\!\cdot\bS_j}|\le\varepsilon$.
\end{lemma}

Writing $\LM=a_{\LM}\Psi_0+\varphi$ with $\varphi\perp\Psi_0$, the component $\varphi$ is
suppressed relative to $\Psi_0$ by the factor $(\mu_1/\mu_0)^t$, so the numerator and the
denominator of~\cref{eq_fw}, divided by $\mu_0^{2t}$, differ from
$a_{\LM}^2\ev{\bS_i\!\cdot\bS_j}$ and $a_{\LM}^2$ by $O((\mu_1/\mu_0)^t)$. The proof is in
\cref{appendix_fwconv}.

To turn this into a depth, note that $0\le W\le\widehat W$ entrywise and every row of
$\widehat W$ sums to $Z$, so $\mu_0\le Z$ by~\cref{eq_rhoW} below. Hence
$-\log(\mu_1/\mu_0)\ge1-\mu_1/\mu_0=4\gamma/\mu_0\ge4\gamma/Z$,
and~\cref{eq_projectioncondition} holds as soon as
\begin{equation}\label{eq_depth}
  t
  \ge
  \frac{Z}{4\gamma}
  \log\frac{24\sqrt{n+1}}{a_{\LM}^2\varepsilon}
\end{equation}
The gap thus controls the number of projection steps, while the Lieb-Mattis overlap
enters only logarithmically.

\subsection{Bounding the number of samples}\label{section_secondmoment}

We now control the statistical error of estimating $\ev{\bS_i\!\cdot\bS_j}_t$ from
finitely many trajectories. For $N$ independent samples $(X_k,Y_k)$ from
\cref{alg_sample}, with sample means $\bar X\coloneq N^{-1}\sum_{k=1}^NX_k$ and
$\bar Y\coloneq N^{-1}\sum_{k=1}^NY_k$, the estimator is
\begin{equation}\label{eq_estimator}
  \frac{\bar X}{\bar Y}
\end{equation}
Since $\E X/\E Y=\ev{\bS_i\!\cdot\bS_j}_t$ by \cref{lemma_fwid}, it remains to
determine how large $N$ must be for the empirical ratio to concentrate around this ratio
of expectations.

Two effects make this ratio expensive to estimate: a small Lieb-Mattis overlap makes $\E Y$
small, and the importance weights may vary from one trajectory to another. A decay of the
weights common to all trajectories cancels between the numerator and the denominator
of~\cref{eq_estimator}; only their spread costs samples, as in importance sampling
generally~\cite{ChatterjeeDiaconis2018}.

To measure this spread, define
\begin{equation}\label{eq_M}
  M_{\pi\tau}
  \coloneq
  \begin{cases}
    \displaystyle
    \frac{W_{\pi\tau}^2}{\widehat W_{\pi\tau}},
    & \widehat W_{\pi\tau}>0,\\[1ex]
    0,
    & \widehat W_{\pi\tau}=0.
  \end{cases}
\end{equation}
Squaring the weight of a reference transition $\pi\to\tau$ turns its sampling weight
$\widehat W_{\pi\tau}$ into $M_{\pi\tau}$, so powers of $M$ govern the second moments of
the importance weights, as powers of $W$ govern their first moments. This is the
second-moment matrix of the Forsythe-Leibler estimator.

Write $\varrho_W\coloneq\varrho(W)$ for the spectral radius of the nonnegative matrix $W$ (its Perron-Frobenius eigenvalue). Since every
row of $\widehat W$ sums to $Z$, we have $\varrho(\widehat W)=Z$. Moreover,
\begin{equation}\label{eq_rhoW}
  \varrho_W=\mu_0=c-4\lambda_0
\end{equation}
Indeed, by the Perron-Frobenius theorem there is a nonzero $x\ge0$ with
$W^\top x=\varrho_Wx$. By \cref{lemma_walk}, $\sum_\pi x_\pi\ket\pi$ is an
eigenvector of $\Tc$ with eigenvalue $\varrho_W$; it is nonzero because its overlap with
$\Neel$ is $2^{-n/2}\sum_\pi x_\pi>0$. Hence $\varrho_W\le\mu_0$. Conversely, the valence-bond
states span $\cS$, so $\Psi_0=\sum_\pi y_\pi\ket\pi$ for some $y$, and
\cref{lemma_walk} gives $\mu_0^m\Psi_0=\sum_\pi\bigl((W^\top)^my\bigr)_\pi\ket\pi$.
Since valence-bond states are unit vectors, $\mu_0^m\le\norm{(W^\top)^my}_1$ for every $m$, and Gelfand's
formula gives $\mu_0\le\varrho_W$.

The relative second moment of a forward-walking sample, which has two legs of length $t$,
is controlled by the following bound on the powers of $M$.

\begin{lemma}[Path-weight bound]\label{lemma_uncond_cert}
  Let $G$ be a connected balanced bipartite graph and $c>|E|$, and let
  $b_0\coloneq4n/(c+|E|)$, so that $b_0<2n/|E|$. For every $t\ge0$,
  \[
    \one^\top M^{2t}\one
    \le
    n!\,
    e^{2b_0t}
    \Bigl(\frac{\varrho_W^2}{Z}\Bigr)^{2t}
  \]
\end{lemma}

\begin{proof}
  Since $0\le W\le\widehat W$ entrywise, $M_{\pi\tau}\le\widehat W_{\pi\tau}$ for all
  $\pi,\tau$, and every row of $\widehat W$ sums to $Z$; hence
  $\one^\top M^{2t}\one\le n!\,Z^{2t}$. The spectral radius of a nonnegative matrix is at
  least its smallest row sum~\cite[Thm.~8.1.22]{HornJohnson2013}, so $\varrho_W\ge c+|E|$,
  and $(Z/(c+|E|))^{4t}=(1+2n/(c+|E|))^{4t}\le e^{2b_0t}$. Finally, $c+|E|>2|E|$ gives
  $b_0<2n/|E|$.
\end{proof}

\begin{lemma}[Relative second moment]\label{lemma_secondmoment}
  Let $t\ge0$ and $ij\in E$, and let $b_0$ be as in \cref{lemma_uncond_cert}. The
  output~\cref{eq_XY} of \cref{alg_sample} satisfies
  \begin{equation}\label{eq_relative_second}
    \max\{\E[X^2],\E[Y^2]\}
    \le
    \frac{3(1+Z)}{a_{\LM}^4}
    e^{2b_0t}
    (\E Y)^2
  \end{equation}
\end{lemma}

Squaring the importance weights turns the powers of $W$ that govern the first moments
into powers of $M$, the random midpoint choice increases the second-moment bound by at most the matrix factor $3(I+M)$, and
\cref{lemma_uncond_cert} then bounds both second moments by
$3(1+Z)e^{2b_0t}(\varrho_W/Z)^{4t}$. Since $\E Y\ge a_{\LM}^2(\varrho_W/Z)^{2t}$
by~\cref{eq_EY} and~\cref{eq_rhoW}, the factor $(\varrho_W/Z)^{4t}$ cancels. The proof is in
\cref{appendix_secondmoment}.

\subsection{Estimator guarantee}\label{section_reduction}

The following theorem makes explicit how the gap, the overlap and the path weights enter.

\begin{theorem}[Estimator reduction]\label{thm_reduction}
  Let $G$ be a connected balanced bipartite graph with $n$ vertices on each side, let
  $ij\in E$ and $c>|E|$, and let $Z$ be as in~\cref{eq_hatW}. Suppose that
  $\gamma\ge\gamma_0>0$ and $a_{\LM}\ge a_{\LM,0}>0$. There are universal constants
  $C_1,C_2>0$ such that the following holds for all $\varepsilon,\delta\in(0,1)$. Let
  \begin{equation}\label{eq_tchoice}
    t
    \ge
    C_1\frac{Z}{\gamma_0}
    \log\frac{\sqrt{n+1}}{a_{\LM,0}^2\,\varepsilon}
  \end{equation}
  let $b_0$ be as in \cref{lemma_uncond_cert}, and let
  \begin{equation}\label{eq_Nchoice}
    N
    \ge
    C_2
    \frac{1+Z}{a_{\LM,0}^4\,\varepsilon^2\delta}
    e^{2b_0t}
  \end{equation}
  Then, with probability at least $1-\delta$, the estimator~\cref{eq_estimator} built
  from $N$ independent runs of \cref{alg_sample} at depth $t$ satisfies
  \[
    \Bigl|\frac{\bar X}{\bar Y}-\ev{\bS_i\!\cdot\bS_j}\Bigr|
    \le\varepsilon
  \]
  After $O(n^2)$ preprocessing, it is computed in time $O(Nt)$.
\end{theorem}

\begin{proof}
  We split the error into a projection part and a statistical part, each at most
  $\varepsilon/2$.

  \emph{Projection.} Apply~\cref{eq_depth} with $\varepsilon/2$ in place of
  $\varepsilon$. We have $\gamma\ge\gamma_0$, $a_{\LM}\ge a_{\LM,0}$ and
  $\sqrt{n+1}/(a_{\LM,0}^2\varepsilon)\ge\sqrt2$, so a suitable universal $C_1$
  in~\cref{eq_tchoice} makes \cref{lemma_fwconv} give
  $|\ev{\bS_i\!\cdot\bS_j}_t-\ev{\bS_i\!\cdot\bS_j}|\le\varepsilon/2$.

  \emph{Sampling.} By \cref{lemma_secondmoment} and $a_{\LM}\ge a_{\LM,0}$, both
  $\E[X^2]$ and $\E[Y^2]$ are at most $\Lambda(\E Y)^2$, where
  \[
    \Lambda
    \coloneq
    \frac{3(1+Z)}{a_{\LM,0}^4}
    e^{2b_0t}
  \]
  Let $\eta\coloneq\varepsilon/8$. By Chebyshev's inequality and a union bound, with
  probability at least $1-2\Lambda/(N\eta^2)$ both $|\bar X-\E X|$ and $|\bar Y-\E Y|$
  are at most $\eta\,\E Y$. The second expression in~\cref{eq_fw} is a Rayleigh
  quotient of $O_{ij}$, so $|\E X/\E Y|\le\norm{O_{ij}}=\tfrac34$ by
  \cref{lemma_fwid}. On this event, linearizing the ratio as usual for ratio
  estimators,
  \[
    \Bigl|\frac{\bar X}{\bar Y}-\frac{\E X}{\E Y}\Bigr|
    =
    \frac{\bigl|(\bar X-\E X)-(\E X/\E Y)(\bar Y-\E Y)\bigr|}{\bar Y}
    \le
    \frac{7\eta/4}{1-\eta}
    \le
    2\eta
    \le
    \frac{\varepsilon}{2}
  \]
  The failure probability $2\Lambda/(N\eta^2)=128\Lambda/(N\varepsilon^2)$ is at most
  $\delta$ under~\cref{eq_Nchoice} with $C_2\coloneq384$. Since $\E X/\E Y=\ev{\bS_i\!\cdot\bS_j}_t$
  by \cref{lemma_fwid}, the two parts combine to the claim.

  \emph{Running time.} Each run of \cref{alg_sample} takes $2t$ steps of the
  reference walk, each in $O(1)$ time after $O(n^2)$ preprocessing (\cref{section_exact}).
\end{proof}

By \cref{lemma_uncond_cert}, the bounds of
\cref{thm_reduction} depend only on the graph size, the gap and the Lieb-Mattis
overlap. Since $b_0=4n/(c+|E|)$ and $Z\le2(c+|E|)$, the depth~\cref{eq_tchoice} gives
$2b_0t=O\bigl((n/\gamma)\log(\sqrt{n+1}/(a_{\LM}^2\varepsilon))\bigr)$,
so~\cref{eq_Nchoice} becomes
\[
  N
  =
  O\biggl(
    \frac{1+Z}{\delta}
    \Bigl(\frac{\sqrt{n+1}}{a_{\LM}^2\,\varepsilon}\Bigr)^{2+O(n/\gamma)}
  \biggr)
\]
The shift $c$ cancels from the exponent: a larger $c$ lengthens the walk in proportion
to $Z$ but lowers $b_0$ in proportion to $1/(c+|E|)$. What accumulates over the projection
is the gap per vertex, $\gamma/n$. For polynomial $c$ the cost is therefore polynomial
whenever $n/\gamma=O(1)$, even for an inverse-polynomial overlap $a_{\LM}=n^{-k}$, at the
price of a factor $n^{O(kn/\gamma)}$ in $N$. On $\mathcal G(p,C)$,
\cref{thm_spectral} gives $n/\gamma\le4/p$ and a constant lower bound on $a_{\LM}$.

\section{Representation Theoretic Decomposition of the Hamiltonian}\label{sec:rep-theory}

In this section we initiate the analysis needed to lower bound the spectral gap and the ground state overlap with the Lieb--Mattis state. We do this by decomposing our Hamiltonian using the representation theory of the symmetric group and the \(\fraksu(2)\) Lie algebra. Intuitively, if our interaction graph were the complete balanced bipartite graph, then we could diagonalize our Hamiltonian and analyze its spectrum in the singlet sector exactly. More than this, the exact solution tells us where the eigenvectors with energy close to the ground state energy live. This provides a starting point for approximating the spectrum of the Hamiltonian with a dense bipartite interaction graph. In particular, using the splitting \(A=\hatp\one\one^{\top}+\xi\), with \(\hatp:=\frac1{n^2}\sum_{ij\in E}A_{ij}\), into a mean term and a centered fluctuation term, and recalling the notation \(\bS_L:=\sum_{i\in L}\bS_i\) and \(\bS_R:=\sum_{j\in R}\bS_j\), we have the following decomposition.
\begin{equation}\label{eq_Hsplit}
H_G=\Hmf+\Hxi,
\qquad
\Hmf:=\hatp\bS_L\cdot\bS_R,
\qquad
\Hxi:=\sum_{i\in L}\sum_{j\in R}\xi_{ij}\bS_i\cdot\bS_j.
\end{equation}

We call \(\Hmf\) the \emph{mean-field Hamiltonian}, which is exactly solvable and diagonalizable, and \(\Hxi\) the \emph{fluctuation operator}, which contains the nonuniform structure of the interaction graph.

Comparing directly to the mean-field Hamiltonian will only get us so far. Instead, we will refine this comparison using the singular-value decomposition of the centered adjacency matrix \(\xi\). This leads to a model describing the graph-dependent fluctuations around the mean-field Hamiltonian and, ultimately, to a bosonic model that can be analyzed explicitly by a Bogoliubov transformation.

\subsection{Representations Across The Pair}\label{ss:pairmodel}

We start the decomposition by recalling the Schur--Weyl duality decomposition of \cref{eq_schur_weyl_dual_decomp}. By Lieb--Mattis \cite{LiebMattis1962}, we know that the ground state must lie in the singlet sector, \(\cS\). On the surface, the global Schur--Weyl decomposition tells us very little beyond the fact that the total spin is zero. However, if we use the fact that our model is bipartite, then we can apply the Schur--Weyl decomposition separately to both sides of the bipartition. Let \(n=2m+r\) for \(r\in\{0,1\}\). Then we have the following decomposition, where each direct sum is taken in integer increments of the spin variable.

\[
(\mathbb{C}^2)^{\otimes n} \otimes (\mathbb{C}^2)^{\otimes n} \cong \bigoplus_{s_1 = r/2}^{n/2} \bigoplus_{s_2 = r/2}^{n/2} V_{s_1} \otimes \operatorname{Sp}_{(n/2+s_1,n/2-s_1)} \otimes V_{s_2} \otimes \operatorname{Sp}_{(n/2+s_2,n/2-s_2)}.
\]

The Clebsch--Gordan decomposition \cite{hall_lie_2015,sakurai_modern_2020} tells us that a total spin-\(0\) subspace can only occur when \(s_1=s_2\). Moreover, within each \(V_s\otimes V_s\) the zero-spin subspace, \(V_0\), occurs with multiplicity one and has dimension one and the embedding \(V_0\to V_s\otimes V_s\) is explicit, \(1\mapsto\sum_{q=-s}^{s}\frac{(-1)^{s-q}}{\sqrt{2s+1}}\,\ket{s,q}\ket{s,-q}\), where \(\ket{s,q}\) are the weight vectors of \(V_s\). This gives us a new identification of \(\cS\):
\begin{align}
\cS \cong \bigoplus_{s = r/2}^{n/2} \operatorname{Sp}_{(n/2+s,n/2-s)} \otimes \operatorname{Sp}_{(n/2+s,n/2-s)}.\label{eq_first_cU_iso}
\end{align}

Fix a spin \(s\) and basis elements \(\alpha\in\operatorname{Sp}_{(n/2+s,n/2-s)}\) and \(\beta\in\operatorname{Sp}_{(n/2+s,n/2-s)}\). We can identify the corresponding copy of \(V_s\otimes V_s\) with the weight basis vectors \(\ket{\alpha;s,q_1}\otimes\ket{\beta;s,q_2}\coloneq\ket{s,q_1}\otimes\alpha\otimes\ket{s,q_2}\otimes\beta\). In particular, the above isomorphism, restricted to the sublattice spin \(s\) sector, is given by \(\operatorname{Sp}_{(n/2+s,n/2-s)}\otimes\operatorname{Sp}_{(n/2+s,n/2-s)}\to\cS,\ \alpha\otimes\beta\mapsto\sum_{q=-s}^{s}\frac{(-1)^{s-q}}{\sqrt{2s+1}}\,\ket{\alpha;s,q}\ket{\beta;s,-q}\). We will expand on this more in the next subsection.

The exact spectrum of the mean-field Hamiltonian in the singlet sector, \(\Hmf |_\cS\), then follows by a simple expansion of the total-spin Casimir, together with the facts that the total-spin Casimir vanishes in the singlet sector and that the left and right sublattice spin Casimirs agree there:
\begin{align}
\bS_{\mathrm{tot}}^2 &= (\bS_L + \bS_R)^2 = \bS_L^2 + \bS_R^2 + 2\bS_L \cdot \bS_R, \nonumber\\
\bS_L \cdot \bS_R &= \frac{1}{2}\left(\bS_{\mathrm{tot}}^2 - \bS_L^2 - \bS_R^2\right), \nonumber\\
\bS_L \cdot \bS_R \ket{\psi} &= - \bS_L^2 \ket{\psi} = - \bS_R^2 \ket{\psi}, \qquad \ket{\psi} \in \cS. \label{eq_Casimir_rewrite}
\end{align}

In particular, on the sublattice spin-\(s\) block of the singlet sector, \(\bS_L^2=\bS_R^2=s(s+1)I\), and thus \(\Hmf=\sum_{s = r/2}^{n/2}-\hatp s(s+1) P_s\), where \(P_s\) is the projector onto the sublattice spin \(s\) sector. From this we can tell that the ground state of \(\Hmf\) in the singlet sector occurs at spin \(s=n/2\). Since \(\operatorname{Sp}_{(n,0)}\) is one-dimensional, this ground state is unique and is the Lieb-Mattis state from \cref{def_LM}.

\subsection{Harmonic Analysis}\label{ss:harmonics}

It is left to decompose the \emph{fluctuation operator}, \(\Hxi\). We do this with tighter control over how \(\cS\) embeds into pairs of spins.

Instead of appealing directly to Schur--Weyl duality, we might decompose our space of qubits in the left partition, \(L=[n]\), into Young modules, or equivalently into excitation spaces,
\begin{align}
(\mathbb{C}^2)^{\otimes n}
\cong
\bigoplus_{k=0}^{n} \mathbb{C}^{\binom{[n]}{k}}.
\label{eq_excitation_decomp}
\end{align}

This is given by the isomorphism \(\cV^L:\bigoplus_{k=0}^{n}\C^{\binom{[n]}{k}}\to(\mathbb{C}^2)^{\otimes n}\), which, for a fixed level-\(k\) set-function, \(f\in\C^{\binom{[n]}{k}}\), maps \(f\mapsto\sum_{S\in\binom{[n]}{k}}f(S)\bigotimes_{i=1}^n\ket{\mathbbm{1}_{S}(i)}\). In particular, subsets specify where the ones are in the standard basis vectors.

We next understand the action of \(\mathfrak{su}(2)\), or equivalently its complexification \(\fraksl(2,\C)\), on \(\bigoplus_{k=0}^{n}\C^{\binom{[n]}{k}}\). To do this we pull back two separate ways \(\mathfrak{su}(2)\) can act on the left sublattice: on the individual qubit level, through \(\bS_i\), and on the level of the total spin of the left sublattice, through \(\bS_L\). To this end, we start with the following definitions.

\begin{definition}[Actions on qubits in the left sublattice]\label{def_left_sublat_qubit_actions}
  We pull back the actions of \(S_i^z\), \(S_i^+\), and \(S_i^-\) under the isomorphism \(\cV^L\) and define the following operators:
  \begin{enumerate}
    \item \(\Pi_i^L \coloneq (\cV^L)^{-1}\left(\frac{1}{2}I-S_i^z\right)\cV^L \in \cL\left(\bigoplus_{k=0}^{n}\C^{\binom{[n]}{k}}\right)\).
    \item \(\partial_i^L \coloneq (\cV^L)^{-1}S_i^+\cV^L \in \cL\left(\bigoplus_{k=0}^{n}\C^{\binom{[n]}{k}}\right)\).
    \item \(\left(\partial_i^L\right)^\dagger \coloneq (\cV^L)^{-1}S_i^-\cV^L \in \cL\left(\bigoplus_{k=0}^{n}\C^{\binom{[n]}{k}}\right)\).
  \end{enumerate}
\end{definition}

To understand these definitions, we first observe that \(\frac{1}{2}I-S_i^z\) is the projector onto standard basis states with the \(i\)th qubit equal to \(1\). Thus, under the identification \(\cV^L\), \(\Pi_i^L\) is the projector onto set-functions supported on sets \(S\) containing \(i\). In particular, \((\Pi_i^L\psi)(S)=\psi(S)\) if \(i\in S\), and \(0\) otherwise. It also follows that \(S_L^z=\sum_{i\in L}S_i^z\) is block diagonal with respect to the decomposition in \cref{eq_excitation_decomp}; in particular, restricted to level \(k\), \(S_L^z|_{\cV^L\mathbb{C}^{\binom{[n]}{k}}}=(n/2-k)I\).

Next, consider the raising operator \(S_i^+\). Let \(\psi\in\C^{\binom{[n]}{k}}\). Then \(\partial_i^L\psi\in\C^{\binom{[n]}{k-1}}\), and in particular \((\partial_i^L\psi)(S)=\psi(S\cup\{i\})\) if \(i\notin S\), and \(0\) otherwise. The lowering operator is the adjoint of the raising operator, \(S_i^-=(S_i^+)^\dagger\), and correspondingly \(\left(\partial_i^L\right)^\dagger\) is the adjoint of \(\partial_i^L\). For \(\psi\in\C^{\binom{[n]}{k}}\), we have \(\left(\partial_i^L\right)^\dagger\psi\in\C^{\binom{[n]}{k+1}}\), with \(\left(\left(\partial_i^L\right)^\dagger\psi\right)(S)=\psi(S\setminus\{i\})\) if \(i\in S\), and \(0\) otherwise.

When considering the action of the total raising operator on the left sublattice, \(S_L^+=\sum_{i=1}^nS_i^+\), we will observe that it moves between excitation spaces in a highly symmetric way. Under the set-function representation, its pullback is the operator commonly called the \emph{down operator} in the harmonic analysis of set functions. Harmonic analysis of set functions is well studied, with applications in many fields \cite{stanley_differential_1988,godsil_association_2010}. In particular, this approach has recently been used to study the Quantum Max-Cut Hamiltonian \cite{basu_quantum_2026}, although not on bipartite graphs restricted to the singlet sector.

\begin{definition}[Left sublattice down/up operators]
We denote the \emph{left sublattice down operator} by \(\partial^{L}\coloneq(\cV^L)^{-1}S_L^+\cV^L\). It is defined on set functions in the following way (for \(\psi\in\C^{\binom{[n]}{k}}\), \(S'\in\binom{[n]}{k-1}\)):
\begin{align}
  \partial^{L}\big|_{\mathbb{C}^{\binom{[n]}{k}}}
  :
  \C^{\binom{[n]}{k}}
  &\to
  \C^{\binom{[n]}{k-1}},
  \\
  (\partial^{L}\psi)(S')
  &=
  \sum_{i\notin S'}\psi(S'\cup\{i\}).\nonumber
\end{align}

We denote the \emph{left sublattice up operator} by \(\left(\partial^{L}\right)^\dagger\coloneq(\cV^L)^{-1}S_L^-\cV^L\). It is defined on set functions in the following way (for \(\psi\in\C^{\binom{[n]}{k}}\), \(S'\in\binom{[n]}{k+1}\)):
\begin{align}
    \left(\partial^{L}\right)^\dagger
    \Big|_{\mathbb{C}^{\binom{[n]}{k}}}
    :
    \C^{\binom{[n]}{k}}
    &\to
    \C^{\binom{[n]}{k+1}},
    \\
    \left(\left(\partial^{L}\right)^\dagger\psi\right)(S')
    &=
    \sum_{i\in S'}\psi(S'\setminus\{i\}).\nonumber
\end{align}
\end{definition}

\begin{definition}[Harmonic functions]\label{def_left_harmonic_function}
A level-\(k\) set-function, \(\psi\in\C^{\binom{[n]}{k}}\), is a \emph{harmonic function} if \(\partial^{L}\psi=0\). Let \(\cH_k^L\coloneq\ker\left(\partial^L|_{\mathbb{C}^{\binom{[n]}{k}}}\right)\) be the space of left sublattice level-\(k\) harmonic set-functions.
\end{definition}

The down operator has the special property that, for \(k\leq n/2\), \(\cH_k^L\cong\operatorname{Sp}_{(n-k,k)}\) \cite{stanley_differential_1988}. To see this, recall from Young's rule, \cref{eq_youngs_rule}, that \(\C^{\binom{[n]}{k}}\cong\bigoplus_{\ell=0}^k\operatorname{Sp}_{(n-\ell,\ell)}\). The level-\(k\) excitation subspace has \(S_L^z\)-weight \(q=n/2-k\). Within the copy of \(\operatorname{Sp}_{(n-\ell,\ell)}\), Schur--Weyl duality identifies the corresponding spin as \(s=n/2-\ell\). The raising operator annihilates this weight space precisely when \(q=s\), or equivalently when \(\ell=k\), i.e., it is a space of highest-weight vectors. Similarly, when \(k>n/2\), we have \(\cH_k^L=0\).

We next consider the right sublattice. Inspired by the antisymmetric/dual nature of the Heisenberg interaction/Quantum Max-Cut, we will identify the right sublattice with the lowest-weight vectors. This choice is fairly arbitrary, but we believe it is more consistent with the intuition that assignments to the bipartite lattice should have opposing spins (\(0\) and \(1\), or \(\uparrow\) and \(\downarrow\)).

We start by defining the isomorphism in \cref{eq_excitation_decomp} as \(\cV^R:\bigoplus_{k=0}^{n}\C^{\binom{[n]}{k}}\to(\mathbb{C}^2)^{\otimes n}\), which, for a fixed level-\(k\) set-function, \(f\in\C^{\binom{[n]}{k}}\), maps \(f\mapsto\sum_{S\in\binom{[n]}{k}}f(S)\bigotimes_{i=1}^n\ket{\mathbbm{1}_{\overline{S}}(i)}\). In particular, subsets specify where the zeros are in the standard basis vectors. Under this isomorphism, the actions of the raising and lowering operators are reversed when pulling back to set functions. 

\begin{definition}[Actions on qubits in the right sublattice]\label{def_right_sublat_qubit_actions}
  We pull back the actions of \(S_i^z\), \(S_i^+\), and \(S_i^-\) under the isomorphism \(\cV^R\) and define the following operators:
  \begin{enumerate}
    \item \(\Pi_i^R \coloneq (\cV^R)^{-1}\left(\frac{1}{2}I+S_i^z\right)\cV^R \in \cL\left(\bigoplus_{k=0}^{n}\C^{\binom{[n]}{k}}\right)\), where \((\Pi_i^R\psi)(S)=\psi(S)\) if \(i\in S\), and \(0\) otherwise.
    \item \(\partial_i^R \coloneq (\cV^R)^{-1}S_i^-\cV^R \in \cL\left(\bigoplus_{k=0}^{n}\C^{\binom{[n]}{k}}\right)\), where \((\partial_i^R\psi)(S)=\psi(S\cup\{i\})\) if \(i\notin S\), and \(0\) otherwise.
    \item \(\left(\partial_i^R\right)^\dagger \coloneq (\cV^R)^{-1}S_i^+\cV^R \in \cL\left(\bigoplus_{k=0}^{n}\C^{\binom{[n]}{k}}\right)\), where \(\left(\left(\partial_i^R\right)^\dagger\psi\right)(S)=\psi(S\setminus\{i\})\) if \(i\in S\), and \(0\) otherwise.
  \end{enumerate}
\end{definition}

Moreover, we define the down operator as the pullback of the lowering operator on the right sublattice as follows.

\begin{definition}[Right sublattice down/up operators]
We denote the \emph{right sublattice down operator} by \(\partial^{R}\coloneq(\cV^R)^{-1}S_R^-\cV^R\). It is defined on set functions in the following way (for \(\psi\in\C^{\binom{[n]}{k}}\), \(S'\in\binom{[n]}{k-1}\)):
\begin{align}
  \partial^{R}\big|_{\mathbb{C}^{\binom{[n]}{k}}}
  :
  \C^{\binom{[n]}{k}}
  &\to
  \C^{\binom{[n]}{k-1}},
  \\
  (\partial^{R}\psi)(S')
  &=
  \sum_{i\notin S'}\psi(S'\cup\{i\}). \nonumber
\end{align}

We denote the \emph{right sublattice up operator} by \(\left(\partial^{R}\right)^\dagger\coloneq(\cV^R)^{-1}S_R^+\cV^R\). It is defined on set functions in the following way (for \(\psi\in\C^{\binom{[n]}{k}}\), \(S'\in\binom{[n]}{k+1}\)):
\begin{align}
    \left(\partial^{R}\right)^\dagger
    \Big|_{\mathbb{C}^{\binom{[n]}{k}}}
    :
    \C^{\binom{[n]}{k}}
    &\to
    \C^{\binom{[n]}{k+1}},
    \\
    \left(\left(\partial^{R}\right)^\dagger\psi\right)(S')
    &=
    \sum_{i\in S'}\psi(S'\setminus\{i\}).\nonumber
\end{align}
\end{definition}

\begin{definition}[Harmonic functions]\label{def_right_harmonic_function}
A level-\(k\) set-function, \(\psi\in\C^{\binom{[n]}{k}}\), is a \emph{harmonic function} if \(\partial^{R}\psi=0\). Let \(\cH_k^R\coloneq\ker\left(\partial^R|_{\mathbb{C}^{\binom{[n]}{k}}}\right)\) be the space of right sublattice level-\(k\) harmonic set-functions.
\end{definition}

Similar to the left sublattice, the right sublattice down operator has the property that, for \(k\leq n/2\), \(\cH_k^R\cong\operatorname{Sp}_{(n-k,k)}\). However, this time it corresponds to the space of lowest-weight vectors. Similarly, when \(k>n/2\), we have \(\cH_k^R=0\). We will study these operators more throughout this section.


Now, both \(\cH_k^L\) and \(\cH_k^R\) are naturally isomorphic to \(\operatorname{Sp}_{(n-k,k)}\). Setting \(s_k=n/2-k\), we then have the more explicit version of the isomorphism from \cref{eq_first_cU_iso} (let \(\alpha\otimes\beta\in\cH_k^L\otimes\cH_k^R\)):
\begin{align}
\cU:
\bigoplus_{k=0}^{\lfloor n/2\rfloor}
\cH_k^L\otimes\cH_k^R
&\to
\cS,\label{eq_cU_iso}
\\
\alpha\otimes\beta
&\mapsto
(-1)^k\sum_{q=-s_k}^{s_k}
\frac{(-1)^{s_k-q}}{\sqrt{2s_k+1}}
\ket{\alpha;s_k,q}\ket{\beta;s_k,-q}. \nonumber
\end{align}
Here \(\ket{\alpha;s_k,s_k}=\cV^L\alpha\) and \(\ket{\beta;s_k,-s_k}=\cV^R\beta\), and
\begin{align}
    \ket{\alpha;s_k,q}
    :=
    \sqrt{\frac{(s_k+q)!}{(2s_k)!\,(s_k-q)!}}\,
    (S_L^-)^{s_k-q}\ket{\alpha;s_k,s_k},
    \quad
    \ket{\beta;s_k,-q}
    :=
    \sqrt{\frac{(s_k+q)!}{(2s_k)!\,(s_k-q)!}}\,
    (S_R^+)^{s_k-q}\ket{\beta;s_k,-s_k}.\label{eq_ket_alpha_s_q_terms}
\end{align}
The overall sign $(-1)^k$, while arbitrary, will ensure that the off diagonal blocks in \cref{prop_bond_item2} are negative.

For notation, we denote by \(\cM_k=\cH_k^L\otimes\cH_k^R\). We can make \(\cM_k\) a Hilbert space by using the standard inner product from \(\C^{\binom{[n]}{k}}\otimes\C^{\binom{[n]}{k}}\). Then \(\cM=\bigoplus_{k=0}^{\lfloor n/2\rfloor}\cM_k\) and \(\cU:\cM\to\cS\). The remainder of this section will study the Hamiltonian restricted to the singlet sector, \(H_G|_{\cS}\), viewed under the above reformulation, \(H\coloneq\cU^{-1}\left(H_G|_{\cS}\right)\cU\). Under this reformulation we can express the decomposition given in \cref{eq_Hsplit} as follows. We let \(P_k\) be the projector onto \(\cM_k\).
\begin{equation}
\cU^{-1}\Hmf\cU
=
-\hatp\frac{n}{2}\left(\frac{n}{2}+1\right)I
+
\hatp\sum_{k=0}^{\lfloor n/2\rfloor}
k(n-k+1)P_k,
\qquad
\cU^{-1}\Hxi\cU
=
\sum_{i\in L}\sum_{j\in R}
\xi_{ij}
\cU^{-1}
\bS_i\cdot\bS_j
\cU.\label{eq_Hmf_and_Hxi}
\end{equation}
Here we used \cref{eq_Casimir_rewrite} for \(\Hmf\).

\begin{lemma}\label{lemma_extremal}
\(\cU\) is unitary. Moreover, let \(\alpha\in\cH_k^L\) and \(\beta\in\cH_k^R\). Then, for every \(\ket{\psi}\in\cS\) we have that
\[\bra{\cU(\alpha\otimes\beta)}\ket{\psi} = (-1)^k\sqrt{2s_k+1}\, \bra{\alpha;s_k,s_k} \bra{\beta;s_k,-s_k} \ket{\psi}.\]
\end{lemma}

\begin{proof}
First, \(\cU\) is unitary simply because the vectors \(\cU(\alpha\otimes\beta)\), as \(k\), \(\alpha\), and \(\beta\) range over orthonormal bases of the corresponding harmonic spaces, form an orthonormal basis for \(\cS\).

Next, by construction,
\(\cU(\alpha\otimes\beta)
        =
        (-1)^k\sum_{q=-s_k}^{s_k}
        \frac{(-1)^{s_k-q}}{\sqrt{2s_k+1}}\,
        \ket{\alpha;s_k,q}
        \ket{\beta;s_k,-q}.\)
For \(\alpha\in\cH_k^L\) and \(\beta\in\cH_k^R\), the corresponding copy of \(V_{s_k}\otimes V_{s_k}\) contains a unique one-dimensional singlet subspace, spanned by \(\cU(\alpha\otimes\beta)\). Since \(\ket{\psi}\in\cS\), its projection onto this copy of \(V_{s_k}\otimes V_{s_k}\) must therefore be proportional to \(\cU(\alpha\otimes\beta)\). Observing \(\bra{\cU(\alpha\otimes\beta)}\ket{\alpha;s_k,s_k}
        \ket{\beta;s_k,-s_k}
        =
        \frac{(-1)^k}{\sqrt{2s_k+1}}
   \) finishes the proof.
\end{proof}

\subsubsection{Single Edge Interaction Terms}

We focus our attention on a single edge interaction term, \(\bS_i\cdot\bS_j\) for \(i \in L\) and \(j \in R\). We initiate this study by expanding out the spin vectors, \(\bS_i \cdot \bS_j = S_i^x S_j^x + S_i^y S_j^y + S_i^z S_j^z\), using the decompositions \(S_i^x = \frac{1}{2}\left(S_i^- + S_i^+\right), S_i^y = \frac{i}{2}\left(S_i^- - S_i^+\right)\) to get the following:
\begin{align}
  \bS_i \cdot \bS_j &= \frac{1}{2}\left(S_i^+ + S_i^-\right) \frac{1}{2}\left(S_j^+ + S_j^-\right) + \frac{i}{2}\left(S_i^- - S_i^+\right) \frac{i}{2}\left(S_j^- - S_j^+\right) + S_i^z S_j^z \nonumber \\
  &= \frac{1}{2}\left(S_i^+S_j^-  + S_i^-S_j^+\right) + S_i^z S_j^z \label{eq_SS_in_sl2_env_alg}
\end{align}

Here the conjugation of \(\cU\) will roughly replace \(S_i^+S_j^-\) with \(\partial_i^L \otimes \partial_j^R\) and \(S_i^z S_j^z\) with \((1/2I - \Pi_i^L) \otimes (1/2I - \Pi_j^R)\), from \cref{def_left_sublat_qubit_actions,def_right_sublat_qubit_actions}. To do this, we further look at the blocks given by \(P_{\ell} \cU^\dagger \bS_i \cdot \bS_j \cU P_k : \cM_k \to \cM_\ell\).

\begin{proposition}\label{prop_bond}
Let $i\in L$ and $j\in R$.
\begin{enumerate}[label=\textup{(\roman*)}]
\item\label{prop_bond_item1} $P_\ell\,\mathcal U^\dagger(\bS_i\!\cdot\!\bS_j)\mathcal U P_k=0$ whenever $|\ell-k|>1$.
\item\label{prop_bond_item2} For $k\ge1$, $P_{k-1}\,\mathcal U^\dagger(\bS_i\!\cdot\!\bS_j)\mathcal U P_k=-\Gamma_k\,P_{k-1}\,\partial_i^L\otimes\partial_j^R P_k$, where
\[
  \Gamma_k:=\frac12\sqrt{\frac{2s_k+3}{2s_k+1}}=\frac12\sqrt{\frac{n-2k+3}{n-2k+1}},
\]
and $P_k\,\mathcal U^\dagger(\bS_i\!\cdot\!\bS_j)\mathcal U P_{k-1}$ is its adjoint.
\item\label{prop_bond_item3} If $s_k>0$, then
\[
  P_k\,\mathcal U^\dagger(\bS_i\!\cdot\!\bS_j)\mathcal U P_k
  =-\frac{s_k+1}{s_k}P_k\Bigl(\frac12I-\Pi_i^L\Bigr)\otimes\Bigl(\frac12I-\Pi_j^R\Bigr)P_k;
\]
if $s_k=0$, then $P_k\,\mathcal U^\dagger(\bS_i\!\cdot\!\bS_j)\mathcal U P_k=0$.
\end{enumerate}
\end{proposition}

\begin{proof}
  In this proof we use four commutator relations that we give here for reference. These are direct results of the commutator relations of the generators of \(\fraksl(2,\C)\) (see \cref{section_rep_thy_prelims}). Nothing changes when \(L\) is replaced with \(R\) and \(i\) is replaced with \(j\).
  \begin{align}
    [S_L^z,S_i^+] &= \sum_{i' \in L}[S_{i'}^z,S_i^+] = [S_i^z,S_i^+] = S_i^+ \label{eq_comm_rel_1} \\
    [S_L^z,S_i^-] &= \sum_{i' \in L}[S_{i'}^z,S_i^-] = [S_i^z,S_i^-] = -S_i^- \label{eq_comm_rel_2} \\
    [S_i^+,S_L^-] &= \sum_{i' \in L}[S_i^+, S_{i'}^-] = [S_i^+, S_i^-] = 2S_i^z \label{eq_comm_rel_3} \\
    [S_i^-,S_L^+] &= \sum_{i' \in L}[S_i^-, S_{i'}^+] = [S_i^-, S_i^+] = -2S_i^z \label{eq_comm_rel_4}
  \end{align}

  Fix \(\alpha \in \cH_k^L\), \(\beta \in \cH_k^R\), \(\alpha' \in \cH_{\ell}^L\), and \(\beta' \in \cH_{\ell}^R\). We use \(s_k = n/2-k\) and \(s_\ell = n/2-\ell\) to denote the corresponding spins. Without loss of generality we will consider \(\ell \leq k\) (equivalently, \(s_\ell \geq s_k\)) as the reverse will follow from the fact that \(\bS_i \cdot \bS_j\) is self adjoint.

  Let \(\ket{\psi} = \bS_i \cdot \bS_j \cU(\alpha \otimes \beta)\), which has \(\ket{\psi} \in \cS\) as \(\bS_i \cdot \bS_j\) commutes with the total spin. \cref{lemma_extremal} then gives us that
  \begin{align*}
    \bra{\alpha' \otimes \beta'}\cU^\dagger\bS_i \cdot \bS_j \cU\ket{\alpha \otimes \beta} = \bra{\cU(\alpha' \otimes \beta')}\ket{\psi} = (-1)^\ell\sqrt{2s_\ell+1}\, \bra{\alpha';s_\ell,s_\ell} \bra{\beta';s_\ell,-s_\ell} \ket{\psi}
  \end{align*}
  We first consider the \(S_i^+ S_j^-\) term (from \cref{eq_SS_in_sl2_env_alg}) and expand \(\cU\) using \cref{eq_cU_iso}:
  \begin{align*}
    \bra{\alpha' \otimes \beta'}\cU^\dagger S_i^+ S_j^- \cU\ket{\alpha \otimes \beta}  &= (-1)^\ell\sqrt{2s_\ell+1}\, \bra{\alpha';s_\ell,s_\ell} \bra{\beta';s_\ell,-s_\ell} S_i^+ S_j^- \cU\ket{\alpha \otimes \beta} \\
    &= (-1)^{\ell+k}\frac{\sqrt{2s_\ell+1}}{\sqrt{2s_k+1}}\,\sum_{q=-s_k}^{s_k} (-1)^{s_k-q} \bra{\alpha';s_\ell,s_\ell}S_i^+\ket{\alpha;s_k,q}\bra{\beta';s_\ell,-s_\ell}S_j^-\ket{\beta;s_k,-q}
  \end{align*}
  We next argue that the term \(\bra{\alpha';s_\ell,s_\ell}S_i^+\ket{\alpha;s_k,q} = 0\) for most values of \(q\). To do this, we observe that the commutation relation in \cref{eq_comm_rel_1} implies \(S_i^+\ket{\alpha;s_k,q}\) is in the eigenspace of \(S_L^z\) with eigenvalue \(q+1\). Indeed, \(S_L^zS_i^+\ket{\alpha;s_k,q} = S_i^+ (S_L^z + 1)\ket{\alpha;s_k,q} = (q+1) S_i^+\ket{\alpha;s_k,q}\). The conclusion then follows as \(\ket{\alpha';s_\ell,s_\ell}\) is in the \(s_\ell\) eigenspace. In particular, for \(q + 1 \neq s_\ell\), \(\ket{\alpha';s_\ell,s_\ell}\) and \(S_i^+\ket{\alpha;s_k,q}\) are in different eigenspaces and thus must be orthogonal (as \(S_L^z\) is self-adjoint), or rather \(\bra{\alpha';s_\ell,s_\ell}S_i^+\ket{\alpha;s_k,q} = 0\) (again, unless \(q + 1= s_\ell\)).
  
  The same argument applies for the \(\bra{\beta';s_\ell,-s_\ell}S_j^-\ket{\beta;s_k,-q}\) term by using \cref{eq_comm_rel_2}. The above completely vanishes unless \(s_\ell - 1 \leq s_k \implies  \ell + 1 \geq k\), which combined with the assumption that \(\ell \leq k\) means that \(\ell \in \{k - 1, k\}\). We will first take the case of \(\ell = k-1\) which is equivalent to \(s_\ell = s_k + 1\).
  \begin{align}
    \bra{\alpha' \otimes \beta'}\cU^\dagger S_i^+ S_j^- \cU\ket{\alpha \otimes \beta} &= -\frac{\sqrt{2s_k + 3}}{\sqrt{2s_k+1}} \bra{\alpha';s_k + 1,s_k + 1}S_i^+\ket{\alpha;s_k,s_k}\bra{\beta';s_k + 1,-s_k - 1}S_j^-\ket{\beta;s_k,-s_k} \nonumber\\
    &= -\frac{\sqrt{2s_k + 3}}{\sqrt{2s_k+1}} \braket{\alpha', \left(\cV^L\right)^\dagger S_i^+ \cV^L\alpha} \braket{\beta', \left(\cV^R\right)^\dagger S_j^-\cV^R \beta} \nonumber\\
    &= -2\Gamma_k \braket{\alpha' \otimes \beta', \partial_i^L \otimes \partial_j^R \alpha \otimes \beta} \label{eq_sjdskskkkskaslkdubdejs1}
  \end{align}
  The last two equalities follow from the definitions of \cref{def_left_sublat_qubit_actions,def_right_sublat_qubit_actions,eq_cU_iso,eq_ket_alpha_s_q_terms}.

  We next consider the case of \(\ell = k\), which is equivalent to \(s_\ell = s_k\). Here, we will assume \(s_k \neq 0\) (i.e., \(k \neq n/2\)) as we will handle that case last. In this case the only non-vanishing term is \(q = s_\ell - 1 = s_k - 1\). We use \cref{eq_ket_alpha_s_q_terms,eq_comm_rel_3,eq_comm_rel_4} to get the following. 
  \begin{align}
    \bra{\alpha' \otimes \beta'}\cU^\dagger S_i^+ S_j^- \cU\ket{\alpha \otimes \beta} &= - \bra{\alpha';s_k,s_k}S_i^+\ket{\alpha;s_k,s_k - 1}\bra{\beta';s_k,-s_k}S_j^-\ket{\beta;s_k,-s_k + 1} \nonumber\\
    &= - \frac{1}{2s_k} \bra{\alpha';s_k,s_k}S_i^+S_L^-\ket{\alpha;s_k,s_k}\bra{\beta';s_k,-s_k}S_j^-S_R^+\ket{\beta;s_k,-s_k} \nonumber\\
    &= \frac{2}{s_k} \bra{\alpha';s_k,s_k}S_i^z\ket{\alpha;s_k,s_k}\bra{\beta';s_k,-s_k}S_j^z\ket{\beta;s_k,-s_k} \nonumber\\
    &= - \frac{2}{s_k} \braket{\alpha' \otimes \beta', \Bigl(\frac12I-\Pi_i^L\Bigr)\otimes\Bigl(\frac12I-\Pi_j^R\Bigr) \alpha \otimes \beta} \label{eq_sjdskskkkskaslkdubdejs2}
  \end{align}
  Here in the third equality we used the fact that \(\ket{\alpha';s_k,s_k}\) is a highest-weight vector and thus \(\bra{\alpha';s_k,s_k}S_L^- = \left(S_L^+ \ket{\alpha';s_k,s_k}\right)^\dagger = 0\) to remove the \(\bra{\alpha';s_k,s_k}S_L^-S_i^+\ket{\alpha;s_k,s_k} = 0\) coming from \cref{eq_comm_rel_3,eq_comm_rel_4}. Notice that \cref{eq_comm_rel_4} contributes an extra sign. Then the last equality is again by \cref{def_left_sublat_qubit_actions,def_right_sublat_qubit_actions,eq_cU_iso,eq_ket_alpha_s_q_terms} where we account for the identity terms in the definitions for \(\Pi^L_i\) and \(\Pi^R_j\).

  We next consider the \(S_i^- S_j^+\) term. The above work all goes through virtually the same way but with \cref{eq_comm_rel_2} for the left sublattice and \cref{eq_comm_rel_1} for the right. This will in turn make all \(q - 1 \neq s_\ell\) terms vanish in the summation when expanding \(\cU\) with \cref{eq_cU_iso}. As the summation is over \(-s_k \leq q \leq s_k\) this requires \(k \leq \ell - 1\), but this contradicts the assumption that \(\ell \leq k\). Thus \(\bra{\alpha' \otimes \beta'}\cU^\dagger S_i^- S_j^+ \cU\ket{\alpha \otimes \beta} = 0\).

  Finally, we consider the term \(S_i^z S_j^z\). We make the observation that \([S_L^z,S_i^z]= 0\). Following the same steps as above we get
  \begin{align*}
    \bra{\alpha' \otimes \beta'}\cU^\dagger S_i^z S_j^z \cU\ket{\alpha \otimes \beta} &= (-1)^{k+\ell}\frac{\sqrt{2s_\ell+1}}{\sqrt{2s_k+1}}\,\sum_{q=-s_k}^{s_k} (-1)^{s_k-q} \bra{\alpha';s_\ell,s_\ell}S_i^z\ket{\alpha;s_k,q}\bra{\beta';s_\ell,-s_\ell}S_j^z\ket{\beta;s_k,-q}
  \end{align*}
  The commutator relation implies that \(\bra{\alpha';s_\ell,s_\ell}S_i^z\ket{\alpha;s_k,q} = 0\) whenever \(q \neq s_\ell\), which requires \(k = \ell\). The same thing happens on the right sublattice, and thus we have that
  \begin{align}
    \bra{\alpha' \otimes \beta'}\cU^\dagger S_i^z S_j^z \cU\ket{\alpha \otimes \beta} &= \bra{\alpha';s_k,s_k}S_i^z\ket{\alpha;s_k,s_k}\bra{\beta';s_k,-s_k}S_j^z\ket{\beta;s_k,-s_k} \nonumber \\
    &= -\braket{\alpha' \otimes \beta', \Bigl(\frac12I-\Pi_i^L\Bigr)\otimes\Bigl(\frac12I-\Pi_j^R\Bigr) \alpha \otimes \beta}. \label{eq_sjdskskkkskaslkdubdejs3}
  \end{align}
  This again follows by accounting for the identity terms in \(\Pi^L_i\) and \(\Pi^R_j\) in \cref{def_left_sublat_qubit_actions,def_right_sublat_qubit_actions}.

  Putting everything together, \cref{prop_bond_item1} follows from the fact that when \(\ell < k - 1\) everything vanishes. Then \cref{prop_bond_item2} follows from \cref{eq_sjdskskkkskaslkdubdejs1} (notice the \(\frac{1}{2}\) in \cref{eq_SS_in_sl2_env_alg}). Finally, \cref{prop_bond_item3} follows from \cref{eq_sjdskskkkskaslkdubdejs2,eq_sjdskskkkskaslkdubdejs3} as \(\frac{1}{s_k} + 1 = \frac{s_k + 1}{s_k}\) (again, notice the \(\frac{1}{2}\) in \cref{eq_SS_in_sl2_env_alg}).

  It remains to consider \(s_k=0\). In this case, \(\ket{\alpha;0,0}\) is both a highest- and lowest-weight vector, so there is no \(\ket{\alpha;s_k,s_k-1}\) term and hence the contribution from \cref{eq_sjdskskkkskaslkdubdejs2} vanishes. For the \(S_i^zS_j^z\) term, we use \cref{eq_comm_rel_3}:
\begin{align*}
  2\bra{\alpha';0,0}S_i^z\ket{\alpha;0,0}
  &=
  \bra{\alpha';0,0}[S_i^+,S_L^-]\ket{\alpha;0,0} \\
  &=
  \bra{\alpha';0,0}S_i^+S_L^-\ket{\alpha;0,0}
  -
  \bra{\alpha';0,0}S_L^-S_i^+\ket{\alpha;0,0}
  =0.
\end{align*}
Therefore, \(\bra{\alpha' \otimes \beta'}\cU^\dagger \bS_i \cdot \bS_j \cU\ket{\alpha \otimes \beta} = 0\).
\end{proof}

\subsubsection{The Complete Bipartite Model}\label{ss:decomposition}

In this section we will use \cref{prop_bond} to have a better expression for \(H_\xi\) as defined in \cref{eq_Hmf_and_Hxi}. It will be useful to separate the same level terms from \cref{prop_bond}, \cref{prop_bond_item3} and the different level/off diagonal terms \cref{prop_bond}, \cref{prop_bond_item2}. We make the following definitions.
\begin{align}
  \Hpair:=\sum_{k=1}^{\lfloor n/2\rfloor}\bigl(P_{k-1}\cU^\dagger H_\xi \cU P_k+P_k\cU^\dagger H_\xi \cU P_{k-1}\bigr),
  \qquad
  \Hsl:=\sum_{k=0}^{\lfloor n/2\rfloor}P_k\cU^\dagger H_\xi \cU P_k
\end{align}

It will be useful to make the following definition.

\begin{definition}[Weighted boundary operators]
  Given real valued weights, \(u \in \R^n\), we define the \emph{left sublattice weighted down operator} as
  \[\partial^L_u = \sum_{i \in L} u(i) \partial_i^L\]
  and the \emph{right sublattice weighted down operator} as
  \[\partial^R_u = \sum_{j \in R} u(j) \partial_j^R.\]
  We also define \(\Pi_u^L = \sum_{i \in L} u(i) \Pi_i^L\) and \(\Pi_u^R = \sum_{j \in R} u(j) \Pi_j^R\) similarly.
\end{definition}

We make the observation that when restricted to \(\cH_k^L\) then \(\partial^L_u |_{\cH_k^L} : \cH_k^L \to \cH_{k-1}^L\). This is because\footnote{Indeed, for \(\psi \in \C^{\binom{[n]}{k}}\) and \(S \in \binom{[n]}{k-2}\) we have that \((\partial^L_i \partial^L_{i'} \psi)(S) = \psi(S \cup \{i,i'\}) = (\partial^L_{i'} \partial^L_{i} \psi)(S)\).} \([\partial^L_u, \partial^L_v] = 0\) for any \(u,v \in \R^n\). In particular, for \(\psi \in \cH^L_k\) (which is to say \(\partial^L_\one \psi = 0\)), we have that \(\partial^L_\one \partial^L_u \psi = \partial^L_u \partial^L_\one \psi = 0\) and thus by \cref{def_left_harmonic_function}, \(\partial_u^L \psi \in \cH^L_{k-1}\) is a harmonic set-function. Moreover, we can observe that for \(u = u' + c\one\) with \(u' \bot \one\) then \(\partial^L_u |_{\cH^L_k} = \partial^L_{u'} |_{\cH^L_k}\), and thus it suffices to consider only \(u \bot \one\).

For the adjoints, \((\partial_u^L)^\dagger\), we no longer have that \([\partial^L_u,(\partial^L_v)^\dagger] = 0\), and thus it is not the case that \((\partial_u^L)^\dagger\) maps harmonic functions to harmonic functions. This won't end up mattering as we always project onto \(\cM\). We can nonetheless compute this commutator exactly as follows.

\begin{lemma}\label{lemma_comm_rel_partial_partial_dagger} We have that \([\partial^L_u,(\partial^L_v)^\dagger] = \braket{u,v}I - 2 \Pi^L_{u \odot v}\), where \(u \odot v \in \R^n\) is the entrywise product. Additionally, we have that \([\partial^R_u,(\partial^R_v)^\dagger] = \braket{u,v}I - 2 \Pi^R_{u \odot v}\).
\end{lemma}

\begin{proof}
  By bilinearity of the commutator,
  \begin{align*}
    [\partial_u^L,(\partial_v^L)^\dagger]
    &=
    \sum_{i,j=1}^n
    u(i)v(j)
    [\partial_i^L,(\partial_j^L)^\dagger].
  \end{align*}
  Next, using \cref{def_left_sublat_qubit_actions} we have that
  \(\left[\partial_i^L,(\partial_j^L)^\dagger\right]
  =(\cV^L)^{-1}[S_i^+,S_j^-]\cV^L\). Here, we use the \(\fraksl(2,\C)\) commutator identity, \([S_i^+,S_i^-]=2S_i^z\), to get
  \begin{align*}
    [\partial_u^L,(\partial_v^L)^\dagger]
    &=
    2\sum_{i=1}^n
    u(i)v(i)
    (\cV^L)^{-1}S_i^z\cV^L \\
    &=
    2\sum_{i=1}^n
    u(i)v(i)
    \left(\frac12I-\Pi_i^L\right) \\
    &=
    \braket{u,v}I
    -
    2\Pi_{u\odot v}^L.
  \end{align*}
  The case for the right sublattice follows more or less identically.
\end{proof}

We next apply \cref{prop_bond} to the term $\Hxi$. Recall from \cref{section_intro_setting} that the fluctuation matrix, \(\xi\), defined by \(A=\hatp\one\one^{\top}+\xi\), with \(\hatp:=\frac1{n^2}\sum_{ij\in E}A_{ij}\), can be split into its doubly centered part, \(\xi_c\), and its degree deviations. Moreover, we write the singular value decomposition of $\xi_c$:
\begin{equation}
\begin{split}
  \xi=\xi_c+\tfrac1n\rho\one^{\top}+\tfrac1n\one\rho'^{\top},
  \qquad
  &\xi_c=\sum_{m=1}^r\sigma_mu_mv_m^{\top}, \\
  \rho_i:=\sum_{j}A_{ij}-\hatp n, \qquad &\rho'_j:=\sum_{i}A_{ij}-\hatp n
\end{split}\label{eq_xi_defs}
\end{equation}
with $r:=\operatorname{rank}\xi_c$, the singular values $\sigma_1\ge\dots\ge\sigma_r>0$, and the orthonormal singular vectors
$(u_m)_{m \in [r]}$, $(v_m)_{m \in [r]}$. Here, doubly centered means that $\xi_c\one=0=\one^{\top}\xi_c$, thus $u_m,v_m\perp\one$. Moreover, by the definition of \(\rho\) and \(\rho'\), $\sum_i\rho_i=\sum_j\rho'_j=0$.

In what follows we abuse notation slightly and write \(\Hmf \coloneq -\hatp\frac{n}{2}\left(\frac{n}{2}+1\right)I + \hatp\sum_{k=0}^{\lfloor n/2\rfloor} k(n-k+1)P_k\) from \cref{eq_Hmf_and_Hxi} rather than use the notation defined in \cref{eq_Hsplit}.

\begin{proposition}[Pair decomposition]\label{prop_pair_decomp}
For \(H_G = \sum_{ij \in E} A_{ij} \bS_i \cdot \bS_j\), the Hamiltonian \(H = \cU^\dagger H_G \cU\) is block tridiagonal in the levels,
\begin{equation}\label{eq:tridiagonal}
  P_\ell HP_k=0\qquad\text{whenever }|\ell-k|>1,
\end{equation}
and decomposes exactly as
\begin{equation}\label{eq:pair-decomposition}
  H= \Hmf + \Hpair + \Hsl,
\end{equation}
where:
\begin{enumerate}[label=\textup{(\roman*)}]
\item\label{prop_pair_decomp_1} the mean-field term is diagonal in the level, with
\begin{equation}\label{eq:Hmf-final}
  \Hmf\big|_{\cM_k}
  =\bigl[\Emf+\hatp\,k(n-k+1)\bigr]I_{\cM_k},
  \qquad
  \Emf:=-\hatp\,\frac n2\Bigl(\frac n2+1\Bigr);
\end{equation}
\item\label{prop_pair_decomp_2} the pair term involves only the doubly centered fluctuation,
\begin{equation}\label{eq:Hpair-explicit}
  \Hpair
  =-\sum_{m=1}^{r}\sigma_m
  \sum_{k=1}^{\lfloor n/2\rfloor}\Gamma_k
  \Bigl[
  P_{k-1}(\partial_{u_m}^L\otimes\partial_{v_m}^R)P_k
  +
  P_k\bigl((\partial_{u_m}^L)^{\dagger}\otimes(\partial_{v_m}^R)^{\dagger}\bigr)P_{k-1}
  \Bigr];
\end{equation}
\item\label{prop_pair_decomp_3} the same-level term vanishes on a level with $s_k=0$, and otherwise
\begin{equation}\label{eq:Hsl-exact}
  P_k\Hsl P_k
  =-\frac{s_k+1}{s_k}
  \sum_{i\in L}\sum_{j\in R}(\xi_c)_{ij}\,
  P_k\Pi_i^L \otimes \Pi_j^R P_k
  +\frac{s_k+1}{n}P_k\left(\Pi_{\rho}^L \otimes I
  +I\otimes \Pi_{\rho'}^R\right)P_k.
\end{equation}
\end{enumerate}
\end{proposition}

\begin{proof}
  The block-tridiagonal structure in \cref{eq:tridiagonal} follows immediately from \cref{prop_bond}, \cref{prop_bond_item1}. The decomposition in \cref{eq:pair-decomposition} then follows from the definitions of \(\Hmf\), \(\Hpair\), and \(\Hsl\). Moreover, \cref{prop_pair_decomp_1} is exactly \cref{eq_Hmf_and_Hxi}.

  For \cref{prop_pair_decomp_2}, \cref{prop_bond}, \cref{prop_bond_item2} gives
  \begin{align*}
    P_{k-1}\Hxi P_k
    &=
    -\Gamma_k
    \sum_{i\in L}\sum_{j\in R}
    \xi_{ij}\,
    P_{k-1}
    \left(\partial_i^L\otimes\partial_j^R\right)
    P_k.
  \end{align*}
  Using \(\xi=\xi_c+\frac1n\rho\one^\top+\frac1n\one(\rho')^\top\), the two degree-deviation terms vanish because
  \(\sum_i\partial_i^L=\partial_{\one}^L=0\) and \(\sum_j\partial_j^R=\partial_{\one}^R=0\) on the harmonic spaces. Hence
  \begin{align*}
    P_{k-1}\Hxi P_k
    &=
    -\Gamma_k
    \sum_{i,j}(\xi_c)_{ij}\,
    P_{k-1}
    \left(\partial_i^L\otimes\partial_j^R\right)
    P_k.
  \end{align*}
  Writing the singular-value decomposition \(\xi_c=\sum_{m=1}^r\sigma_m u_m v_m^\top\), we obtain
  \begin{align*}
    P_{k-1}\Hxi P_k
    &=
    -\sum_{m=1}^r\sigma_m\Gamma_k\,
    P_{k-1}
    \left(\partial_{u_m}^L\otimes\partial_{v_m}^R\right)
    P_k.
  \end{align*}
  Summing over \(k\) and adding the adjoint blocks gives \cref{eq:Hpair-explicit}.

  We now prove \cref{prop_pair_decomp_3}. Suppose \(s_k>0\). By \cref{prop_bond_item3},
  \begin{align*}
    P_k\Hsl P_k
    &=
    -\frac{s_k+1}{s_k}
    \sum_{i,j}\xi_{ij}\,
    P_k
    \left(\frac12I-\Pi_i^L\right)
    \otimes
    \left(\frac12I-\Pi_j^R\right)
    P_k.
  \end{align*}
  For the doubly centered term, \(\xi_c\one=0\) and \(\one^\top\xi_c=0\), so after expanding the product all constant and linear terms vanish. Thus its contribution is
  \begin{align*}
    -\frac{s_k+1}{s_k}
    \sum_{i,j}(\xi_c)_{ij}\,
    P_k\Pi_i^L\otimes\Pi_j^R P_k.
  \end{align*}

  For the term \(\frac1n\rho\one^\top\), we use
  \(\sum_j\Pi_j^R=kI\), and hence \(\sum_j(\frac12I-\Pi_j^R)=s_kI\). Therefore its contribution is
  \begin{align*}
    -\frac{s_k+1}{n}
    \sum_i\rho_i
    \left(\frac12I-\Pi_i^L\right)\otimes I
    &=
    \frac{s_k+1}{n}\Pi_\rho^L\otimes I,
  \end{align*}
  where we used \(\sum_i\rho_i=0\). Similarly, the term \(\frac1n\one(\rho')^\top\) contributes
  \(\frac{s_k+1}{n}I\otimes\Pi_{\rho'}^R\). Combining the three contributions gives \cref{eq:Hsl-exact}.

  Finally, if \(s_k=0\), each same-level bond block vanishes by \cref{prop_bond_item3}, and therefore \(P_k\Hsl P_k=0\).
\end{proof}

\subsection{Bounds on \texorpdfstring{\(\Hsl\)}{Hsl} and \texorpdfstring{\(\Hpair\)}{Hpair}}

Here, we prove two bounds on \(\Hsl\) and \(\Hpair\). Given a function \(f: \{0,1,\dotsc, \lfloor n/2 \rfloor\} \to \R\), it will be useful to define an operator of multiplication on each level \(\C^{\binom{[n]}{k}}\), which we will denote by \(f(\Pi^L)\) acting on the left sublattice and \(f(\Pi^R)\) acting on the right sublattice. Here we use \(\Pi^L = \sum_{i \in L} \Pi^L_i\) and same for \(R\). Notice that \(\Pi^L|_{\C^{\binom{[n]}{k}}} = kI\). When \(\psi \in \cM\), then \(\Pi^L \psi = \Pi^R \psi\). We note that for \(0 \leq k \leq n/2\) we have the inequality \(k \leq \frac2n k(n+1-k)\), which can be lifted to an operator inequality
\begin{equation}\label{eq:level-cost}
  \Pi^L \preccurlyeq \frac2n \Pi^L(n+1-\Pi^L).
\end{equation}
This inequality will become important when we will need to compare the spectra of \(\Hpair\) and \(\Hmf\).

\begin{proposition}[Same-level bound]\label{lem:sl}
In the context of \cref{prop_pair_decomp}, for every level $k$,
\begin{equation}\label{eq:sl-bound}
  \norm{P_k\Hsl P_k}_{\opn}\le 3k(\sigma_1+\Delta),
\end{equation}
where \(\Delta:=\max\{\max_i|\rho_i|,\ \max_j|\rho'_j|\}\). Equivalently, for every $\psi\in\cM$,
\begin{equation}\label{eq:sl-form-bound}
  |\ip{\psi}{\Hsl\psi}|\le3(\sigma_1+\Delta)\,\ip{\psi}{\Pi^L\psi}.
\end{equation}
\end{proposition}

\begin{proof}
  If \(s_k=0\), then \(P_k\Hsl P_k=0\) by \cref{prop_bond}, \cref{prop_bond_item3}, so suppose \(s_k>0\).

  Since \(0\preccurlyeq \Pi_i^L,\Pi_j^R\preccurlyeq I\), we also have
  \begin{align*}
    \left(P_k\Pi_i^L P_k\right)^2
    \preccurlyeq
    P_k\Pi_i^L P_k, \qquad
    \left(P_k\Pi_j^R P_k\right)^2
    \preccurlyeq
    P_k\Pi_j^R P_k.
  \end{align*}
  Moreover, on level \(k\), \(\sum_i\Pi_i^L=kI\) and \(\sum_j\Pi_j^R=kI\), and hence
  \begin{align*}
    \sum_i\left(P_k\Pi_i^L P_k\right)^2
    \preccurlyeq kP_k, \qquad
    \sum_j\left(P_k\Pi_j^R P_k\right)^2
    \preccurlyeq kP_k.
  \end{align*}

  Applying the operator-valued Cauchy--Schwarz inequality to the first term of \cref{eq:Hsl-exact} gives
  \begin{align*}
    \left\|
    \sum_{i,j}(\xi_c)_{ij}\,
    P_k\Pi_i^L\otimes\Pi_j^R P_k
    \right\|_{\opn}
    &\leq
    \norm{\xi_c}_{\opn}
    \left\|
    \sum_i\left(P_k\Pi_i^L P_k\right)^2
    \right\|_{\opn}^{1/2}
    \left\|
    \sum_j\left(P_k\Pi_j^R P_k\right)^2
    \right\|_{\opn}^{1/2} \\
    &\leq
    k\sigma_1.
  \end{align*}

  For the degree-deviation terms, let \(\phi\in\cH_k^L\) be a unit vector. Since
  \(\ip{\phi}{\Pi_i^L\phi}\geq0\) and
  \(\sum_i\ip{\phi}{\Pi_i^L\phi}=k\),
  \begin{align*}
    \left|
    \ip{\phi}{\Pi_\rho^L\phi}
    \right|
    &=
    \left|
    \sum_i\rho_i\ip{\phi}{\Pi_i^L\phi}
    \right|
    \leq
    k\Delta.
  \end{align*}
  Hence \(\norm{\Pi_\rho^L}_{\opn}\leq k\Delta\), and similarly
  \(\norm{\Pi_{\rho'}^R}_{\opn}\leq k\Delta\).

  Using \cref{eq:Hsl-exact},
  \begin{align*}
    \norm{P_k\Hsl P_k}_{\opn}
    &\leq
    \frac{s_k+1}{s_k}k\sigma_1
    +
    \frac{2(s_k+1)}{n}k\Delta \\
    &\leq
    \frac{s_k+1}{s_k}k(\sigma_1+\Delta).
  \end{align*}
  Since the smallest positive value of \(s_k\) is \(1/2\), we have \((s_k+1)/s_k\leq3\), proving \cref{eq:sl-bound}.

  Finally, since \(\Hsl\) is level diagonal,
  \begin{align*}
    |\ip{\psi}{\Hsl\psi}|
    &\leq
    \sum_k
    3k(\sigma_1+\Delta)\norm{P_k\psi}^2 \\
    &=
    3(\sigma_1+\Delta)\ip{\psi}{\Pi^L\psi},
  \end{align*}
  which proves \cref{eq:sl-form-bound}.
\end{proof}

\begin{proposition}\label{prop:OI}
  In the context of \cref{prop_pair_decomp,lem:sl}, we have
  \begin{equation}\label{eq:Hpair-lower}
    \Hpair
    \succcurlyeq
    -\left(
    \frac{\hatp n}{4}
    +
    \frac{3\sigma_1^2}{\hatp n}
    \right)\Pi^L
    -
    3(1-\hatp)nI.
  \end{equation}
  Moreover, if \(\sigma_1+\Delta\leq\hatp n/100\), then
  \begin{equation}\label{eq:OI-def}
    \Hxi
    \succcurlyeq
    -\frac{3\hatp}{4}\,
    \Pi^L(n+1-\Pi^L)
    -3nI.
  \end{equation}
\end{proposition}

\begin{proof}
  This follows from a sum of Hermitian squares proof of positivity. In particular, fix \(\omega>0\) and \(m\in[r]\). We use the notation \(\Gamma_{\Pi^R}=f(\Pi^R)\), where \(f(k)=\frac12\sqrt{\frac{n-2k+3}{n-2k+1}}\) as defined in \cref{prop_bond}, \cref{prop_bond_item2}. Let
  \begin{equation}
    Y_m
    :=
    \sqrt{\omega}\,\partial_{u_m}^L\otimes I
    +
    \frac{\sigma_m}{\sqrt{\omega}}\,
    I\otimes\Gamma_{\Pi^R}(\partial_{v_m}^R)^\dagger.
  \end{equation}
  We then have
  \begin{align*}
    0\preccurlyeq Y_m^\dagger Y_m
    &=
    \left(
    \sqrt{\omega}(\partial_{u_m}^L)^\dagger\otimes I
    +
    \frac{\sigma_m}{\sqrt{\omega}}
    I\otimes\partial_{v_m}^R\Gamma_{\Pi^R}
    \right)
    \left(
    \sqrt{\omega}\partial_{u_m}^L\otimes I
    +
    \frac{\sigma_m}{\sqrt{\omega}}
    I\otimes\Gamma_{\Pi^R}(\partial_{v_m}^R)^\dagger
    \right) \\
    &=
    \omega(\partial_{u_m}^L)^\dagger\partial_{u_m}^L\otimes I \\
    &\quad+
    \sigma_m\left(
    (\partial_{u_m}^L)^\dagger\otimes \Gamma_{\Pi^R}
    (\partial_{v_m}^R)^\dagger
    +
    \partial_{u_m}^L\otimes
    \partial_{v_m}^R\Gamma_{\Pi^R}
    \right) \\
    &\quad+
    \frac{\sigma_m^2}{\omega}
    I\otimes
    \partial_{v_m}^R\Gamma_{\Pi^R}^2(\partial_{v_m}^R)^\dagger.
  \end{align*}
  Restricted to \(\cM\), the middle term is the negative of the \(m\)th summand of \(\Hpair\). Thus, after summing over \(m\),
  \begin{align}
    \Hpair
    &\preccurlyeq
    \omega\sum_m
    (\partial_{u_m}^L)^\dagger\partial_{u_m}^L\otimes I +
    \frac1\omega\sum_m\sigma_m^2
    I\otimes
    \partial_{v_m}^R\Gamma_{\Pi^R}^2(\partial_{v_m}^R)^\dagger.
    \label{eq:Hpair-upper-temp}
  \end{align}

  We use four facts to get to the next step. First, because the \(u_m\)'s are orthonormal, \(\sum_m u_mu_m^\mathsf{T}\preccurlyeq I\), and thus
  \(\sum_m(\partial_{u_m}^L)^\dagger\partial_{u_m}^L\preccurlyeq\Pi^L\).
  Second, \(\Gamma_k^2\leq3/4\) for every \(k\leq n/2\).
  Third, \(\sigma_1\geq\sigma_m\) for all \(m\).
  Finally, \cref{lemma_comm_rel_partial_partial_dagger} tells us that
  \(\partial_{v_m}^R(\partial_{v_m}^R)^\dagger
  \preccurlyeq
  (\partial_{v_m}^R)^\dagger\partial_{v_m}^R+I\).
  Applying the first fact again to the \(v_m\)'s gives
  \begin{align}
    \Hpair
    &\preccurlyeq
    \omega\Pi^L
    +
    \frac{3}{4\omega}
    \left(
    \sigma_1^2\Pi^R
    +
    \sum_m\sigma_m^2I
    \right)
    \nonumber\\
    &=
    \left(
    \omega+\frac{3\sigma_1^2}{4\omega}
    \right)\Pi^L
    +
    \frac{3\norm{\xi_c}_F^2}{4\omega}I
    \nonumber\\
    &\preccurlyeq
    \left(
    \omega+\frac{3\sigma_1^2}{4\omega}
    \right)\Pi^L
    +
    \frac{3\hatp(1-\hatp)n^2}{4\omega}I.
    \label{eq:Hpair-upper}
  \end{align}
  Here for the first equality we used that \(\Pi^L=\Pi^R\) on \(\cM\) and \(\sum_m\sigma_m^2=\norm{\xi_c}_F^2\). Then for the second inequality we use the fact that \(\norm{\xi_c}_F^2\leq\norm{\xi}_F^2 = \sum_{i,j}(A_{ij}-\hatp)^2 = |E|(1-\hatp)^2+(n^2-|E|)\hatp^2\). Moreover, using \(|E|=\hatp n^2\) this becomes \(\norm{\xi}_F^2 = \hatp(1-\hatp)n^2\).

  Because \(\Hpair\) is tridiagonal with zeros on the diagonal, the spectrum is symmetric. In particular, let \(W=\sum_k(-1)^kP_k\), then \(W\Hpair W=-\Hpair\) and \(W\Pi^L W=\Pi^L\). Conjugating \cref{eq:Hpair-upper} by \(W\) gives
  \begin{equation}
    \Hpair
    \succcurlyeq
    -\left(
    \omega+\frac{3\sigma_1^2}{4\omega}
    \right)\Pi^L
    -
    \frac{3\hatp(1-\hatp)n^2}{4\omega}I.
  \end{equation}
  Taking \(\omega=\hatp n/4\) proves \cref{eq:Hpair-lower}.

  Combining this with \cref{lem:sl} gives
  \begin{equation}
    \Hxi
    \succcurlyeq
    -\left(
    \frac{\hatp n}{4}
    +
    \frac{3\sigma_1^2}{\hatp n}
    +
    3(\sigma_1+\Delta)
    \right)\Pi^L
    -
    3(1-\hatp)nI.
  \end{equation}
  Using \cref{eq:level-cost} and the hypothesis \(\sigma_1+\Delta\leq\hatp n/100\), we obtain
  \begin{align*}
    \Hxi
    &\succcurlyeq
    -\left(
    \frac{\hatp}{2}
    +
    \frac{6\sigma_1^2}{\hatp n^2}
    +
    \frac{6(\sigma_1+\Delta)}{n}
    \right)
    \Pi^L(n+1-\Pi^L)
    -3nI \\
    &\succcurlyeq
    -\frac{3\hatp}{4}
    \Pi^L(n+1-\Pi^L)
    -3nI,
  \end{align*}
  which proves \cref{eq:OI-def}.
\end{proof}

\section{From the quadratic model to the singlet spectrum}\label{section_boson}

In this section we determine the bottom of the singlet spectrum: the two lowest
eigenvalues of $H$, and hence its gap, together with a lower bound on the ground state's
overlap with the Lieb-Mattis state. We start from the decomposition of
\cref{prop_pair_decomp},
\[
H=\Hmf+\Hpair+\Hsl
\]
Relative to the Lieb-Mattis state, the mean-field Hamiltonian $\Hmf$ charges level $k$ the
energy $\hatp k(n-k+1)$, the pair term $\Hpair$ changes the level by one, and the
same-level term $\Hsl$ acts within each level.

Intuitively, near the Lieb-Mattis state the level cost is almost linear,
$\hatp k(n-k+1)=\hatp nk+O(k^{2})$, and $\Hpair$ creates or removes one excitation on each
sublattice. At low levels $H$ should therefore behave like a quadratic bosonic Hamiltonian,
whose pairing strengths are the singular values of $\xi_c$ and which a Bogoliubov
transformation solves exactly~\cite{HolsteinPrimakoff1940,Anderson1952,Bogoliubov1947}. We
call $\Hcov\coloneq\Hmf+\Hpair$ the \emph{comparison operator}; it omits the same-level term,
which we absorb separately.

This picture is accurate only while $k$ is small. We make it precise at low levels by
embedding the pair model into bosonic Fock space (\cref{section_compression}), absorb
$\Hsl$ into a small change of the mean-field coefficient (\cref{section_shift}), join the
low levels to the high-level bound of \cref{sec:rep-theory} by IMS localization
(\cref{section_transfer}), and finally project a truncated squeezed state into the
singlet space for a matching upper bound (\cref{section_wick}).

\subsection{The Bogoliubov Hamiltonian}\label{section_bogo}

We start by defining the bosonic model. Let
$\cF\coloneq\ell^2(\mathbb Z_{\ge0}^{L\sqcup R})$ be the bosonic Fock space over the two
sublattices, with vacuum $\ket{\emptyset,\emptyset}$, canonical operators
$a_i,a_i^{\dagger}$ on $L$ and $b_j,b_j^{\dagger}$ on $R$, and number operators
$\cN_L\coloneq\sum_ia_i^{\dagger}a_i$ and $\cN_R\coloneq\sum_jb_j^{\dagger}b_j$. For
$u\in\mathbb R^{L}$ and $v\in\mathbb R^{R}$, let $a_u\coloneq\sum_iu(i)a_i$ and
$b_v\coloneq\sum_jv(j)b_j$.

At low levels we make the substitutions
\[
  \partial_{u_m}^L\rightsquigarrow a_{u_m},
  \qquad
  \partial_{v_m}^R\rightsquigarrow b_{v_m},
  \qquad
  \Gamma_k\rightsquigarrow\tfrac12,
  \qquad
  k(n-k+1)\rightsquigarrow nk
\]
Thus the exact level cost is linearized around $k=0$, while the SVD
$\xi_c=\sum_m\sigma_m u_m v_m^\top$ decomposes the pair term into independent left-right
modes. \cref{lemma_embed} below makes the first two substitutions precise.

\begin{definition}[Bogoliubov Hamiltonian]
Let
\[
  \omega\coloneq\frac{\hatp n}{2},\qquad g_m\coloneq\frac{\sigma_m}{2},\qquad
  \omega_m\coloneq\sqrt{\omega^{2}-g_m^{2}}
\]
The \emph{Bogoliubov Hamiltonian} is the operator on $\cF$
\[
  \HB\coloneq\Emf+\omega\,(\cN_L+\cN_R)-\sum_{m=1}^{r}g_m\bigl(a_{u_m}b_{v_m}+a_{u_m}^{\dagger}b_{v_m}^{\dagger}\bigr)
\]
where $\Emf$ is the Lieb-Mattis baseline of~\cref{eq:Hmf-final}. The
\emph{equal-number subspace} is $\cF_0\coloneq\{\cN_L=\cN_R\}$.
\end{definition}

The operator $\HB$ preserves $\cN_L-\cN_R$, so $\cF_0$ is invariant. It is the relevant
subspace because level $k$ of the singlet space has $k$ excitations on each side.

When $\sigma_1<\hatp n$, every $\omega_m$ is positive and $\HB$ is diagonalized by a
Bogoliubov transformation. Let
\[
  \omega_\star\coloneq\min_m\omega_m=\sqrt{\omega^{2}-\sigma_1^{2}/4},
  \qquad
  t_m\coloneq\frac{g_m}{\omega+\omega_m}
\]
The ground state is the \emph{squeezed vacuum}
\begin{equation}\label{eq_squeezed}
  \Omega\coloneq\exp\Bigl(\sum_m t_m a_{u_m}^{\dagger}b_{v_m}^{\dagger}\Bigr)\ket{\emptyset,\emptyset}
\end{equation}
Since $t_m\le\sigma_m/(\hatp n)<1$, $\Omega$ has finite norm. Double centering is an
orthogonal projection for the Frobenius inner product, and $0\le A_{ij}\le1$ gives
$\norm{\xi}_F^{2}=\sum_{ij}A_{ij}^{2}-\hatp^{2}n^{2}\le\hatp(1-\hatp)n^{2}$, so
\begin{equation}\label{eq_frob}
  \sum_m\sigma_m^{2}=\norm{\xi_c}_F^{2}\le\hatp(1-\hatp)n^{2}
\end{equation}
and therefore
\begin{equation}\label{eq_t_bounds}
  \norm t^{2}\coloneq\sum_m t_m^{2}\le\frac{1-\hatp}{\hatp},\qquad
  t_{\max}^{2}\coloneq\max_m t_m^{2}\le\Bigl(\frac{\sigma_1}{\hatp n}\Bigr)^{2}
\end{equation}

\begin{lemma}[Bogoliubov diagonalization]\label{lemma_bogo}
  Suppose $\sigma_1<\hatp n$. Then $\Omega\in\cF_0$ and $\HB\Omega=\Estar\Omega$, where
  \begin{equation}\label{eq_bogo_inputs}
    \lambda_0(\HB|_{\cF_0})=\Estar\coloneq\Emf+\sum_m(\omega_m-\omega),
    \qquad
    \lambda_1(\HB|_{\cF_0})\ \ge\ \Estar+2\omega_\star
  \end{equation}
\end{lemma}

The transformation decouples $\HB$ into modes of frequencies $\omega_m$. The gap within
$\cF_0$ is $2\omega_\star$ because an admissible excitation must preserve equal particle
numbers on the two sides. The state $\Omega$ can be viewed as the Lieb-Mattis vacuum
dressed by correlated excitation pairs. Although this dressing lowers the energy by
order $n$, the bound $\norm t^{2}=O(1)$ shows that it involves only $O(1)$ excitation
pairs on average, as in the mean-field Bose gas~\cite{LNSS2015}. The stronger bound on
$t_{\max}$ will later control the error when $\Omega$ is projected back to the hard-core
spin space. The proof is in \cref{appendix_bogo}.

\subsection{Embedding the singlet space and bounding low levels}\label{section_compression}

The Bogoliubov Hamiltonian is solvable; the question is how faithfully it represents
$\Hcov$ at low levels. We start by placing $\cM$ inside $\cF$. We split
$\cF=\cF_L\otimes\cF_R$ with $\cF_L\coloneq\ell^2(\mathbb Z_{\ge0}^{L})$ and
$\cF_R\coloneq\ell^2(\mathbb Z_{\ge0}^{R})$. Set functions enter Fock space through the map
\begin{align*}
  \cB^L\colon\bigoplus_{k=0}^{n}\C^{\binom{[n]}{k}}&\to\cF_L\\
  f&\mapsto\sum_{S}f(S)\ket{\one_S}
\end{align*}
where $\ket{\one_S}$ has one boson on each site of $S$. In particular, subsets specify the
occupied sites, just as they specify the flipped spins under $\cV^L$. We define
$\cB^R\colon\bigoplus_k\C^{\binom{[n]}{k}}\to\cF_R$ in the same way. Both maps are
isometries onto the hard-core subspaces, that is, onto the vectors with at most one boson per
site, and we identify set functions with their images. In particular the harmonic spaces
$\cH_k^L$ and $\cH_k^R$ become subspaces of $\cF_L$ and $\cF_R$. We let
\[
  \cH^L\coloneq\bigoplus_{k=0}^{\lfloor n/2\rfloor}\cH_k^L\subseteq\cF_L,
  \qquad
  \cH^R\coloneq\bigoplus_{k=0}^{\lfloor n/2\rfloor}\cH_k^R\subseteq\cF_R
\]
with orthogonal projections $\Pi_{\cH^L}$, $\Pi_{\cH^R}$, and we denote by $\Pi_\cM$ the
orthogonal projection of $\cF$ onto $\cM$. For $u\in\mathbb R^L$ and $v\in\mathbb R^R$ we
use the weighted down operators $\partial_u^L$ and $\partial_v^R$ of \cref{ss:decomposition}.
They map $\cH_k^L$ into $\cH_{k-1}^L$ and $\cH_k^R$ into $\cH_{k-1}^R$
(\cref{ss:decomposition}), so they preserve $\cH^L$ and $\cH^R$. Their adjoints
$(\partial_u^L)^{\dagger}$, $(\partial_v^R)^{\dagger}$ are taken on the whole set-function
space, and need not preserve $\cH^L$, $\cH^R$. The level of a set function is measured by
$\Pi^L$ and $\Pi^R$, and functions of the level act through $f(\Pi^L)$, $f(\Pi^R)$, as in
\cref{sec:rep-theory}.

\begin{lemma}[Fock embedding of the pair model]\label{lemma_embed}
  Let $u\in\mathbb R^L$ and $v\in\mathbb R^R$.
  \begin{enumerate}[label=\textup{(\roman*)}]
    \item\label{lemma_embed_item1} $\cH^L$ consists of the hard-core vectors $\psi\in\cF_L$ with
          $\sum_{i\in L}a_i\psi=0$, and $\cH^R$ of the hard-core vectors
          $\psi\in\cF_R$ with $\sum_{j\in R}b_j\psi=0$. Moreover
          $\cM=(\cH^L\otimes\cH^R)\cap\cF_0$. On hard-core vectors $\cN_L$ acts as $\Pi^L$
          and $\cN_R$ as $\Pi^R$, so $\cN_L=\cN_R=\Pi^L$ on $\cM$.
    \item\label{lemma_embed_item2} On hard-core vectors $a_u$ acts as $\partial_u^L$, and $a_u^{\dagger}$ followed by
          the projection onto hard-core vectors acts as $(\partial_u^L)^{\dagger}$. In
          particular $a_u$ maps $\cH^L$ into itself, and
          $\Pi_{\cH^L}a_u^{\dagger}|_{\cH^L}=\Pi_{\cH^L}(\partial_u^L)^{\dagger}|_{\cH^L}$;
          likewise on $R$ with $b_v$ and $\partial_v^R$.
    \item\label{lemma_embed_item3} The operator $a_ub_v$ maps $\cM$ into itself, with
          \[
            a_ub_v|_\cM=\partial_u^L\otimes\partial_v^R|_\cM,\qquad
            \Pi_\cM a_u^{\dagger}b_v^{\dagger}|_\cM
            =\Pi_\cM\bigl((\partial_u^L)^{\dagger}\otimes(\partial_v^R)^{\dagger}\bigr)|_\cM
          \]
          For all functions $f,g$ on $\mathbb Z_{\ge0}$, the operator $f(\cN_L)g(\cN_R)$
          commutes with $\Pi_\cM$ and acts on $\cM$ as $f(\Pi^L)g(\Pi^L)$.
  \end{enumerate}
\end{lemma}

\begin{proof}
We first prove~(i). On hard-core vectors $a_i$ deletes $i$ from the set of occupied sites,
so it acts as $\partial_i^L$, and $\sum_ia_i$ acts as the down operator $\partial^L$
by the definition of the down operator in \cref{ss:harmonics}. The hard-core vectors annihilated by $\sum_ia_i$ are therefore
the sums of harmonic functions, and $\cH_k^L=0$ for $k>n/2$. This gives $\cH^L$, and $\cH^R$
follows in the same way with $b_j$ acting as $\partial_j^R$. Intersecting
$\cH^L\otimes\cH^R=\bigoplus_{k,k'}\cH_k^L\otimes\cH_{k'}^R$ with $\cF_0$ keeps $k=k'$.

Next, for~(ii), $a_u^{\dagger}$ adds a boson at each site with weight $u(i)$. Projecting
onto hard-core vectors discards the doubly occupied terms, which is the action of the
pullback of $\sum_iu(i)S_i^-$. Since $\partial_u^L$ maps $\cH^L$ into itself, so does $a_u$,
and compressing the adjoint to $\cH^L$ gives the second identity.

Finally, for~(iii), $a_ub_v$ lowers both numbers by one and preserves
$\cH^L\otimes\cH^R$, hence $\cM$, and the adjoint identity follows by compressing~(ii) on
both sides. The subspace $\cM$ is spanned by joint eigenvectors of $\cN_L$ and $\cN_R$, so
$\Pi_\cM$ commutes with $f(\cN_L)g(\cN_R)$.
\end{proof}

Thus $\cM$ consists of bosonic states satisfying three constraints: at most one excitation
per site, annihilation by the collective operator on each side, and equal excitation counts
on the two sides. We call these the \emph{hard-core}, \emph{harmonic}, and
\emph{equal-number constraints}. Creation, unlike
annihilation, may doubly occupy a site and must be projected back, and the projection only
weakens it. For a unit vector $u\in\mathbb R^L$ and a hard-core vector $\psi$, the vector
$(\partial_u^L)^{\dagger}\psi$ is the projection of $a_u^{\dagger}\psi$ by part~(ii), so the
canonical commutation relation gives
$\norm{(\partial_u^L)^{\dagger}\psi}^{2}\le\norm{a_u^{\dagger}\psi}^{2}=\norm{\partial_u^L\psi}^{2}+\norm\psi^{2}$,
that is,
$\partial_u^L(\partial_u^L)^{\dagger}\preccurlyeq(\partial_u^L)^{\dagger}\partial_u^L+I$. Summing
over the orthonormal $u_1,\dots,u_{r}$ and using Bessel's inequality
$\sum_m a_{u_m}^{\dagger}a_{u_m}\preccurlyeq\cN_L$ gives, on hard-core vectors,
\begin{equation}\label{eq_mode_sums}
  \sum_{m=1}^{r}(\partial_{u_m}^L)^{\dagger}\partial_{u_m}^L\preccurlyeq\cN_L,
  \qquad
  \sum_{m=1}^{r}\partial_{u_m}^L(\partial_{u_m}^L)^{\dagger}\preccurlyeq\cN_L+r
\end{equation}
and likewise on $R$ with $\partial_{v_m}^R$ and $\cN_R$. On level $k$, annihilation therefore has
total strength $\sqrt k$ and creation $\sqrt{k+r}$. These two inequalities are all that
replaces the canonical commutation relations in the estimates below. They enter through the
following bound, which is Cauchy-Schwarz combined with~\cref{eq_mode_sums}, for an operator
that lowers the level by exactly one.

\begin{lemma}[Band estimate]\label{lemma_band}
Let $f\colon\mathbb Z_{\ge0}\to\mathbb R$ satisfy $|f(k)|\le f_{\max}$, and let
\[
\cA_f\coloneq\sum_{m=1}^{r}\sigma_m\,\bigl(\partial_{u_m}^L\otimes\partial_{v_m}^R\bigr)\,f(\Pi^L)
\]
on $\cM$ be the annihilation half of the pair term~\cref{eq:Hpair-explicit}, with the
level-dependent coefficient inserted at the level the annihilation acts from. Then, for
$\varphi\in\cM_k$ and $\psi\in\cM_{k+1}$,
\[
\bigl|\ip{\varphi}{\cA_f\psi}\bigr|
\ \le\ \sigma_1f_{\max}\sqrt{(k+r)(k+1)}\,\norm\varphi\,\norm\psi
\]
The same bound holds on $\cF_0$ with $a_{u_m},b_{v_m}$ in place of
$\partial_{u_m}^L,\partial_{v_m}^R$, with $k$ the sector number.
\end{lemma}

\begin{proof}
Since $f(\Pi^L)\psi=f(k+1)\psi$ and
$\partial_{u_m}^L\otimes\partial_{v_m}^R=(\partial_{u_m}^L\otimes I)(I\otimes\partial_{v_m}^R)$,
moving the left factor across the inner product as its adjoint and applying
Cauchy-Schwarz, first in the tensor product and then in $m$, gives
\begin{align*}
\bigl|\ip{\varphi}{\cA_f\psi}\bigr|
&=\Bigl|f(k+1)\sum_{m=1}^{r}\sigma_m
  \ip{((\partial_{u_m}^L)^{\dagger}\otimes I)\varphi}{(I\otimes\partial_{v_m}^R)\psi}\Bigr|\\
&\le\sigma_1f_{\max}
  \Bigl(\sum_{m=1}^{r}\norm{((\partial_{u_m}^L)^{\dagger}\otimes I)\varphi}^{2}\Bigr)^{1/2}
  \Bigl(\sum_{m=1}^{r}\norm{(I\otimes\partial_{v_m}^R)\psi}^{2}\Bigr)^{1/2}
\end{align*}
The two sums are the quadratic forms of
$\bigl(\sum_m\partial_{u_m}^L(\partial_{u_m}^L)^{\dagger}\bigr)\otimes I$ at $\varphi$ and of
$I\otimes\bigl(\sum_m(\partial_{v_m}^R)^{\dagger}\partial_{v_m}^R\bigr)$ at $\psi$.
By~\cref{eq_mode_sums} these operators are at most $(\cN_L+r)\otimes I$ and $I\otimes\cN_R$,
and on $\cM$ we have $\cN_L=\cN_R=\Pi^L$ by \cref{lemma_embed_item1} of \cref{lemma_embed}. Hence the first sum is at
most $(k+r)\norm\varphi^{2}$ and the second at most $(k+1)\norm\psi^{2}$, which proves the
bound. On $\cF_0$ the same computation applies, since the canonical commutator
$[a_u,a_u^{\dagger}]=I$ for unit $u$ and Bessel's inequality
$\sum_m a_{u_m}^{\dagger}a_{u_m}\preccurlyeq\cN_L$ give the two inequalities of~\cref{eq_mode_sums}
for the free modes.
\end{proof}

By \cref{lemma_embed}, every matrix element of $\Hcov$ is a compressed matrix element of
a free-space operator, but not of $\HB$ itself, which has already frozen $\Gamma_k$ and
linearized the level cost. We therefore interpose an \emph{intermediate free-space operator} $\widetilde H$ that lives
on $\cF$ like $\HB$ but keeps the exact finite-$n$ diagonal and the exact coefficient
$\Gamma_k$ like $\Hcov$. Both are extended past the physical range $k\le n/2$, by
$\Gamma_k\coloneq\tfrac12$ for $k>n/2$ and
\[
  d(k)\coloneq\begin{cases}k(n-k+1),&k\le n/2,\\ \tfrac n2(k+1),&k>n/2.\end{cases}
\]
The function $d$ is increasing, agrees with the physical level cost for $k\le n/2$, and
satisfies $d(k)\ge\tfrac n2k$ everywhere. On $\cF$ let
\[
  \widetilde H\coloneq\Emf+\tfrac{\hatp}2\bigl(d(\cN_L)+d(\cN_R)\bigr)-\sum_m\sigma_m\bigl(a_{u_m}b_{v_m}\,\Gamma_{\cN_R}+\Gamma_{\cN_R}\,a_{u_m}^{\dagger}b_{v_m}^{\dagger}\bigr)
\]
with $d(\cN_L)$, $\Gamma_{\cN_R}$ acting by functional calculus and $\Gamma_{\cN_R}$
placed next to the annihilation operator on both sides, matching~\cref{eq:Hpair-explicit}.

The comparison then has two steps: $\Hcov=\Pi_\cM\widetilde H\Pi_\cM$ exactly, and
$\widetilde H\approx\HB$ only on levels $k<4\kcut$. Replacing $\partial_u^L$ by $a_u$ costs
nothing, since the compression absorbs it exactly. The only errors come from the two frozen
quantities, the level cost, with $k(n-k+1)-nk=-k(k-1)$, and the Clebsch-Gordan coefficient,
with $\Gamma_k-\frac12=O(1/n)$ at low level: on vectors of $\cF_0$ with $\cN_L\le n/2$,
\begin{equation}\label{eq_HB_error}
  \widetilde H-\HB=-\hatp\,\cN_L(\cN_L-1)-(\cA_f+\cA_f^{\dagger}),
  \qquad
  \cA_f\coloneq\sum_m\sigma_m\,a_{u_m}b_{v_m}f(\cN_R),\qquad f(k)\coloneq\Gamma_k-\frac12
\end{equation}
Both errors grow with the level, which is why the comparison stops at $4\kcut$ and transfers
the closed-form Bogoliubov bounds~\cref{eq_bogo_inputs} of \cref{lemma_bogo} to $\Hcov$
only there.

\begin{lemma}[Compression]\label{lemma_CB}
  \begin{enumerate}[label=\textup{(\roman*)}]
    \item\label{lemma_CB_item1} $\Hcov=\Pi_{\cM}\widetilde H\Pi_{\cM}$ on $\cM$, and
          $\lambda_j(\Hcov)\ge\lambda_j(\widetilde H|_{\cF_0})$ for every $j$.
    \item\label{lemma_CB_item2} Let $\kcut\ge1$ with $16\kcut\le n$, suppose $\sigma_1<\hatp n$, and let
          \[
            \etaB(\kcut)\coloneq16\kcut^{2}+5\sigma_1\sqrt{\kcut/n}
          \]
          Then $|\ip{\psi}{(\widetilde H-\HB)\psi}|\le\etaB(\kcut)\norm\psi^{2}$ for every
          $\psi\in\cF_0\cap\{\cN_L<4\kcut\}$, and consequently
          \[
            \lambda_0(\Hcov|_{<4\kcut})\ge\Estar-\etaB(\kcut),
            \qquad
            \lambda_1(\Hcov|_{<4\kcut})\ge\Estar+2\omega_\star-\etaB(\kcut)
          \]
          where $\Hcov|_{<4\kcut}$ is the compression of $\Hcov$ to the levels $k<4\kcut$.
  \end{enumerate}
\end{lemma}

\begin{proof}
  \cref{appendix_low}. Part~(i) holds term by term by \cref{lemma_embed_item3} of \cref{lemma_embed} and
  Poincar\'e separation; part~(ii) uses $|\Gamma_k-\tfrac12|\le1/n$ below level $4\kcut$ and
  needs no delocalization of the singular vectors $u_m,v_m$.
\end{proof}

\subsection{Absorbing the same-level error}\label{section_shift}

\cref{lem:sl} compares $H$ with $\Hcov$ up to the same-level term. Below we use its loss
in the weaker form $\csl\Pi^L$, with $\csl\coloneq3\sigma_1+9\Delta\ge3(\sigma_1+\Delta)$.
This loss grows linearly with the level, while the transfer argument below needs a
comparison on all of $\cM$ with a loss independent of the level. The mean-field term of
$\Hcov$ grows faster, like $\hatp\,k(n+1-k)$ on level $k$, so we absorb the loss by slightly
decreasing the density in that term. What remains is a scalar shift, independent of the
level, and it costs little in the Bogoliubov energy.

\begin{lemma}[Shifted comparison]\label{lemma_shift}
  Let $\Hcov(\hatp)$ denote the comparison operator $\Hcov$ with density $\hatp$ in its
  mean-field term. Let $\hatp'\coloneq\hatp-2\csl/n$, and let $\Estar(\hatp')$ and
  $\omega_\star(\hatp')$ denote the Bogoliubov quantities of~\cref{eq_bogo_inputs} at the
  density $\hatp'$.
  \begin{enumerate}[label=\textup{(\roman*)},nosep]
    \item\label{lemma_shift_item1} For every $\psi\in\cM$,
      $\ip{\psi}{H\psi}\ge\ip{\psi}{\Hcov(\hatp')\psi}-\tfrac12\csl(n+2)\norm\psi^{2}$.
    \item\label{lemma_shift_item2} If $\sigma_1+\Delta\le\hatp n/100$, then $\sigma_1<\hatp'n$,
      $2\omega_\star(\hatp')\ge\hatp n/2$ and
      \[
        \Estar(\hatp')-\tfrac12\csl(n+2)\ \ge\ \Estar-\frac{72(\sigma_1+\Delta)}{\hatp}
      \]
  \end{enumerate}
\end{lemma}

\begin{proof}
  For~(i), the two comparison operators differ only in the mean field $-\hatp s(s+1)$, and
  $s(s+1)=\tfrac n2(\tfrac n2+1)-\Pi^L(n+1-\Pi^L)$, so
  \[
    \Hcov(\hatp)-\Hcov(\hatp')
    =\frac{2\csl}{n}\Bigl(\Pi^L(n+1-\Pi^L)-\frac n2\Bigl(\frac n2+1\Bigr)\Bigr)
  \]
  \cref{lem:sl} gives $H\succcurlyeq\Hcov(\hatp)-\csl\Pi^L$, and
  $\tfrac{2\csl}{n}\Pi^L(n+1-\Pi^L)\succcurlyeq\csl\Pi^L$ by~\cref{eq:level-cost}. Combining the three
  gives $H\succcurlyeq\Hcov(\hatp')-\tfrac{2\csl}{n}\cdot\tfrac n2(\tfrac n2+1)=\Hcov(\hatp')-\tfrac12\csl(n+2)$.

  For~(ii), $\csl\le9(\sigma_1+\Delta)\le9\hatp n/100$ gives $\hatp'\ge\tfrac45\hatp$, so $\sigma_1<\hatp'n$
  and $2\omega_\star(\hatp')=\sqrt{(\hatp'n)^{2}-\sigma_1^{2}}\ge\hatp'n-\sigma_1\ge\hatp n/2$.
  Let $F(\omega)\coloneq\sum_m\bigl(\omega-\sqrt{\omega^{2}-g_m^{2}}\bigr)$, so that
  $\Estar=\Emf-F(\hatp n/2)$ by~\cref{eq_bogo_inputs}. As $\Emf$ is linear in the density,
  with $\Emf(\hatp')=\Emf+\tfrac12\csl(n+2)$, the left-hand side of the claim equals
  $\Estar-\bigl(F(\hatp'n/2)-F(\hatp n/2)\bigr)$. On $[\hatp'n/2,\hatp n/2]$ we have
  $\omega\ge\hatp n/4$ and $\omega_m\ge\omega/2$, so
  \[
    |F'(\omega)|=\sum_m\frac{g_m^{2}}{\omega_m(\omega+\omega_m)}\le\frac{2\norm{\xi_c}_F^{2}/4}{\omega^{2}}\le\frac8{\hatp}
  \]
  by~\cref{eq_frob}. The interval has length $\csl$, so the difference is at most
  $8\csl/\hatp\le72(\sigma_1+\Delta)/\hatp$.
\end{proof}

\subsection{Combining low- and high-level bounds}\label{section_transfer}

A state supported across the boundary between the low levels, where the spectrum is now
understood accurately, and the high levels, where it is bounded below, could in principle
gain energy from the coupling between them. We therefore split a state into low- and
high-level parts with the IMS localization identity~\cite{CFKS1987}. Because $H$ changes
the level by at most one, the energy lost at the split is controlled by how quickly the
cutoffs vary, and it is small when they vary on a scale $\kcut\gg1$.

For $\kcut\ge1$, let
$\phi(k)\coloneq\tfrac\pi2\min\bigl(1,\max(0,(k-2\kcut)/(2\kcut))\bigr)$ be the clamped linear
ramp from level $2\kcut$ to level $4\kcut$, and let $\chi_{\mathrm{in}}\coloneq\cos\phi$ and
$\chi_{\mathrm{out}}\coloneq\sin\phi$ as functions of the level. These satisfy
$\chi_{\mathrm{in}}^{2}+\chi_{\mathrm{out}}^{2}=1$ and commute with each other and with
$\Hmf$. The IMS formula, in the form for level-banded operators on Fock
space~\cite[Prop.~6.1]{LNSS2015}, splits the energy into an inner and an outer piece plus a
double-commutator error:
\begin{equation}\label{eq_bg_ims}
  \ip{\psi}{H\psi}
  =\ip{\chi_{\mathrm{in}}\psi}{H\chi_{\mathrm{in}}\psi}
  +\ip{\chi_{\mathrm{out}}\psi}{H\chi_{\mathrm{out}}\psi}
  +\tfrac12\!\!\sum_{\chi\in\{\chi_{\mathrm{in}},\chi_{\mathrm{out}}\}}\!\!\ip{\psi}{[\chi,[\chi,H]]\psi}
\end{equation}

\begin{lemma}[The IMS error of the explicit cutoffs]\label{lemma_ims}
  Let $\kcut\ge1$, and let $\chi_{\mathrm{in}},\chi_{\mathrm{out}}$ be the cutoffs above. Then
  for every $\psi\in\cM$
  \[
    \Bigl|\tfrac12\!\!\sum_{\chi\in\{\chi_{\mathrm{in}},\chi_{\mathrm{out}}\}}\!\!\ip{\psi}{[\chi,[\chi,H]]\psi}\Bigr|
    \le2\,\etaIMS(\kcut)\,\norm\psi^{2},
    \qquad
    \etaIMS(\kcut)\coloneq\frac{2\sigma_1\sqrt n}{\kcut^{3/2}}
  \]
\end{lemma}

\begin{proof}
The double commutator $[\chi,[\chi,H]]$ has entries $(\chi(k')-\chi(k))^{2}H_{k'k}$,
as in the proof of~\cite[Prop.~6.1]{LNSS2015}. Since
$\Hmf$ and $\Hsl$ commute with functions of the level and $\Hpair$ changes the level by
exactly one, the double commutator is supported on the band $2\kcut\le k\le4\kcut$. For $H$
symmetric and level-banded, and $\chi$ a function of the level,
\[
\bigl|\ip{\psi}{[\chi,[\chi,H]]\psi}\bigr|\ \le\ 2\max_k\Bigl[(\chi(k+1)-\chi(k))^{2}\,\norm{P_{k+1}HP_k}\Bigr]\,\norm\psi^{2}
\]
\cref{lemma_band} on $\cM$, with $f(k)=\Gamma_k$ and $f_{\max}=\sqrt{3/4}$ (the value of
$\Gamma_k$ at $s_k=0$), bounds $|\ip{\varphi}{\Hpair\psi}|$ by
$\sigma_1\sqrt{3/4}\sqrt{(k+r)(k+1)}\norm\varphi\norm\psi$
for $\varphi$ on level $k$ and $\psi$ on level $k+1$. In the band, $k+1\le4\kcut$ and
$k+r\le n+r$, and the increments of $\cos\phi$ and $\sin\phi$ are at most $\pi/(4\kcut)$.
Summing over the two cutoffs, the left-hand side of the bound in \cref{lemma_ims} is
therefore at most
$2(\pi/(4\kcut))^{2}\sigma_1\sqrt{3/4}\sqrt{(n+r)4\kcut}\le(\pi^{2}\sqrt6/8)\,\sigma_1\sqrt n\,\kcut^{-3/2}\le2\etaIMS(\kcut)$,
using $r\le n$.
\end{proof}

\begin{theorem}[Transfer]\label{thm_transfer}
  Let $\kcut\ge1$ with $16\kcut\le n$, and let $\chi_{\mathrm{in}},\chi_{\mathrm{out}}$ be the cutoffs
  of \cref{lemma_ims}. Let $\lambda_0\le\lambda_1$ be the two lowest eigenvalues of $H$
  on $\cM$, $\gamma=\lambda_1-\lambda_0$, and $\Psi_0\in\cM$ a unit ground vector of $H$.
  Let $H'$ be a symmetric operator on $\cM$, and
  let $\Eout,E',U\in\mathbb R$, $\gamma',C'\ge0$, and $\Phi\in\cM$ be given with:
  \begin{enumerate}[label=\textup{(\alph*)}]
    \item \emph{(high-level lower bound)} $\ip{\psi}{H\psi}\ge\Eout\norm{\psi}^{2}$ for every
          $\psi\in\cM$ supported on levels $\ge2\kcut$, that is, with $P_k\psi=0$ for all
          $k<2\kcut$;
    \item \emph{(global lower comparison)} $\ip{\psi}{H\psi}\ge\ip{\psi}{H'\psi}-C'\norm{\psi}^{2}$
          for every $\psi\in\cM$;
    \item \emph{(low-level eigenvalue bounds)} $\lambda_0(H'|_{<4\kcut})\ge E'$ and
          $\lambda_1(H'|_{<4\kcut})\ge E'+\gamma'$;
    \item \emph{(trial-state upper bound)} $\norm\Phi=1$ and $\ip{\Phi}{H\Phi}\le U$.
  \end{enumerate}
  Then
  \[
    \lambda_1\ \ge\ \min\bigl(E'+\gamma'-C',\ \Eout\bigr)-2\etaIMS(\kcut),
    \qquad
    \min\bigl(E'-C',\ \Eout\bigr)-2\etaIMS(\kcut)\ \le\ \lambda_0\ \le\ U
  \]
  and, if $\gamma>0$, for every unit $\psi\in\cM$,
  \[
    |\ip{\psi}{\Psi_0}|\ \ge\ |\ip{\psi}{\Phi}|-\sqrt{\frac{\ip{\Phi}{H\Phi}-\lambda_0}{\gamma}}
  \]
\end{theorem}

\begin{proof}
  We first bound the two lowest eigenvalues. Let $\Psi'$ be a ground vector of
  $H'|_{<4\kcut}$ and $v\coloneq\chi_{\mathrm{in}}\Psi'$. If
  $\psi\perp v$ then $\chi_{\mathrm{in}}\psi\perp\Psi'$ is supported on levels $<4\kcut$, so (b)
  and (c) give
  $\ip{\chi_{\mathrm{in}}\psi}{H\chi_{\mathrm{in}}\psi}\ge(E'+\gamma'-C')\norm{\chi_{\mathrm{in}}\psi}^{2}$;
  for arbitrary $\psi$ the same argument gives $(E'-C')\norm{\chi_{\mathrm{in}}\psi}^{2}$. By
  (a), $\ip{\chi_{\mathrm{out}}\psi}{H\chi_{\mathrm{out}}\psi}\ge\Eout\norm{\chi_{\mathrm{out}}\psi}^{2}$.
  Combining the two regions with~\cref{eq_bg_ims} and \cref{lemma_ims}, every unit $\psi$
  has $\ip\psi{H\psi}\ge\min(E'-C',\Eout)-2\etaIMS(\kcut)$, and
  $\ip\psi{H\psi}\ge\min(E'+\gamma'-C',\Eout)-2\etaIMS(\kcut)$ if $\psi\perp v$. The
  Courant-Fischer theorem~\cite{HornJohnson2013} turns these into the lower bounds on
  $\lambda_0$ and $\lambda_1$, and (d) gives $\lambda_0\le\ip\Phi{H\Phi}\le U$.

  Finally, for the fidelity, write $\Phi=\beta\Psi_0+\Phi^{\perp}$ with $\Phi^{\perp}\perp\Psi_0$ and
  $|\beta|\le1$.
  Eckart's inequality~\cite{Eckart1930} gives
  $\norm{\Phi^{\perp}}^{2}\le(\ip\Phi{H\Phi}-\lambda_0)/\gamma$, and
  $|\ip\psi\Phi|\le|\beta|\,|\ip\psi{\Psi_0}|+\norm{\Phi^{\perp}}
  \le|\ip\psi{\Psi_0}|+\norm{\Phi^{\perp}}$ for every unit $\psi$.
\end{proof}

\subsection{A trial state with controlled energy and overlap}\label{section_wick}

\cref{thm_transfer} needs a unit vector in $\cM$ with energy near $\Estar$ and
nonzero Lieb-Mattis overlap. The squeezed vacuum $\Omega$ of~\cref{eq_squeezed} has both
properties in $\cF$ by \cref{lemma_bogo}; we project it into $\cM$ and bound the losses.
By level, it decomposes as
\[
\Omega=e^{\cP}\ket{\emptyset,\emptyset}=\sum_{k\ge0}\Omega_k,
\qquad
\cP\coloneq\sum_m t_m a_{u_m}^\dagger b_{v_m}^\dagger,
\qquad
\Omega_k\coloneq\frac{\cP^k\ket{\emptyset,\emptyset}}{k!}\in\{\cN_L=\cN_R=k\}
\]
Let $\ell\coloneq\lfloor n^{1/4}\rfloor+1$. We build the trial state from $\Omega$ in four
steps. We first truncate above level $\ell$, keeping
$\Omega_{\le\ell}\coloneq\sum_{k=0}^\ell\Omega_k$. We then impose the hard-core constraint,
replacing $\Omega_{\le\ell}$ by $\Omega_{\le\ell}|_{\hc}$, where $|_{\hc}$ projects onto
occupations at most one per site. Next we project each level onto $\cM_k$, which gives
\[
\Upsilon\coloneq\sum_{k=0}^\ell P_k\bigl(\Omega_k|_{\hc}\bigr)\in\cM
\]
where $P_k$ projects onto $\cM_k$ within the hard-core space. Finally we normalize,
$\Phi\coloneq\Upsilon/\norm\Upsilon$. Both projections preserve the vacuum, so
$\Upsilon(\emptyset,\emptyset)=1$. Since $\cU\ket{\emptyset,\emptyset}=\LM$ is the
Lieb-Mattis state of \cref{def_LM}, the Lieb-Mattis overlap of $\Phi$ is its vacuum
coefficient, $\ip\LM{\cU\Phi}=\Phi(\emptyset,\emptyset)$.

\begin{proposition}[The truncated squeezed state]\label{prop_wick}
Suppose that $n\ge(100/\hatp)^{4}$ and $\sigma_1\le\hatp^{3/2}n^{3/4}/4$. Then
$\Upsilon\neq0$, so $\Phi=\Upsilon/\norm{\Upsilon}$ is a well-defined unit vector of $\cM$,
and:
\begin{enumerate}[label=\textup{(\roman*)},nosep]
\item $\ip{\Phi}{\Hcov\Phi}\le\Estar+\eta_\Phi$, where
  $\eta_\Phi\coloneq70\sqrt n+320\,\sigma_1^{2}\hatp^{-3}n^{-1/4}$;
\item $\ip{\LM}{\cU\Phi}=\Phi(\emptyset,\emptyset)\ge e^{-(2/\hatp+1/2)}$;
\item $\ip{\Phi}{\Pi^L\Phi}\le3/\hatp$.
\end{enumerate}
\end{proposition}

Factorial decay of $\norm{\Omega_k}$ controls the truncation, and the maximal squeezing
$t_{\max}^2\le(\sigma_1/(\hatp n))^2$ controls both projection losses. The proof is in
\cref{appendix_wick}.

\section{Proof of the main theorems}\label{section_assembly}

The spectral theorem follows by combining the low-level comparison, the high-level
lower bound, and a trial state through \cref{thm_transfer}. The estimator
reduction, \cref{thm_reduction}, then turns the spectral bounds into the
algorithmic guarantee.

\subsection{Proof of \texorpdfstring{\cref{thm_spectral}}{Theorem B}}\label{section_proof_spectral}

The low-level comparison error grows with the cutoff $\kcut$, whereas the
localization error decreases and the high-level energy bound improves. The choice
$\kcut=\lfloor n^{1/4}\rfloor$ makes both errors $o(n)$ and separates the high-level
spectrum from the ground energy.

\begin{proof}[Proof of \cref{thm_spectral}]
  Fix $G\in\mathcal G(p,C)$. Then $\hatp\ge p$ and
  $\sigma_1,\Delta\le C\sqrt{n\log n}$. Choose $\kcut\coloneq\lfloor n^{1/4}\rfloor$ and set
  \[
    \eta_{\mathrm{tot}}
    \coloneq\etaB(\kcut)+2\etaIMS(\kcut)+\eta_\Phi+\frac{99(\sigma_1+\Delta)}{\hatp}
  \]
  Here $\eta_\Phi$ is from \cref{prop_wick}. The four summands are bounded,
  in order, by
  \[
    16\sqrt n+5C n^{1/8}\sqrt{\log n},\qquad
    12C n^{5/8}\sqrt{\log n},\qquad
    70\sqrt n+320C^2p^{-3}n^{3/4}\log n,\qquad
    198Cp^{-1}\sqrt{n\log n}
  \]
  Thus $\eta_{\mathrm{tot}}=O_{p,C}(n^{7/8})=o(n)$. Increase
  $n_0=n_0(p,C)$ so that, whenever $n\ge n_0$,
  \begin{gather*}
    \sigma_1+\Delta\le\frac{pn}{100},\qquad
    \sigma_1\le\frac{p^{3/2}n^{3/4}}4,\qquad
    n\ge\Bigl(\frac{100}p\Bigr)^{4},\\
    \kcut\ge\frac{20}p,\qquad 16\kcut\le n,\qquad
    \eta_{\mathrm{tot}}\le\frac p{32}e^{-(4/p+1)}n
  \end{gather*}
  The first three inequalities also imply the corresponding hypotheses with $\hatp$
  in place of $p$. Identify $\cS$ with $\cM$ through $\cU$, so that $H$ and
  $H_G|_{\cS}$ have the same spectrum. We now verify the transfer theorem's hypotheses.

  \emph{Low levels.} By \cref{lemma_shift_item1} of \cref{lemma_shift}, hypothesis~(b) holds with
  $H'=\Hcov(\hatp')$ and $C'=\tfrac12\csl(n+2)$. By \cref{lemma_shift_item2} of \cref{lemma_shift},
  $\sigma_1<\hatp'n$, so \cref{lemma_CB_item2} of \cref{lemma_CB} applies at the density $\hatp'$ and gives
  hypothesis~(c) with $E'=\Estar(\hatp')-\etaB(\kcut)$ and
  $\gamma'=2\omega_\star(\hatp')\ge\hatp n/2$. By \cref{lemma_shift_item2} of \cref{lemma_shift} again,
  \begin{equation}\label{eq_assembly_low}
    E'-C'\ge\Estar-\etaB(\kcut)-\frac{72(\sigma_1+\Delta)}{\hatp}
  \end{equation}

  \emph{High levels.} Let $\Eout\coloneq\Estar+2n$. On levels $2\kcut\le k\le n/2$,
  \cref{eq:OI-def} and~\cref{eq:Hmf-final} give
  \[
    H=\Hmf+\Hxi\ge\Emf+\frac{\hatp}{4}k(n+1-k)-3n\ge\Emf+\Bigl(\frac{\hatp\kcut}4-3\Bigr)n\ge\Eout
  \]
  Here $k(n+1-k)\ge2\kcut(n+1-2\kcut)\ge\kcut n$ throughout the indicated
  range. Since $\hatp\kcut\ge20$ and $\Emf\ge\Estar$, this verifies~(a).

  \emph{Trial state.} Let $\Phi$ be the trial state of \cref{prop_wick}. By its
  parts~(i) and~(iii) and~\cref{eq:sl-form-bound},
  \[
    \ip{\Phi}{H\Phi}\le\ip{\Phi}{\Hcov\Phi}+\csl\ip{\Phi}{\Pi^L\Phi}
    \le\Estar+\eta_\Phi+\frac{27(\sigma_1+\Delta)}{\hatp}\eqcolon U
  \]
  This verifies~(d).
  Part~(ii) of the same proposition gives, for the vacuum $\psi\in\cM$ with
  $\cU\psi=\LM$,
  \[
    \ip{\psi}{\Phi}\ge e^{-(2/\hatp+1/2)}\ge2f_0,
    \qquad f_0\coloneq\tfrac12e^{-(4/p+1)/2}
  \]

  By~\cref{eq_assembly_low}, $E'+\gamma'-C'\ge
  \Estar+\hatp n/2-\etaB(\kcut)-72(\sigma_1+\Delta)/\hatp$.
  Also $\Eout=\Estar+2n\ge\Estar+\hatp n/2$. \cref{thm_transfer}
  therefore gives
  \begin{equation}\label{eq_energy_window}
    |\lambda_0-\Estar|\le\eta_{\mathrm{tot}},
    \qquad
    \lambda_1\ge\Estar+\frac{\hatp n}2-\eta_{\mathrm{tot}}
  \end{equation}
  The first bound also gives a ground-energy estimate. Subtracting the two bounds
  gives $\gamma\ge\hatp n/2-2\eta_{\mathrm{tot}}\ge pn/4$.
  Finally, the trial-state bound and~\cref{eq_energy_window} yield
  $\ip{\Phi}{H\Phi}-\lambda_0\le2\eta_{\mathrm{tot}}$. Apply the fidelity
  conclusion of \cref{thm_transfer} to the vacuum to obtain
  \[
    |\ip{\LM}{\Psi_0}|
    \ge2f_0-\sqrt{\frac{2\eta_{\mathrm{tot}}}{pn/4}}
    \ge2f_0-f_0
    =f_0
    \qedhere
  \]
\end{proof}

All inputs to these bounds are computable from $\hatp$, $\sigma_1$, $\Delta$, and the
singular values of $\xi_c$. Thus membership in $\mathcal G(p,C)$ and the resulting
gap, overlap, and energy bounds can be checked in polynomial time for an individual graph.

\subsection{Proof of \texorpdfstring{\cref{thm_sampler}}{Theorem A}}\label{section_proof_sampler}

\cref{thm_spectral} supplies the gap and starting-state overlap required by
\cref{thm_reduction}. We first obtain a single-edge estimate, then amplify its
success probability and sum the estimates to obtain the energy.

\begin{proof}[Proof of \cref{thm_sampler}]
  Take $n\ge\max(n_0,3)$ and $c\coloneq n^2+1>|E|$. Then
  $Z=c+|E|+2n\le3n^2$. \cref{thm_spectral} gives the inputs
  $\gamma_0=pn/4$ and $a_{\LM,0}=f_0$ for \cref{thm_reduction};
  \cref{lemma_uncond_cert} gives $b_0=4n/(c+|E|)\le4/n$.
  Choose the smallest integer depth allowed by~\cref{eq_tchoice}. Since
  $Z/\gamma_0\le12n/p$ and $f_0$ depends only on $p$,
  \[
    t=O_p\bigl(n\log(n/\varepsilon)\bigr),\qquad
    2b_0t\le\frac{8t}{n}=O_p\bigl(\log(n/\varepsilon)\bigr)
  \]
  With the smallest sample count allowed by~\cref{eq_Nchoice}, one run therefore
  uses $N=\mathrm{poly}(n)(n/\varepsilon)^{O(1/p)}\delta^{-1}$ samples and takes
  $O(Nt)$ time after $O(n^2)$ preprocessing.

  \emph{Correlations.} Use failure probability $1/4$ for each run, make
  $N_{\mathrm{med}}\coloneq\lceil8\log(1/\delta)\rceil$ independent runs, and return their median. If the
  median is off by more than $\varepsilon$, at least half the runs failed.
  Hoeffding's inequality bounds this probability by $e^{-N_{\mathrm{med}}/8}\le\delta$.
  Thus the median trick~\cite{JerrumValiantVazirani1986} changes the confidence
  factor in the running time from $\delta^{-1}$ to
  $O(1+\log(1/\delta))$.

  \emph{Energy.} Since
  $\lambda_0=\sum_{ij\in E}A_{ij}\ev{\bS_i\!\cdot\bS_j}$, apply the correlation estimate to every edge
  with accuracy $\varepsilon/|E|$ and failure probability $\delta/|E|$, and sum the
  results weighted by $A_{ij}\le1$. A union bound gives total error at most $\varepsilon$ with probability
  at least $1-\delta$. Finally, $|E|\le n^2$ and
  $\log(|E|/\delta)\le2\log n+\log(1/\delta)$, so the additional factor of $|E|$
  and the smaller per-edge accuracy are absorbed into
  $\mathrm{poly}(n)(n/\varepsilon)^{O(1/p)}(1+\log(1/\delta))$.
\end{proof}

\section{Discussion}\label{section_discussion}

\cref{thm_reduction} holds on every connected balanced bipartite graph and converts
a singlet-sector gap and a starting-state overlap into a running time. Density and expansion
enter only through \cref{thm_spectral}, which supplies these two quantities. The
projector is therefore not intrinsically restricted to dense graphs, and extending
\cref{thm_sampler} amounts to asking which parts of the spectral input and the
sampling analysis survive as the degree decreases.

\emph{Below constant density.} The average degree is $\hatp n$. If the gap keeps
its natural scale $\gamma=\Omega(\hatp n)$, the depth~\cref{eq_tchoice} is of order
$(|E|/\gamma)\log(\sqrt n/(a_{\LM}^2\varepsilon))$, which stays polynomial. The restrictive
term is the sample size: our bound charges the worst-case one-step distortion $b_0$, of
order $1/(\hatp n)$, to every step, which produces the factor
$e^{2b_0t}=(\sqrt n/(a_{\LM}^2\varepsilon))^{O(n/\gamma)}$ in~\cref{eq_Nchoice}, whose
exponent grows as the density decreases. A concentration estimate for the typical
fluctuations of the path weights, in place of worst-case control, could improve this
without changing the projector or the reference walk.

The present analysis already reaches below constant density. The estimates of
\cref{thm_spectral} remain effective for $p\ge C_3/\log n$ with $C_3$ large, and at
$\hatp=\Theta(1/\log n)$ \cref{thm_reduction} gives quasipolynomial running time
$(n/\varepsilon)^{O(\log n)}$. Going substantially further is obstructed by the overlap
rather than the gap. Our bound $|\ip{\LM}{\Psi_0}|\ge e^{-O(1/\hatp)}$ matches the qualitative
prediction of the quadratic model: as $\hatp$ decreases, its ground state contains more excitation pairs, reducing the level-zero component and hence the overlap with the Lieb-Mattis starting state. Even on a strong expander, that state eventually loses
the inverse-polynomial overlap that efficient projection needs.

\emph{Choosing the starting state.} For a coefficient vector $z\ge0$ on $\Sn$, let
$\psi_z\propto\sum_\pi z_\pi\ket\pi$. Using $\psi_z$ on both sides of the forward-walking
ratio~\cref{eq_fw} preserves the positivity the sampling argument needs, and, up to a
factor from the variation of the boundary weights, the reduction then depends on
$|\ip{\psi_z}{\Psi_0}|$ in place of $a_{\LM}$. The fidelity argument behind
\cref{thm_transfer} gives
\[
  |\ip{\psi_z}{\Psi_0}|\ge1-\sqrt{\bigl(\ip{\psi_z}{H_G\psi_z}-\lambda_0\bigr)/\gamma}
\]
so it suffices to find an efficiently samplable positive valence-bond state whose energy
lies within a small fraction of the gap above $\lambda_0$. The trial state of
\cref{prop_wick} has the required energy but no valence-bond expansion suitable
for the estimator. A candidate assigns each matching $\pi$ the nonnegative, factored coefficient
$z_\pi=\prod_{i\in L}\kappa_{i\pi(i)}$. When $\kappa$ is constant this gives the Lieb-Mattis starting state. A natural choice of $\kappa$ adds the Bogoliubov pair matrix $\sum_mt_mu_mv_m^\top$ of the graph to a constant, potentially improving that overlap. This is a resonating-valence-bond trial state of the type used in projector Monte Carlo~\cite{Sandvik2005}. Can its energy be shown to
lie within the gap, and can its boundary weights be evaluated with small variation?

\emph{Toward sparse expanders.} The sparse problem thus splits into three logically
distinct tasks: a singlet-sector gap of order $\hatp n$ (spectral), a positive samplable starting state
within a fraction of the gap in energy (variational), and control of the path-weight
fluctuations over the required depth (probabilistic). In the dense regime all three come
from mean-field structure. At low degree the spectral input in particular must come from
methods that are not perturbations of the complete bipartite model.

\section*{AI disclosure}
\label{sec_AI_disclosure}
\addcontentsline{toc}{section}{AI disclosure}

The authors used AI-based tools to find the original proof strategy. In particular, the connection to the Bogoliubov transformation was made by Claude. AI was also used in improving the presentation. All proofs, calculations, and conclusions suggested by AI were independently verified by all authors.

\clearpage
\printbibliography[heading=bibintoc]

\clearpage
\appendix
\section{Deferred proofs for the estimator}\label{appendix_estimator}

\subsection{The projection error}\label{appendix_fwconv}

\begin{proof}[Proof of \cref{lemma_fwconv}]
  We use the second expression in~\cref{eq_fw}. Fix the phase of $\Psi_0$ so that
  $\ip{\Psi_0}{\LM}=a_{\LM}$, and write $\LM=a_{\LM}\Psi_0+\varphi$ with
  $\varphi\perp\Psi_0$ and $\norm\varphi\le1$. Then
  $\Tc^t\LM=\mu_0^t(a_{\LM}\Psi_0+r)$ with $r\coloneq\mu_0^{-t}\Tc^t\varphi\perp\Psi_0$, and
  $\norm r\le(\mu_1/\mu_0)^t$ by the contraction estimate following~\cref{eq_ratio}.
  Dividing the numerator and the denominator of~\cref{eq_fw} by $\mu_0^{2t}$ leaves
  \[
    a_{\LM}^2\ev{\bS_i\!\cdot\bS_j}+2a_{\LM}\operatorname{Re}\ip{\Psi_0}{O_{ij}r}
    +\ip{r}{O_{ij}r}
    \qquad\text{and}\qquad
    a_{\LM}^2+\norm r^2
  \]
  The operator $O_{ij}$ has eigenvalues $-\tfrac34$ and $\tfrac14$, so
  $\norm{O_{ij}}=\tfrac34$ and $|\ev{\bS_i\!\cdot\bS_j}|\le\tfrac34$. Subtracting
  $\ev{\bS_i\!\cdot\bS_j}$ times the denominator from the numerator therefore leaves at
  most $\tfrac32a_{\LM}(\mu_1/\mu_0)^t+\tfrac32(\mu_1/\mu_0)^{2t}\le3(\mu_1/\mu_0)^t$ in
  absolute value, using $a_{\LM}\le1$ and $\mu_1\le\mu_0$. The denominator is at least
  $a_{\LM}^2$, so~\cref{eq_projectioncondition} gives
  $|\ev{\bS_i\!\cdot\bS_j}_t-\ev{\bS_i\!\cdot\bS_j}|\le3a_{\LM}^{-2}(\mu_1/\mu_0)^t
  \le\varepsilon/(8\sqrt{n+1})\le\varepsilon$.
\end{proof}

\subsection{The relative second moment}\label{appendix_secondmoment}

\begin{proof}[Proof of \cref{lemma_secondmoment}]
  All quantities below are nonnegative, and inequalities between functions on $\Sn$ are
  entrywise. The proof has three independent ingredients: the branching second moment at
  the midpoint, the squared path measure on each leg, and a spectral lower bound on the
  denominator.

  \emph{The midpoint.} Let $g\colon\Sn\to[0,\infty)$ and write
  $\pi'\coloneq\pi\circ(i\ \pi^{-1}(j))$. The branch rule of \cref{alg_sample} gives
  \[
    \E[\varsigma^2g(y_0)\mid x_t=\pi]\le\tfrac9{16}\bigl(g(\pi)+g(\pi')\bigr),
    \qquad
    \E[\iota^2g(y_0)\mid x_t=\pi]\le3g(\pi)
  \]
  where the term $g(\pi')$ is present only when $\pi(i)\ne j$. In that case
  $M_{\pi\pi'}=W_{\pi\pi'}=2\mult(\pi,\pi')\ge2$, because $W$ and $\widehat W$ agree off
  the diagonal and the edge $ij$ itself induces the reconnection. Both conditional
  expectations are therefore at most $3\bigl((I+M)g\bigr)(\pi)$.

  \emph{The two legs.} A leg of length $t$ started uniformly has probability
  $\prod_s\widehat W_{x_sx_{s+1}}/(n!\,Z^t)$, so squaring its weight gives
  $\E[w(x)^2g(x_t)]=\one^\top M^tg/(n!\,Z^t)$, and likewise
  $\E[w(y)^2\mid y_0=\tau]=(M^t\one)_\tau/Z^t$. Conditioning on the first leg and on the
  branch,
  \[
    \max\{\E[X^2],\E[Y^2]\}
    \le
    \frac{3\,\one^\top M^t(I+M)M^t\one}{n!\,Z^{2t}}
    \le
    \frac{3(1+Z)\,\one^\top M^{2t}\one}{n!\,Z^{2t}}
    \le
    3(1+Z)e^{2b_0t}\Bigl(\frac{\varrho_W}{Z}\Bigr)^{4t}
  \]
  where the second step uses $M\one\le\widehat W\one=Z\one$ and the third is the
  path-weight bound of \cref{lemma_uncond_cert}.

  \emph{The denominator.} By \cref{lemma_walk} and $\ip{\Neel}{\pi}=2^{-n/2}$,
  \[
    2^{-n/2}\one^\top W^{2t}\one=\ip{\Neel}{\Tc^{2t}\textstyle\sum_\pi\ket\pi}
  \]
  Since $\sum_\pi\ket\pi=n!\,2^{-n/2}\sqrt{n+1}\,\LM$ by~\cref{eq_LM} and $\Neel$ has
  singlet component $\LM/\sqrt{n+1}$, it follows that
  $\one^\top W^{2t}\one=n!\,\ip{\LM}{\Tc^{2t}\LM}$.
  Since $\Tc$ is positive on $\cS$, it follows from~\cref{eq_EY} and~\cref{eq_rhoW} that
  \[
    \E Y
    =
    \frac{\ip{\LM}{\Tc^{2t}\LM}}{Z^{2t}}
    \ge
    a_{\LM}^2\Bigl(\frac{\varrho_W}{Z}\Bigr)^{2t}
  \]

  Dividing the second-moment bound by the square of this lower bound cancels the factor
  $(\varrho_W/Z)^{4t}$ and gives
  \[
    \max\{\E[X^2],\E[Y^2]\}\le\frac{3(1+Z)}{a_{\LM}^4}e^{2b_0t}(\E Y)^2
  \]
  which is~\cref{eq_relative_second}.
\end{proof}

\section{Deferred proofs for the quadratic model}\label{appendix_boson}

\subsection{Diagonalization of the Bogoliubov Hamiltonian}\label{appendix_bogo}

We diagonalize $\HB$ by a unitary change of basis. In the new basis its eigenvalues are
$\Estar+\sum_m\omega_m(k_m+k'_m)$, where $k_m,k'_m$ are nonnegative integer coordinates.
The constraint defining $\cF_0$ becomes $\sum_m k_m=\sum_m k'_m$: any nonzero admissible
coordinate vector therefore adds at least $2\omega_\star$ to the minimum eigenvalue.

\begin{proof}[Proof of \cref{lemma_bogo}]
\emph{Coordinates.} Complete the left and right singular vectors of $\xi_c$ to
orthonormal bases, write $a_m\coloneq a_{u_m}$, $b_m\coloneq b_{v_m}$, and set $g_m=t_m=0$ and
$\omega_m=\omega$ for $m>r$. In these coordinates $\cF$ has an orthonormal basis
$e_{k,k'}$ indexed by $k,k'\in\mathbb Z_{\ge0}^n$. The operators $a_m,b_m$ lower the
corresponding coordinate, e.g. $a_m e_{k,k'}=\sqrt{k_m}\,e_{k-\mathbf e_m,k'}$ (zero if
$k_m=0$), where $\mathbf e_m$ is the $m$-th standard unit vector; $\dagger$ denotes the
adjoint, which raises that coordinate. Thus $a_m^{\dagger}a_m$ and
$b_m^{\dagger}b_m$ act diagonally with entries $k_m$ and $k'_m$, respectively.
We have
$[a_m,a_{m'}^{\dagger}]=[b_m,b_{m'}^{\dagger}]=\delta_{mm'}I$; lowering operators
commute with each other, as do their adjoints, and every $a$-operator commutes with
every $b$-operator. Consequently
\[
  \HB=\Emf+\sum_{m=1}^nh_m,
  \qquad
  h_m\coloneq\omega\bigl(a_m^{\dagger}a_m+b_m^{\dagger}b_m\bigr)-g_m\bigl(a_m b_m+a_m^{\dagger}b_m^{\dagger}\bigr)
\]
and the summands act on disjoint pairs of coordinates. Polynomial identities below
hold first on finite linear combinations of basis vectors.

\emph{Change of operators.} Since $g_m<\omega$, the parameter
$t_m=g_m/(\omega+\omega_m)$ lies in $[0,1)$. Set
$\theta_m\coloneq\tfrac12\log((1+t_m)/(1-t_m))$, so that $\tanh\theta_m=t_m$, and define
\[
  \tilde a_m\coloneq\cosh\theta_m\,a_m-\sinh\theta_m\,b_m^{\dagger},
  \qquad
  \tilde b_m\coloneq\cosh\theta_m\,b_m-\sinh\theta_m\,a_m^{\dagger}
\]
These satisfy the same commutator identities, since $\cosh^2\theta_m-\sinh^2\theta_m=1$.
Moreover $\cosh2\theta_m=\omega/\omega_m$ and $\sinh2\theta_m=g_m/\omega_m$, so expansion gives
\[
  \tilde a_m^{\dagger}\tilde a_m+\tilde b_m^{\dagger}\tilde b_m
  =\frac{h_m}{\omega_m}+\Bigl(\frac{\omega}{\omega_m}-1\Bigr)I
\]
Summing yields the diagonal form (scalar constants multiply the identity):
\begin{equation}\label{eq_bogo_diag}
  \HB=\Estar+\sum_{m=1}^n\omega_m\bigl(\tilde a_m^{\dagger}\tilde a_m+\tilde b_m^{\dagger}\tilde b_m\bigr)
\end{equation}
Subtracting instead of adding gives
$\tilde a_m^{\dagger}\tilde a_m-\tilde b_m^{\dagger}\tilde b_m=a_m^{\dagger}a_m-b_m^{\dagger}b_m$,
so the change preserves the difference between the two total coordinate sums.

\emph{An eigenbasis.} To justify reading eigenvalues from~\cref{eq_bogo_diag}, we need
a unitary implementing the change. Let
$K\coloneq\sum_m\theta_m(a_m^{\dagger}b_m^{\dagger}-a_m b_m)$.
Then $e^{K}$ is a well-defined unitary on $\cF$~\cite{NamNapiorkowskiSolovej2016}, and the identities $[K,a_m]=-\theta_m b_m^{\dagger}$ and
$[K,b_m^{\dagger}]=-\theta_m a_m$ give
$e^{K}a_me^{-K}=\cosh\theta_m\,a_m-\sinh\theta_m\,b_m^{\dagger}=\tilde a_m$ and similarly
$e^{K}b_me^{-K}=\tilde b_m$ (solve the corresponding two-variable linear differential
equation). Hence the vectors $e^{K}e_{k,k'}$ form a complete orthonormal eigenbasis of $\HB$, with
eigenvalues $\Estar+\sum_m\omega_m(k_m+k'_m)$.
Also $e^{K}$ preserves $\cN_L-\cN_R$, so the vectors with
$\sum_m k_m=\sum_m k'_m$ form an eigenbasis of $\cF_0$.

\emph{The minimum eigenvector and gap.} Expanding~\cref{eq_squeezed} in the original
coordinate basis gives the explicit vector
\[
  \Omega=\sum_{k\in\mathbb Z_{\ge0}^n}\Bigl(\prod_{m=1}^nt_m^{k_m}\Bigr)e_{k,k},
  \qquad
  \norm\Omega^2=\prod_{m=1}^n(1-t_m^2)^{-1}<\infty
\]
with the convention $t_m^0=1$. Thus $\Omega\in\cF_0$; geometric decay also gives finite
moments of all coordinate sums. Comparing coefficients yields
$a_m\Omega=t_m b_m^{\dagger}\Omega$ and $b_m\Omega=t_m a_m^{\dagger}\Omega$, hence
$\tilde a_m\Omega=\tilde b_m\Omega=0$. It follows that $\Omega$ is proportional to
$e^{K}e_{0,0}$ and $\HB\Omega=\Estar\Omega$.
Since all $\omega_m\ge\omega_\star>0$, this is the unique minimum eigenvector up to
scaling. Every other eigenvector in $\cF_0$ has
$\sum_m k_m=\sum_m k'_m\ge1$, so its eigenvalue is at least $\Estar+2\omega_\star$.
This proves~\cref{eq_bogo_inputs}.
\end{proof}

\subsection{Low-level comparison with the Bogoliubov Hamiltonian}\label{appendix_low}

We use the harmonic set-function spaces of \cref{sec:rep-theory}:
$\cM_k=\cH_k^L\otimes\cH_k^R$, with level projections $P_k$ and down operators
$\partial_u^L$, $\partial_v^R$. Under the embedding of \cref{lemma_embed}, level $k$
has $k$ occupied sites on each side. We first recover $\Hcov=\Hmf+\Hpair$ by
compressing $\widetilde H$ to $\cM$, then bound its difference from $\HB$ below level
$4\kcut$. Here compressing an operator to a subspace means multiplying it on both sides by the
orthogonal projection onto that subspace and restricting it there.

\begin{proof}[Proof of \cref{lemma_CB}]
\emph{(i) Exact compression.} On $\cM_k$, the diagonal part of $\widetilde H$ is
$[\Emf+\hatp\,d(k)]I=[\Emf+\hatp\,k(n-k+1)]I$, exactly the mean-field block
of~\cref{eq:Hmf-final}. For the block from level $k$ to $k-1$,
\cref{lemma_embed_item3} of \cref{lemma_embed} gives
\[
  P_{k-1}\widetilde HP_k
  =-\Gamma_k\sum_m\sigma_m\,
    P_{k-1}(\partial_{u_m}^L\otimes\partial_{v_m}^R)P_k
\]
The coefficient $\Gamma_k$ is evaluated at the source level, before the two down
operators act. The reverse block is its adjoint, with the projections retained
because the up operators need not preserve the harmonic spaces. These are precisely
the pair blocks of~\cref{eq:Hpair-explicit}; all other blocks vanish. Thus
$\Pi_\cM\widetilde H\Pi_\cM=\Hmf+\Hpair=\Hcov$ on $\cM$.
Both pair directions preserve $\cN_L-\cN_R$, so $\widetilde H$ acts on $\cF_0\supseteq\cM$.
The min-max principle (restricting the admissible subspaces cannot lower an
eigenvalue) gives $\lambda_j(\Hcov)\ge\lambda_j(\widetilde H|_{\cF_0})$.

\emph{(ii) Error on low levels.} Let $\widetilde P_k$ project onto
$\{\psi\in\cF_0:\cN_L\psi=k\psi\}$, so $\widetilde P_k|_\cM=P_k$, and let
$\widetilde P\coloneq\sum_{0\le k<4\kcut}\widetilde P_k$. Unlike $P_k$, $\widetilde P_k$ allows multiple occupations and
does not impose the harmonic constraints. Since $4\kcut\le n/4$,~\cref{eq_HB_error} gives
\[
  \widetilde P(\widetilde H-\HB)\widetilde P=-\hatp\,\cN_L(\cN_L-I)\widetilde P-\widetilde P(\cA_f+\cA_f^{\dagger})\widetilde P
\]
The first term is diagonal with entries $-\hatp\,k(k-1)$ for $0\le k<4\kcut$, hence has
norm at most $16\kcut^2$, using $\hatp\le1$.

For the second term, the block $\widetilde P_{k-1}\cA_f\widetilde P_k$ uses $f(k)=\Gamma_k-\tfrac12$ at its source
level $k$. For $k<4\kcut$, $n-2k+1>n/2$, so $\sqrt{1+x}\le1+x/2$ gives
\[
0\ \le\ \Gamma_k-\frac12
=\frac12\Bigl(\sqrt{1+\frac2{n-2k+1}}-1\Bigr)
\ \le\ \frac1{2(n-2k+1)}
\ \le\ \frac1n
\]
Setting $f(k)=0$ for $k\ge4\kcut$ leaves $\widetilde P\cA_f\widetilde P$ unchanged and makes $|f|\le1/n$ globally.
\cref{lemma_band} on $\cF_0$ therefore gives, for $1\le k<4\kcut$,
\begin{align*}
\norm{\widetilde P_{k-1}\cA_f\widetilde P_k}_\opn
&\le\frac{\sigma_1}n\sqrt{(k-1+r)k}
\le\frac{\sigma_1}n\sqrt{(4\kcut+r)\,4\kcut}\\
&\le\frac{\sigma_1}n\sqrt{\tfrac54n\cdot4\kcut}
=\sqrt5\,\sigma_1\sqrt{\kcut/n}
\end{align*}
using $r\le n$ and $4\kcut\le n/4$. Since
$\widetilde P\cA_f\widetilde P=\sum_{1\le k<4\kcut}\widetilde P_{k-1}\cA_f\widetilde P_k$ has mutually orthogonal block domains and
ranges, its norm is the maximum block norm. Thus
\[
\norm{\widetilde P(\cA_f+\cA_f^{\dagger})\widetilde P}_\opn\le2\norm{\widetilde P\cA_f\widetilde P}_\opn
\le2\sqrt5\,\sigma_1\sqrt{\kcut/n}\le5\sigma_1\sqrt{\kcut/n}
\]
Adding the bounds yields $\norm{\widetilde P(\widetilde H-\HB)\widetilde P}_\opn\le\etaB(\kcut)$, and hence
$|\ip{\psi}{(\widetilde H-\HB)\psi}|\le\etaB(\kcut)\norm\psi^2$ for $\psi\in\cF_{<4\kcut}\coloneq\operatorname{ran}\widetilde P$.

\emph{Eigenvalue comparison.} Let $\widetilde H|_{<4\kcut}$ and $\HB|_{<4\kcut}$ be the
compressions of $\widetilde H$ and $\HB$ to the finite-dimensional space $\cF_{<4\kcut}$,
which contains $\bigoplus_{k<4\kcut}\cM_k$. By (i), $\Hcov|_{<4\kcut}$ is the compression of
$\widetilde H|_{<4\kcut}$ to $\bigoplus_{k<4\kcut}\cM_k$.
The min-max principle, the norm bound above (Weyl's inequality), and the min-max
principle again give, for $j=0,1$,
\begin{align*}
\lambda_j(\Hcov|_{<4\kcut})
&\ge\lambda_j(\widetilde H|_{<4\kcut})
\ge\lambda_j(\HB|_{<4\kcut})-\etaB(\kcut)\\
&\ge\lambda_j(\HB|_{\cF_0})-\etaB(\kcut)
\end{align*}
\cref{lemma_bogo}, which applies since $\sigma_1<\hatp n$, completes the proof.
\end{proof}

\subsection{The trial state}\label{appendix_wick}

The proof of \cref{prop_wick} has three parts. \cref{lemma_wick_trunc} shows
that the truncated squeezed state $\Omega_{\le\ell}$ is an almost exact eigenvector of $\HB$,
\cref{lemma_wick_defects} bounds the two projection losses, and the proof of the
proposition combines them.

\emph{Parameters.} Throughout this subsection assume the hypotheses of
\cref{prop_wick}, let $\ell=\lfloor n^{1/4}\rfloor+1$ as in
\cref{section_wick}, and let $B\coloneq4/\hatp+1$. Since $n^{1/4}\ge100/\hatp\ge100$, the
following hold.
\begin{enumerate}[label=\textup{(\alph*)},nosep]
  \item $\ell\ge100$, $(\ell-1)^{4}\le n\le\ell^{4}$, $\ell\le2n^{1/4}$, $16\ell\le n$,
    $4\ell+16\le n$ and $n-2\ell+2\ge n/2$.
  \item $\norm t^{2}\le1/\hatp$ and
    $t_{\max}^{2}\ell^{2}\le4\sigma_1^{2}/(\hatp^{2}n^{3/2})\le\hatp/4$
    by~\cref{eq_t_bounds}; in particular $B\ge2\norm t^{2}+t_{\max}^{2}\ell^{2}$.
  \item $eB\le14/\hatp\le\ell/e$, so $B^{\ell}/\ell!\le(eB/\ell)^{\ell}\le e^{-\ell}$.
  \item $\norm g^{2}\coloneq\sum_m g_m^{2}=\norm{\xi_c}_F^{2}/4\le\hatp(1-\hatp)n^{2}/4$
    by~\cref{eq_frob}, and $0\le\Emf-\Estar=\sum_m g_m^{2}/(\omega+\omega_m)\le\norm g^{2}/\omega\le n/2$.
\end{enumerate}
The two projection losses will be bounded by
\[
  \epsilon_{\hc}\coloneq t_{\max}^{2}\norm t^{2}\ell^{2}\le\frac{4\sigma_1^{2}}{\hatp^{3}n^{3/2}},
  \qquad
  \epsilon_{\mathrm h}\coloneq\frac{4t_{\max}^{2}\norm t^{2}\ell^{3}}{n-2\ell+2}\le\frac{64\sigma_1^{2}}{\hatp^{3}n^{9/4}}
\]
By~(a) and~(b), $\epsilon_{\hc}\le1/4$ and
$\epsilon_{\mathrm h}=4\ell\epsilon_{\hc}/(n-2\ell+2)\le\epsilon_{\hc}$, hence
$\epsilon_{\hc}+\epsilon_{\mathrm h}\le1/2$.

\emph{The pairing operator.} Let $\cA\coloneq\sum_m g_m a_{u_m}b_{v_m}$, so that
$\HB=\Emf+\omega(\cN_L+\cN_R)-\cA-\cA^{\dagger}$. Cauchy-Schwarz in $m$, Bessel's
inequality on each side, and the canonical commutation relations give
\begin{align*}
  \norm{\cA\psi}^{2}
  &\le\norm g^{2}\sum_m\norm{a_{u_m}b_{v_m}\psi}^{2}
   \le\norm g^{2}\ip{\psi}{\cN_L\cN_R\psi}\\
  \norm{\cA^{\dagger}\psi}^{2}
  &=\norm{\cA\psi}^{2}+\sum_m g_m^{2}\ip{\psi}{(a_{u_m}^{\dagger}a_{u_m}+b_{v_m}^{\dagger}b_{v_m}+I)\psi}
   \le\norm{\cA\psi}^{2}+\norm g^{2}\ip{\psi}{(\cN_L+\cN_R+I)\psi}
\end{align*}
Hence, for every $\psi\in\cF_0$ supported on levels at most $\ell$,
\begin{equation}\label{eq_wick_ops}
  \norm{\cA\psi}\le\norm g\,\ell\norm\psi,
  \qquad
  \norm{\cA^{\dagger}\psi}\le\norm g\,(\ell+2)\norm\psi
\end{equation}

\emph{Ladder identities.} Let $\nu_k\coloneq\norm{\Omega_k}^{2}$, so that
$\norm{\Omega_{\le\ell}}^{2}=\sum_{k\le\ell}\nu_k$. Since $\Omega_{k+1}=\cP\Omega_k/(k+1)$, the
commutators
\begin{equation}\label{eq_wick_ladder}
  [a_{u_m},\cP]=t_m b_{v_m}^{\dagger},
  \qquad
  [b_{v_m},\cP]=t_m a_{u_m}^{\dagger}
\end{equation}
give $a_{u_m}\Omega_{k+1}=t_m b_{v_m}^{\dagger}\Omega_k$ and
$b_{v_m}\Omega_{k+1}=t_m a_{u_m}^{\dagger}\Omega_k$. For a unit vector $u$ and a vector
$\psi$ with $k$ bosons on $L$, Bessel's inequality gives
$\norm{a_u^{\dagger}\psi}^{2}=\norm\psi^{2}+\norm{a_u\psi}^{2}\le(k+1)\norm\psi^{2}$, and
likewise on $R$. Therefore
\begin{equation}\label{eq_wick_b}
  \norm{a_{u_m}\Omega_k}^{2},\ \norm{b_{v_m}\Omega_k}^{2}\ \le\ t_m^{2}\,k\,\nu_{k-1}
  \qquad(k\ge1)
\end{equation}

\begin{lemma}[Truncation estimates]\label{lemma_wick_trunc}
  \leavevmode
  \begin{enumerate}[label=\textup{(\roman*)},nosep]
    \item\label{lemma_wick_trunc_item1} $k!\,\norm{\Omega_k}^{2}\le B^{k}$ for $k\le\ell+1$, and
      $\sum_{k\le\ell}k\,\norm{\Omega_k}^{2}\le\norm t^{2}(1+t_{\max}^{2}\ell)\norm{\Omega_{\le\ell}}^{2}$.
    \item\label{lemma_wick_trunc_item2} $\HB\,\Omega_{\le\ell}=\Estar\Omega_{\le\ell}+\cA\,\Omega_{\ell+1}-\cA^{\dagger}\Omega_\ell$.
    \item\label{lemma_wick_trunc_item3} $\ip{\Omega_{\le\ell}}{(\HB-\Estar)\Omega_{\le\ell}}\le1$.
    \item\label{lemma_wick_trunc_item4} $\norm{\cA\,\Omega_{\ell+1}-\cA^{\dagger}\Omega_\ell}^{2}\le1$.
  \end{enumerate}
\end{lemma}

\begin{proof}
\emph{Part (i).} Pairing $(k+1)\Omega_{k+1}=\cP\Omega_k$ with the first ladder identity,
and then using~\cref{eq_wick_b},
\begin{equation}\label{eq_wick_recursion}
\begin{aligned}
  (k+1)\,\nu_{k+1}
  &=\sum_m t_m\ip{a_{u_m}\Omega_{k+1}}{b_{v_m}^{\dagger}\Omega_k}
   =\norm t^{2}\nu_k+\sum_m t_m^{2}\norm{b_{v_m}\Omega_k}^{2}\\
  &\le\norm t^{2}\nu_k+t_{\max}^{2}\norm t^{2}\,k\,\nu_{k-1}.
\end{aligned}
\end{equation}
By~(b), $B\ge2\norm t^{2}+t_{\max}^{2}\ell^{2}$, so
$B^{2}\ge\norm t^{2}B+\norm t^{2}t_{\max}^{2}\ell^{2}$. Since $\nu_0=1$, induction
in~\cref{eq_wick_recursion} gives $\nu_k\le B^k/k!$ for $k\le\ell+1$. For the second
assertion, sum~\cref{eq_wick_recursion} over $0\le k<\ell$:
\[
  \sum_{k=0}^{\ell}k\,\nu_k
  \le\norm t^{2}\sum_{k=0}^{\ell-1}\nu_k+t_{\max}^{2}\norm t^{2}\sum_{k=1}^{\ell-1}k\,\nu_{k-1}
  \le\norm t^{2}(1+t_{\max}^{2}\ell)\norm{\Omega_{\le\ell}}^{2}
\]

\emph{Part (ii).} The level-$k$ component of the eigenvalue equation $\HB\Omega=\Estar\Omega$
of \cref{lemma_bogo} reads
$(\Emf+2\omega k)\Omega_k-\cA\Omega_{k+1}-\cA^{\dagger}\Omega_{k-1}=\Estar\Omega_k$.
Summing over $k\le\ell$ leaves only the two boundary terms.

\emph{Part (iii).} The boundary term $\cA^{\dagger}\Omega_\ell$ lies on level $\ell+1$, so
by~(ii) the form equals $\ip{\Omega_\ell}{\cA\Omega_{\ell+1}}$. Since $a_{u_m}$ and
$b_{v_m}$ commute, the first ladder identity gives
\begin{equation}\label{eq_wick_boundary}
  \cA\Omega_{\ell+1}=\sum_m g_m t_m\bigl(I+b_{v_m}^{\dagger}b_{v_m}\bigr)\Omega_\ell
\end{equation}
With $\sum_m g_m t_m=\sum_m(\omega-\omega_m)=\Emf-\Estar$, then~\cref{eq_wick_b}, part~(i)
with $\nu_{\ell-1}\le(\ell/B)B^{\ell}/\ell!$, and $t_{\max}^{2}\ell^{2}\le B$ from~(b),
\begin{align*}
  \ip{\Omega_\ell}{\cA\Omega_{\ell+1}}
  &=\sum_m g_m t_m\bigl(\nu_\ell+\norm{b_{v_m}\Omega_\ell}^{2}\bigr)
   \le(\Emf-\Estar)\bigl(\nu_\ell+t_{\max}^{2}\ell\,\nu_{\ell-1}\bigr)\\
  &\le2(\Emf-\Estar)\frac{B^{\ell}}{\ell!}
   \le n\,e^{-\ell}
   \le\ell^{4}e^{-\ell}
   \le1
\end{align*}
using~(c) and~(d) in the third step and~(a) in the last two.

\emph{Part (iv).} The vector $b_{v_m}\Omega_\ell$ has $\ell-1$ bosons on $R$, so
$\norm{b_{v_m}^{\dagger}b_{v_m}\Omega_\ell}\le\sqrt\ell\,\norm{b_{v_m}\Omega_\ell}\le t_m\ell\sqrt{\nu_{\ell-1}}$
by~\cref{eq_wick_b}. Hence~\cref{eq_wick_boundary} and~\cref{eq_wick_ops} give
\[
  \norm{\cA\Omega_{\ell+1}}\le(\Emf-\Estar)\bigl(\sqrt{\nu_\ell}+t_{\max}\ell\sqrt{\nu_{\ell-1}}\bigr),
  \qquad
  \norm{\cA^{\dagger}\Omega_\ell}\le\norm g(\ell+2)\sqrt{\nu_\ell}
\]
Using $(x+y)^{2}\le2x^{2}+2y^{2}$ twice, then part~(i) and~(b) as in part~(iii), and finally
(a), (c) and (d),
\begin{align*}
  \norm{\cA\,\Omega_{\ell+1}-\cA^{\dagger}\Omega_\ell}^{2}
  &\le4(\Emf-\Estar)^{2}\bigl(\nu_\ell+t_{\max}^{2}\ell^{2}\nu_{\ell-1}\bigr)+2\norm g^{2}(\ell+2)^{2}\nu_\ell\\
  &\le\Bigl(4(\Emf-\Estar)^{2}(1+\ell)+2\norm g^{2}(\ell+2)^{2}\Bigr)\frac{B^{\ell}}{\ell!}
   \le n^{2}(\ell+2)^{2}e^{-\ell}
   \le\ell^{11}e^{-\ell}
   \le1
   \qedhere
\end{align*}
\end{proof}

\begin{lemma}[Projection defects]\label{lemma_wick_defects}
  \leavevmode
  \begin{enumerate}[label=\textup{(\roman*)},nosep]
    \item $\norm{\Omega_{\le\ell}-\Omega_{\le\ell}|_{\hc}}^{2}\le\epsilon_{\hc}\norm{\Omega_{\le\ell}}^{2}$;
    \item $\norm{\Omega_{\le\ell}|_{\hc}-\Upsilon}^{2}\le\epsilon_{\mathrm h}\norm{\Omega_{\le\ell}}^{2}$.
  \end{enumerate}
\end{lemma}

\begin{proof}
Both parts
reduce to the double annihilations $a_i^{2}\Omega_k$ and $b_j^{2}\Omega_k$, which we bound
first. Let $t(i)\coloneq\sum_m t_m u_m(i)v_m\in\mathbb R^{R}$ be the $i$-th row of the pair matrix
$\sum_m t_m u_m v_m^{\top}$, so that $\norm{t(i)}^{2}\le t_{\max}^{2}$ and
$\sum_i\norm{t(i)}^{2}=\norm t^{2}$. As in~\cref{eq_wick_ladder},
$a_i\Omega_{k+1}=b^{\dagger}_{t(i)}\Omega_k$, hence
$a_i^{2}\Omega_k=(b^{\dagger}_{t(i)})^{2}\Omega_{k-2}$, and on the sector with $k-2$ bosons on
$R$ the operator $(b^{\dagger}_{t(i)})^{2}$ has norm at most $\norm{t(i)}^{2}\sqrt{k(k-1)}$.
Therefore
\begin{equation}\label{eq_wick_double}
  \sum_{i\in L}\norm{a_i^{2}\Omega_k}^{2}
  \le t_{\max}^{2}\norm t^{2}\,k(k-1)\,\nu_{k-2}
\end{equation}
and likewise for $\sum_{j\in R}\norm{b_j^{2}\Omega_k}^{2}$.

\emph{Part (i).} The number of multiply occupied sites is at most half the sum of the second
factorial occupation numbers. By orthogonality of the levels,~\cref{eq_wick_double}, and
$k(k-1)\le\ell^{2}$,
\[
  \norm{\Omega_{\le\ell}-\Omega_{\le\ell}|_{\hc}}^{2}
  \le\frac12\sum_{k=0}^{\ell}\Bigl(
     \sum_{i\in L}\norm{a_i^{2}\Omega_k}^{2}
     +\sum_{j\in R}\norm{b_j^{2}\Omega_k}^{2}\Bigr)
  \le t_{\max}^{2}\norm t^{2}\ell^{2}\sum_{2\le k\le\ell}\nu_{k-2}
  \le\epsilon_{\hc}\norm{\Omega_{\le\ell}}^{2}
\]

\emph{Part (ii).} Since $\Upsilon=\sum_{k\le\ell}P_k(\Omega_k|_{\hc})$,
$\norm{\Omega_{\le\ell}|_{\hc}-\Upsilon}^{2}=\sum_{k\le\ell}\norm{\Omega_k|_{\hc}-P_k(\Omega_k|_{\hc})}^{2}$. We
bound each term by a spectral gap and then bound the gap's input.

\emph{The spectral gap.} Recall $a_{\one}=\sum_{i\in L}a_i$ and $b_{\one}=\sum_{j\in R}b_j$. By
\cref{lemma_embed}, on hard-core vectors $a_{\one}$ acts as the down operator $\partial^L$,
the pullback of $S_L^{+}$, whose kernel on level $k$ is $\cH_k^L$. On the orthogonal
complement of $\cH_k^L$ in level $k$, every vector lies in spins $s\ge s_k+1$ at weight
$s_k$, where $(\partial^L)^{\dagger}\partial^L$, the pullback of
$S_L^-S_L^+=\bS_L^{2}-(S_L^z)^{2}-S_L^z$, is at least
$(s_k+1)(s_k+2)-s_k(s_k+1)=n-2k+2$. The same holds on $R$. On level $k$
we have $P_k=\Pi_{\cH^L}\otimes\Pi_{\cH^R}$ and
$I-P_k\preccurlyeq(I-\Pi_{\cH^L})\otimes I+I\otimes(I-\Pi_{\cH^R})$, so for hard-core $\psi$ of level $k\le n/2$
\begin{equation}\label{eq_wick_gap}
  (n-2k+2)\norm{\psi-P_k\psi}^{2}\le\norm{a_{\one}\psi}^{2}+\norm{b_{\one}\psi}^{2}
\end{equation}

\emph{The input.} Because $u_m,v_m\perp\one$, both $a_{\one}$ and $b_{\one}$ commute with $\cP$ and
annihilate the vacuum, hence annihilate $\Omega_k$. Since $a_{\one}$ preserves hard-core
vectors, $a_{\one}(\Omega_k|_{\hc})=-\bigl(a_{\one}(\Omega_k-\Omega_k|_{\hc})\bigr)|_{\hc}$. A configuration that is not hard-core
becomes hard-core after one annihilation on $L$ only if it has exactly one doubly occupied
site, on $L$, and no other multiple occupancy; each hard-core output has at most $k$ such
predecessors. Cauchy-Schwarz over them and~\cref{eq_wick_double} give
\[
  \norm{a_{\one}(\Omega_k|_{\hc})}^{2}
  \le k\sum_{i\in L}\norm{a_i^{2}\Omega_k}^{2}
  \le t_{\max}^{2}\norm t^{2}\,k^{3}\nu_{k-2}
\]
and the same bound for $b_{\one}(\Omega_k|_{\hc})$.

\emph{Summation.} By~\cref{eq_wick_gap} with $n-2k+2\ge n-2\ell+2$ and $k^{3}\le\ell^{3}$,
\[
  \sum_{k\le\ell}\norm{\Omega_k|_{\hc}-P_k(\Omega_k|_{\hc})}^{2}
  \le\frac{2t_{\max}^{2}\norm t^{2}\ell^{3}}{n-2\ell+2}\sum_{2\le k\le\ell}\nu_{k-2}
  \le\epsilon_{\mathrm h}\norm{\Omega_{\le\ell}}^{2}
  \qedhere
\]
\end{proof}

\begin{proof}[Proof of \cref{prop_wick}]
Let $\zeta_{\hc}\coloneq\Omega_{\le\ell}-\Omega_{\le\ell}|_{\hc}$ and $\zeta_{\mathrm h}\coloneq\Omega_{\le\ell}|_{\hc}-\Upsilon$, so
that $\Upsilon=\Omega_{\le\ell}-\zeta_{\hc}-\zeta_{\mathrm h}$.

\emph{Step 1: the free model.} By \cref{lemma_CB_item1} of \cref{lemma_CB},
$\ip{\Upsilon}{\Hcov\Upsilon}=\ip{\Upsilon}{\widetilde H\Upsilon}$. Since $\Upsilon$ is
supported on levels $\le\ell<4\ell$, $16\ell\le n$ by~(a), and $\sigma_1<\hatp n$,
\cref{lemma_CB_item2} of \cref{lemma_CB} with cutoff $\ell$ gives
\begin{equation}\label{eq_wick_free}
  \ip{\Upsilon}{(\Hcov-\Estar)\Upsilon}\le\ip{\Upsilon}{(\HB-\Estar)\Upsilon}+\etaB(\ell)\norm\Upsilon^{2},
  \qquad
  \etaB(\ell)=16\ell^{2}+5\sigma_1\sqrt{\ell/n}\le64\sqrt n+2n^{3/8}\le66\sqrt n
\end{equation}

\emph{Step 2: the norm of $\Upsilon$.} The defects are the losses of two nested orthogonal
projections, so by \cref{lemma_wick_defects}
\[
  \norm\Upsilon^{2}=\norm{\Omega_{\le\ell}}^{2}-\norm{\zeta_{\hc}}^{2}-\norm{\zeta_{\mathrm h}}^{2}
  \ge(1-\epsilon_{\hc}-\epsilon_{\mathrm h})\norm{\Omega_{\le\ell}}^{2}\ge\tfrac12\norm{\Omega_{\le\ell}}^{2}
\]
Since $\Upsilon$ has vacuum component $1$, also $\norm\Upsilon^{2}\ge1$.

\emph{Step 3: the energy of $\Upsilon$.} We expand $\ip{\Upsilon}{(\HB-\Estar)\Upsilon}$ around $\Omega_{\le\ell}$,
whose energy is controlled by \cref{lemma_wick_trunc}: by its parts~(ii) to~(iv),
$\ip{\Omega_{\le\ell}}{(\HB-\Estar)\Omega_{\le\ell}}\le1$ and $\norm{(\HB-\Estar)\Omega_{\le\ell}}^{2}\le1$. The defects are bounded through the crude bounds
\begin{equation}\label{eq_wick_crude}
  |\ip{\psi}{(\HB-\Estar)\psi}|\le3n\ell\norm\psi^{2},
  \qquad
  \norm{(\HB-\Estar)\psi}^{2}\le9n^{2}\ell^{2}\norm\psi^{2}
\end{equation}
valid for $\psi\in\cF_0$ supported on levels at most $\ell$, by~\cref{eq_wick_ops},
$\Emf-\Estar\le n/2$, $\omega\le n/2$ and $\norm g\le n/4$. Since $\Omega_{\le\ell}|_{\hc}=\Omega_{\le\ell}-\zeta_{\hc}$ and
$\Upsilon=\Omega_{\le\ell}|_{\hc}-\zeta_{\mathrm h}$,
\begin{align*}
  \ip{\Upsilon}{(\HB-\Estar)\Upsilon}
  &=\ip{\Omega_{\le\ell}}{(\HB-\Estar)\Omega_{\le\ell}}-2\operatorname{Re}\ip{\zeta_{\hc}}{(\HB-\Estar)\Omega_{\le\ell}}+\ip{\zeta_{\hc}}{(\HB-\Estar)\zeta_{\hc}}
    -2\operatorname{Re}\ip{\zeta_{\mathrm h}}{(\HB-\Estar)(\Omega_{\le\ell}|_{\hc})}+\ip{\zeta_{\mathrm h}}{(\HB-\Estar)\zeta_{\mathrm h}}
\end{align*}
The five terms are bounded as follows, using~\cref{eq_wick_crude} and, in the fourth line,
$\norm{(\HB-\Estar)(\Omega_{\le\ell}|_{\hc})}^{2}\le2\norm{(\HB-\Estar)\Omega_{\le\ell}}^{2}+2\norm{(\HB-\Estar)\zeta_{\hc}}^{2}$:
\begin{align*}
  \ip{\Omega_{\le\ell}}{(\HB-\Estar)\Omega_{\le\ell}}&\le1\\
  2|\ip{\zeta_{\hc}}{(\HB-\Estar)\Omega_{\le\ell}}|&\le\norm{\zeta_{\hc}}^{2}+\norm{(\HB-\Estar)\Omega_{\le\ell}}^{2}\le\norm{\zeta_{\hc}}^{2}+1\\
  \ip{\zeta_{\hc}}{(\HB-\Estar)\zeta_{\hc}}&\le3n\ell\norm{\zeta_{\hc}}^{2}\\
  2|\ip{\zeta_{\mathrm h}}{(\HB-\Estar)(\Omega_{\le\ell}|_{\hc})}|&\le n^{2}\norm{\zeta_{\mathrm h}}^{2}+n^{-2}\norm{(\HB-\Estar)(\Omega_{\le\ell}|_{\hc})}^{2}
    \le n^{2}\norm{\zeta_{\mathrm h}}^{2}+2n^{-2}+18\ell^{2}\norm{\zeta_{\hc}}^{2}\\
  \ip{\zeta_{\mathrm h}}{(\HB-\Estar)\zeta_{\mathrm h}}&\le3n\ell\norm{\zeta_{\mathrm h}}^{2}
\end{align*}
Adding,
\begin{equation}\label{eq_wick_energy}
  \ip{\Upsilon}{(\HB-\Estar)\Upsilon}
  \le2+2n^{-2}+(1+3n\ell+18\ell^{2})\norm{\zeta_{\hc}}^{2}+(n^{2}+3n\ell)\norm{\zeta_{\mathrm h}}^{2}
\end{equation}

\emph{Step 4: normalization, and part~(i).} By \cref{lemma_wick_defects} and Step~2,
$\norm{\zeta_{\hc}}^{2}\le\epsilon_{\hc}\norm{\Omega_{\le\ell}}^{2}\le2\epsilon_{\hc}\norm\Upsilon^{2}$, and likewise for
$\zeta_{\mathrm h}$. By~(a), $1+3n\ell+18\ell^{2}\le4n\ell\le8n^{5/4}$ and $n^{2}+3n\ell\le2n^{2}$.
Dividing~\cref{eq_wick_free} and~\cref{eq_wick_energy} by $\norm\Upsilon^{2}\ge1$ and
inserting the bounds on $\epsilon_{\hc}$ and $\epsilon_{\mathrm h}$,
\begin{align*}
  \ip{\Phi}{\Hcov\Phi}-\Estar
  &\le66\sqrt n+2+2n^{-2}+2(1+3n\ell+18\ell^{2})\epsilon_{\hc}+2(n^{2}+3n\ell)\epsilon_{\mathrm h}\\
  &\le66\sqrt n+3+16n^{5/4}\epsilon_{\hc}+4n^{2}\epsilon_{\mathrm h}
   \le70\sqrt n+\frac{64\sigma_1^{2}}{\hatp^{3}n^{1/4}}+\frac{256\sigma_1^{2}}{\hatp^{3}n^{1/4}}
\end{align*}
which is part~(i), since $64+256=320$.

\emph{Part (ii).} By \cref{lemma_wick_trunc_item1} of \cref{lemma_wick_trunc},
$\norm\Upsilon^{2}\le\norm{\Omega_{\le\ell}}^{2}=\sum_{k\le\ell}\nu_k\le\sum_kB^{k}/k!\le e^{B}$. Since
$\Upsilon(\emptyset,\emptyset)=1$, this gives
$\Phi(\emptyset,\emptyset)\ge e^{-B/2}=e^{-(2/\hatp+1/2)}$.

\emph{Part (iii).} By level orthogonality, \cref{lemma_wick_trunc_item1} of \cref{lemma_wick_trunc}, (b) and Step~2,
\[
  \ip{\Upsilon}{\Pi^L\Upsilon}
  =\sum_kk\,\norm{P_k(\Omega_k|_{\hc})}^{2}
  \le\sum_kk\,\nu_k
  \le\frac1{\hatp}\Bigl(1+\frac{\hatp}4\Bigr)\norm{\Omega_{\le\ell}}^{2}
  \le\frac2{\hatp}\Bigl(1+\frac{\hatp}4\Bigr)\norm\Upsilon^{2}
\]
so $\ip{\Phi}{\Pi^L\Phi}\le2/\hatp+1/2\le3/\hatp$.
\end{proof}

\end{document}